\documentclass{aims}
\pdfoutput=1
\usepackage[utf8]{inputenc}
\usepackage[T1]{fontenc}
\usepackage[shortlabels]{enumitem}
\usepackage{chngcntr}
\usepackage{txfonts}
\def\typeofarticle{Research article}
\def\currentvolume{5}
\def\currentissue{2}
\def\currentyear{2023}
\def\ppages{1--59.}
\def\DOI{10.3934/mine.2023036}
\def\Received{07 April 2022}
\def\Revised{23 May 2022}
\def\Accepted{23 May 2022}
\def\Published{02 June 2022}
\usepackage{cite}
\newcommand{\doiurl}[1]{\url{https://doi.org/#1}}

\let\arXiv\arxiv

\usepackage{amsmath,amssymb,amsthm,bm}
\numberwithin{equation}{section}
\usepackage[nointegrals]{wasysym}
\renewcommand{\bigcirc}{\ocircle}

\newcommand{\To}{\Rightarrow}
\newcommand{\N}{\mathbb{N}}
\newcommand{\Z}{\mathbb{Z}}
\newcommand{\R}{\mathbb{R}}

\newcommand{\simto}{\xrightarrow{\raisebox{-0.7ex}[0ex][0ex]{$\sim$}}}

\newcommand{\norm}[1]{\lVert#1\rVert}

\newcommand{\Open}{\mathcal{O}}
\DeclareMathOperator{\res}{res}
\DeclareMathOperator{\Orient}{Or}
\newcommand{\VectField}{\mathfrak{X}}
\newcommand{\Lap}{\Delta}
\usepackage{accents}
\newcommand*{\dt}[1]{\accentset{\mbox{\large\bfseries .}}{#1}}

\newcommand{\cat}[1]{\mathsf{#1}}
\newcommand{\CAT}[1]{\mathsf{#1}}
\newcommand{\catTwo}[1]{\bm{#1}}
\newcommand{\Set}{\CAT{Set}}
\newcommand{\Cat}{\CAT{Cat}}
\newcommand{\Cart}{\CAT{Cart}}
\newcommand{\SMC}{\CAT{SMC}}
\newcommand{\catSet}[1]{\cat{#1}\text{-}\Set}
\newcommand{\Diag}{\operatorname*{\CAT{Diag}}_{\rightarrow}}
\newcommand{\DiagOp}{\operatorname*{\CAT{Diag}}_{\leftarrow}}
\newcommand{\Vect}{\CAT{Vect}}
\newcommand{\Man}{\CAT{Man}}

\newcommand{\op}{\mathrm{op}}
\DeclareMathOperator{\Ob}{Ob}
\DeclareMathOperator{\Hom}{Hom}
\DeclareMathOperator{\cod}{cod}
\DeclareMathOperator{\id}{id}
\DeclareMathOperator{\El}{El}
\DeclareMathOperator*{\colim}{colim}
\DeclareMathOperator{\Disc}{Disc}
\DeclareMathOperator{\Sub}{Sub}
\DeclareMathOperator{\PSh}{PSh}
\DeclareMathOperator{\Sh}{Sh}
\DeclareMathOperator{\DOpf}{DOpf}
\newcommand{\sheaf}[1]{\mathcal{#1}}

\usepackage{subcaption}
\usepackage[capitalize,noabbrev,nameinlink]{cleveref}

\usepackage{bookmark}
\bookmarksetup{numbered,open}

\newtheorem{theorem}{Theorem}[section]
\newtheorem{proposition}[theorem]{Proposition}
\newtheorem{lemma}[theorem]{Lemma}

\theoremstyle{definition}
\newtheorem{definition}[theorem]{Definition}
\newtheorem{example}[theorem]{Example}
\newtheorem{remark}[theorem]{Remark}

\usepackage{tikz}
\usetikzlibrary{fit,positioning,shapes.geometric}
\usetikzlibrary{cd}
\usepackage{quiver}
\tikzset{ %
  touch src/.style={shorten <=-4pt},
  touch tgt/.style={shorten >=-4pt}}
\tikzset{ %
  place/.style={ellipse,thick,draw=black!75,minimum size=5mm,inner sep=0.25mm},
  big place/.style={ellipse,thick,draw=black!75,minimum size=5mm,inner sep=0.25mm},
  transition/.style={rectangle,thick,draw=black!75},
  sum place/.style={circle,thick,draw=black!75,minimum size=2.5mm,inner sep=0}}
\tikzset{ %
  boxUWD/.style={rectangle,draw,rounded corners}
}

\begin{document}

\title{A diagrammatic view of differential equations in physics}
\author{%
  Evan Patterson\affil{1,}\corrauth,
  Andrew Baas\affil{2},
  Timothy Hosgood\affil{1} and
  James Fairbanks\affil{3}
}
\shortauthors{the Author(s)}

\address{%
  \addr{\affilnum{1}}{Topos Institute, Berkeley, CA, USA}
  \addr{\affilnum{2}}{Georgia Tech Research Institute, Atlanta, GA, USA}
  \addr{\affilnum{3}}{University of Florida, Gainesville, FL, USA}}

\corraddr{Email: evan@epatters.org.
}

\begin{abstract}
Presenting systems of differential equations in the form of diagrams has become common in certain parts of physics, especially electromagnetism and computational physics. In this work, we aim to put such use of diagrams on a firm mathematical footing, while also systematizing a broadly applicable framework to reason formally about systems of equations and their solutions. Our main mathematical tools are category-theoretic diagrams, which are well known, and morphisms between diagrams, which have been less appreciated. As an application of the diagrammatic framework, we show how complex, multiphysical systems can be modularly constructed from basic physical principles. A wealth of examples, drawn from electromagnetism, transport phenomena, fluid mechanics, and other fields, is included.
\end{abstract}


\maketitle

\section{Introduction} \label{sec:introduction}

Presenting physical quantities and equations between them in the form of diagrams has a long tradition in physics and engineering, especially in those parts, such as electromagnetism, where the exterior calculus is commonly used. To a first approximation, these diagrams are directed graphs where the nodes represent physical quantities, such as fields or densities, and the arrows represent the action of operators on the quantities. The operators are usually differential operators such as the exterior derivative.

The most famous example, evocatively called ``Maxwell's house'' by Bossavit \cite{bossavit1998d}, depicts Maxwell's equations for electric and magnetic fields as a three-dimensional diagram. Two variants of Maxwell's house are reproduced from the literature in \cref{fig:maxwell-house-traditional}. In the exterior calculus formulation (\cref{subfig:maxwell-house-traditional-exterior}), the layout of the diagram is meaningful: (untwisted) time-dependent differential forms, such the electric field 1-form $E$ and magnetic field 2-form $B$, are placed on the left side, whereas twisted differential forms, such as the charge density 0-form $\widetilde \rho$ and current density 1-form $\widetilde J$, are placed on the right side. Spatial (exterior) derivatives go up or down the page, time derivatives go into or out of the page, and the Hodge star operators transforming between $k$-forms and twisted $(3-k)$-forms go left or right. For comparison, the classical formulation of Maxwell's equations using vector caclulus is shown in \cref{subfig:maxwell-house-traditional-vector}, where the spatial layout of the diagram is less meaningful.

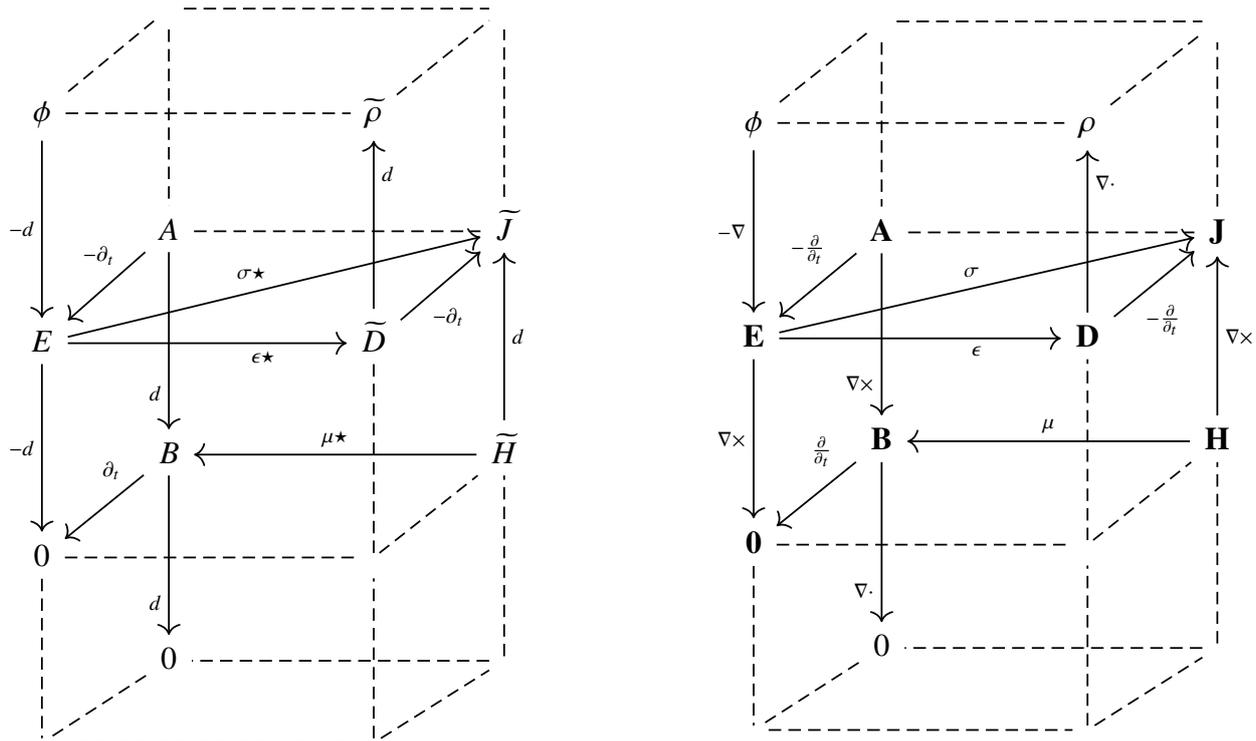
\begin{figure}[H]
    \centering
    \begin{subfigure}{0.45\textwidth}
    \[\begin{tikzcd}
    	& {} &&& {} \\
    	\phi &&& {\widetilde \rho} \\
    	& A &&& {\widetilde J} \\
    	E &&& {\widetilde D} \\
    	& B &&& {\widetilde H} \\
    	0 &&& {} \\
    	& 0 &&& {} \\
    	{} &&& {}
    	\arrow["{-\partial_t}"'{pos=0.3}, from=3-2, to=4-1]
    	\arrow["{\epsilon \star}"'{pos=0.7}, from=4-1, to=4-4]
    	\arrow[dashed, no head, from=3-2, to=3-5]
    	\arrow["{-\partial_t}"'{pos=0.3}, from=4-4, to=3-5]
    	\arrow["{-d}"', from=4-1, to=6-1]
    	\arrow["{-d}"', from=2-1, to=4-1]
    	\arrow[dashed, no head, from=1-2, to=2-1]
    	\arrow[dashed, no head, from=1-2, to=3-2]
    	\arrow["d"'{pos=0.8}, from=3-2, to=5-2]
    	\arrow["{\partial_t}"'{pos=0.2}, from=5-2, to=6-1]
    	\arrow["d"'{pos=0.8}, from=5-2, to=7-2]
    	\arrow["{\mu \star}"', from=5-5, to=5-2]
    	\arrow[dashed, no head, from=7-2, to=8-1]
    	\arrow["d"', from=5-5, to=3-5]
    	\arrow["d"'{pos=0.8}, from=4-4, to=2-4]
    	\arrow[dashed, no head, from=2-1, to=2-4]
    	\arrow[dashed, no head, from=1-2, to=1-5]
    	\arrow[dashed, no head, from=2-4, to=1-5]
    	\arrow[dashed, no head, from=3-5, to=1-5]
    	\arrow[dashed, no head, from=6-1, to=6-4]
    	\arrow[dashed, no head, from=6-4, to=5-5]
    	\arrow[dashed, no head, from=6-4, to=4-4]
    	\arrow[dashed, no head, from=6-1, to=8-1]
    	\arrow[dashed, no head, from=8-1, to=8-4]
    	\arrow[dashed, no head, from=8-4, to=6-4]
    	\arrow[dashed, no head, from=8-4, to=7-5]
    	\arrow[dashed, no head, from=7-5, to=5-5]
    	\arrow[dashed, no head, from=7-5, to=7-2]
    	\arrow["{\sigma \star}", from=4-1, to=3-5]
    \end{tikzcd}\]
    \caption{Maxwell's house in exterior calculus}
    \label{subfig:maxwell-house-traditional-exterior}
    \end{subfigure}
    \hfill
    \begin{subfigure}{0.45\textwidth}
    \[\begin{tikzcd}
    	& {} &&& {} \\
    	\phi &&& \rho \\
    	& {\mathbf{A}} &&& {\mathbf{J}} \\
    	{\mathbf{E}} &&& {\mathbf{D}} \\
    	& {\mathbf{B}} &&& {\mathbf{H}} \\
    	{\mathbf{0}} &&& {} \\
    	& 0 &&& {} \\
    	{} &&& {}
    	\arrow["{-\frac{\partial}{\partial_t}}"'{pos=0.3}, from=3-2, to=4-1]
    	\arrow["\epsilon"'{pos=0.7}, from=4-1, to=4-4]
    	\arrow[dashed, no head, from=3-2, to=3-5]
    	\arrow["{-\frac{\partial}{\partial_t}}"'{pos=0.3}, from=4-4, to=3-5]
    	\arrow["{\nabla \times}"', from=4-1, to=6-1]
    	\arrow["{-\nabla}"', from=2-1, to=4-1]
    	\arrow[dashed, no head, from=1-2, to=2-1]
    	\arrow[dashed, no head, from=1-2, to=3-2]
    	\arrow["{\nabla \times}"'{pos=0.8}, from=3-2, to=5-2]
    	\arrow["{\frac{\partial}{\partial_t}}"'{pos=0.2}, from=5-2, to=6-1]
    	\arrow["{\nabla \cdot}"'{pos=0.8}, from=5-2, to=7-2]
    	\arrow["\mu"', from=5-5, to=5-2]
    	\arrow[dashed, no head, from=7-2, to=8-1]
    	\arrow["{\nabla \times}"', from=5-5, to=3-5]
    	\arrow["{\nabla \cdot}"'{pos=0.8}, from=4-4, to=2-4]
    	\arrow[dashed, no head, from=2-1, to=2-4]
    	\arrow[dashed, no head, from=1-2, to=1-5]
    	\arrow[dashed, no head, from=2-4, to=1-5]
    	\arrow[dashed, no head, from=3-5, to=1-5]
    	\arrow[dashed, no head, from=6-1, to=6-4]
    	\arrow[dashed, no head, from=6-4, to=5-5]
    	\arrow[dashed, no head, from=6-4, to=4-4]
    	\arrow[dashed, no head, from=6-1, to=8-1]
    	\arrow[dashed, no head, from=8-1, to=8-4]
    	\arrow[dashed, no head, from=8-4, to=6-4]
    	\arrow[dashed, no head, from=8-4, to=7-5]
    	\arrow[dashed, no head, from=7-5, to=5-5]
    	\arrow[dashed, no head, from=7-5, to=7-2]
    	\arrow["\sigma", from=4-1, to=3-5]
    \end{tikzcd}\]
    \caption{Maxwell's house in vector calculus}
    \label{subfig:maxwell-house-traditional-vector}
    \end{subfigure}
    \caption{Traditional drawing of ``Maxwell's house,'' depicting Maxwell's equations for electromagnetism in matter. The formulation in exterior calculus (a) is adapted from Bossavit \cite[Figure 6]{bossavit1998d} and the formulation in vector calculus (b) from Cortes Garcia et al. \cite[Figure~1]{garcia2018}. Further variations on Maxwell's house in the literature include \cite[Diagram ELE3]{tonti2013}, \cite[Table III]{deschamps1981}, and \cite[Figure 8]{lacomba1991}.}
    \label{fig:maxwell-house-traditional}
\end{figure}

Similar diagrams presenting systems of equations in electromagnetics and other branches of physics have been proposed by many authors, including Tonti, Deschamps, and Bossavit \cite{tonti1972,tonti2013,deschamps1981,bossavit1988,bossavit1998d}. Enzo Tonti, in particular, has been an early and consistent advocate for diagrams as a systematic method to describe and classify physical theories. Such diagrams are therefore sometimes called ``Tonti diagrams'' \cite{bossavit1988}, \cite[\S 7E]{gross2004}.

Although it was observed early on that Tonti's diagrams resemble commutative diagrams in algebra~\cite{lacomba1991}, this analogy has never been made precise, and indeed the diagrams, as used in practice, seem to involve a certain degree of imprecision. Let us consider in more detail how \cref{fig:maxwell-house-traditional} encodes Maxwell's equations. An arrow $x \xrightarrow{f} y$ asserts the equation $y = f(x)$. For instance, the arrow $A \xrightarrow{d} B$ asserts that the magnetic field $B$, a 2-form, and the magnetic potential $A$, a 1-form, are related by the equation $B = d A$, where $d$ is the exterior derivative. Moreover, a node with multiple incoming arrows is implicitly the \emph{sum} of those values, so that the pair of arrows $\phi \xrightarrow{-d} E \xleftarrow{-\partial_t} A$ encodes the equation $E = -d\phi - \partial_t A$ that relates a time-dependent electric field to derivatives of the electric and magnetic potentials. Unfortunately, these conventions are not always used consistently, even within a single diagram. In \cref{fig:maxwell-house-traditional}, the current density $\widetilde J$, a twisted 2-form, has three incoming arrows and so, by the second convention, should satisfy the equation $\widetilde J = \sigma \star E - \partial_t \widetilde D + d \widetilde H$, but this is plainly wrong. Rather, the diagram's intended interpretation is that the current density obeys the constitutive equation $\widetilde J = \sigma \star E$, where $\sigma$ is the conductivity of the material, and also the Maxwell equation $\widetilde J = -\partial_t \widetilde D + d \widetilde H$. In Tonti's version of Maxwell's house, inconsistency is avoided but only through a more ad hoc and informal use of diagrams \cite[Diagram ELE3]{tonti2013}.

If one views diagrams merely as a visual aide for comprehending physical theories, treated formally as systems of equations, then any ambiguities in interpretation may seem a minor problem, but if one wishes to understand diagrams as a formal, systematic tool for presenting physical theories, as Tonti and others do \cite{tonti2013,alotto2010}, it becomes essential that the diagrams have a precise and unambiguous meaning. A first aim of this paper is to put diagrams of differential equations in physics on a rigorous footing, realizing them as well-defined algebraic and combinatorial objects.

Besides the intrinsic value of increased clarity, making diagrams into proper mathematical objects enables new ways of reasoning about them, some of which have no obvious analogue for systems of equations represented as symbolic expressions. These new capabilities all stem from having a meaningful notion of \emph{morphism} between diagrams. Morphisms of diagrams will be used to attach boundary conditions to systems, thus formalizing boundary value problems within the diagrammatic approach. More generally, diagram morphisms formalize relationships between different physical theories, ranging from simple inclusions and isomorphisms to more complex relations such as those between time-dependent and time-independent (steady-state) theories.

Naturally, given the central role played by diagrams in category theory, we will formalize diagrams presenting physical theories using category-theoretic ideas. A diagram representing a system of algebraic or differential equations will indeed be nothing other than a diagram in a suitable category $\cat{C}$ of spaces and operators (\cref{sec:diagrams,sec:diff-eqs}). As we will explain, a solution to the system of equations presented by such a diagram is then a lift of the diagram to a category of generalized elements of $\cat{C}$. Thus, solving equations amounts to solving a lifting problem of diagrams. More generally, we will see that solving a boundary value problem---a system of equations subject to boundary conditions---is the same as solving an extension-lifting problem of diagrams (\cref{sec:diagram-morphisms}). Through the influence of Steenrod \cite{steenrod1972}, extension and lifting problems involving geometrical spaces have become a mainstay in algebraic topology \cite{arkowitz2011,riehl2014}. Analogous problems involving categories instead are much less studied (with an exception being Spivak's framing of queries on relational databases as extension-lifting problems \cite{spivak2014}). We study how morphism of diagrams lift through discrete opfibrations, a concept which we also review (\cref{sec:lifting-properties}).

Stylized though it is, the formalism of diagrams in a category suffices to express several fundamental equations of mathematical physics, such as the diffusion equation, the wave equation, and the Maxwell-Faraday equations. However, many important systems, including the full Maxwell equations and the Navier-Stokes equations of fluid mechanics, require additional algebraic structure. The diagrammatic formalism can be extended to accommodate cartesian products (\cref{sec:cartesian-diagrams}) and tensor products (\cref{sec:monoidal-diagrams}). Cartesian products are essential in presenting Maxwell's equations in a manner that avoids the ambiguities of the traditional Maxwell's house (\cref{fig:maxwell-house-traditional}). Tensor products are useful for incorporating multilinear operators on differential forms, including the Lie derivative, interior product, and wedge product. Such operators are needed to formulate the equations of fluid mechanics, where an informal approach to diagrams in physics becomes more obviously untenable (cf. \cite[Diagrams FLU6 and FLU10]{tonti2013}).

Advances in computational physics have brought increasing attention to \emph{multiphysics} models: complex, often analytically intractable models that couple multiple physical phenomena within the same system \cite{keyes2013}. Conjugate heat transfer, for example, describes heat transfer from a solid body into a fluid flowing around the body, combining the physics of heat conduction and fluid mechanics. Diagram morphisms enable physical theories to be composed from basic building blocks, yielding a framework for specifying multiphysics systems that is at once precise and visually intuitive (\cref{sec:composing-diagrams}).

Throughout this paper, we emphasize the syntactical nature of diagrams, inasmuch as different (nonisomorphic) diagrams may still present physical theories that would usually be considered equivalent. This is why we are careful to say that diagrams \emph{present} physical theories, rather than that they \emph{are} physical theories, analogous to how a group or a category may be given different presentations by generators and relations. Although we avoid the confounding question of ``what exactly \emph{is} a physical theory?'', we study in detail the better-posed question of when two diagrams have interchangeable lifts, i.e., when the corresponding systems of equations have interchangeable solutions (\cref{sec:weak-equivalences}). The notion of diagram equivalence thus obtained is weaker than isomorphism of diagrams and is therefore called \emph{weak equivalence}. We present two sufficient conditions for a morphism of diagrams to be a weak equivalence: the first is based on the classical concept of an \emph{initial} functor, which we review; the second, more general condition, uses the notion of a \emph{relatively initial} functor, which seems to be original with this work.

{\bf Background.} We assume that the reader is fluent in the basic concepts of category theory, namely categories, functors, natural transformations, and limits and colimits, as described in introductory texts such as Leinster's \cite{leinster2014} or Riehl's \cite{riehl2016}. More advanced category-theoretic concepts will be introduced in a self-contained manner as they are needed. Some familiarity with differential geometry, particularly differential forms and exterior calculus, will be helpful in understanding the examples. Baez and Muniain elegantly exposit the exterior calculus in the context of electromagnetism \cite{baez1994}, whereas Abraham, Marsden, and Ratiu give a comprehensive treatment of exterior calculus including many applications to continuum mechanics \cite{abraham1988}.

{\bf Related work.} We have already indicated the origins of this work in physics, particularly electromagnetism and computational physics \cite{tonti1972,tonti2013,deschamps1981,bossavit1998d}. More recently, Lahtinen, Kotiuga, and Stenvall have advocated for greater use of category-theoretic methods in discretizing and simulating continuum physics \cite{lahtinen2018}.

In this paper, a central mathematical role is played by categories of diagrams. Despite originating with the founding document of category theory by Eilenberg and Mac Lane \cite[\S 23]{eilenberg1945}, diagram categories do not appear in the standard texts and have remained somewhat obscure in the literature. Early references are Anders Kock's PhD thesis \cite{kock1967} and the works (written in French) of Guitart \cite{guitart1973,guitart1974,guitart1977}. The study of diagram categories has recently been resumed by Peschke and Tholen \cite{peschke2020} and Perrone and Tholen \cite{perrone2021}.

Robinson's work on multi-model systems \cite{robinson2017} is superficially similar to ours inasmuch as systems of differential and difference equations are described using category-theoretic diagrams. However, rather than directly specifying equations via a diagram of spaces and operators as in Tonti's diagrams and this work, Robinson uses set-valued diagrams indexed by posets to encode the dependency relations between the variables in a set of equations, resulting in a significantly different formalism.

{\bf Conventions.} Given composable morphisms $x \xrightarrow{f} y \xrightarrow{g} z$ in a category, the composite is written in diagrammatic order as $f \cdot g$ or $fg$ or in conventional order as $g \circ f$. The word ``graph'' is used in the category theorist's sense to mean a directed graph, possibly with multiple edges and loops.

\section{Systems of equations as diagrams} \label{sec:diagrams}

We begin by explaining how systems of equations, as well as their solutions, can be described using category-theoretic diagrams. To make the ideas as transparent as possible, we work with linear difference equations, which are discretizations of linear differential equations that require nothing more than linear algebra to formulate. We pass from such algebraic equations to differential equations in the next section.

Recall that the term ``diagram'' has a precise meaning in category theory.

\begin{definition}[Diagram]
A \emph{diagram} in a category $\cat{C}$ is a functor $D: \cat{J} \to \cat{C}$, where $\cat{J}$, the \emph{shape} or \emph{indexing category} of the diagram, is a small category.
\end{definition}

Two common classes of diagrams are free diagrams and commutative diagrams. A diagram $D: \cat{J} \to \cat{C}$ is \emph{free} if the indexing category $\cat{J}$ is a free category, meaning that there exists a graph $G$ such that $\cat{J}$ is the path category generated by $G$. Note that this does not imply anything about the category $\cat{C}$. In our examples the diagrams will always happen to be free, but the theory does not rely on any assumption of freeness. A free diagram \emph{commutes} if, for any two paths in the generating graph with the same source and target, their images under the functor $D$ are equal morphisms in $\cat{C}$.

The fundamental interpretation to keep in mind is the following. From the perspective of equations, a commutative diagram asserts equations that hold \emph{universally}, whereas a diagram that does not commute presents equations that \emph{do not} generically hold but for which one can seek \emph{specific} solutions.

\begin{example}[Discrete heat equation] \label{ex:discrete-heat-equation}
The classical heat equation has a discrete analogue on graphs. Let $G$ be a \emph{symmetric weighted graph}: a diagram of sets and functions
\begin{equation*}
    \begin{tikzcd}
    V & E \arrow[l, shift right, "s"'] \arrow[l, shift left, "t"]
     \arrow[loop above, "i"] \arrow[r, "\mu"] & \R_{>0}
    \end{tikzcd} 
\end{equation*}
where $s,t$ are the source and target maps, $i$ is the edge involution, and $\mu$ is the edge weighting, satisfying $i^2 = \id_E$, $i s = t$, $i t = s$, and $i \mu = \mu$. Assume that $G$ has no isolated vertices and is \emph{locally finite}, meaning that every vertex has finitely many incident edges. For example, $G$ could be the $d$-dimensional integer lattice $\Z^d$, viewed as a symmetric weighted graph with all weights equal to 1.

In discrete time, a time-dependent real-valued function on the vertices of the graph is a vector $u \in \R^{\N \times V}=\Set(\N\times V,\R)$. The \emph{discrete time derivative} is the difference operator $\partial_n\colon\R^{\N\times V}\to\R^{\N\times V}$ given by
\begin{equation*}
    (\partial_n u)(n,x) := u(n+1, x) - u(n, x).
\end{equation*}
The \emph{discrete Laplace operator} $\Lap\colon\R^{\N\times V}\to\R^{\N\times V}$ is defined by
\begin{equation*}
    (\Lap u)(n,x) := \frac{1}{\mu(x)} \sum_{e \in s^{-1}(x)} u(n, t(e))\, \mu(e) - u(n, x)
    = \frac{1}{\mu(x)} \sum_{e \in s^{-1}(x)} (u(n,t(e)) - u(n,x))\, \mu(e),
\end{equation*}
where $\mu(x) := \sum_{e \in s^{-1}(x)} \mu(e)$ is the \emph{weight} of vertex $x \in V$ \cite[\S 1.4]{grigoryan2018}. (Since the graph $G$ is symmetric, the roles of $s$ and $t$ can be switched in these formulas.)

With this setup, the \emph{discrete heat equation} on the graph $G$ is
\begin{equation} \label{eq:discrete-heat}
    \partial_n u = \Lap u, \qquad u \in \R^{\N \times V}.
\end{equation}
When $G = \Z^d$, this equation is a discrete version of the heat equation in $\R^d$ \cite[\S 1.3]{lawler2010}. We can present the discrete heat equation via the diagram
\begin{equation} \label{eq:discrete-heat-diagram}
    \begin{tikzcd}
    \R^{\N \times V} \arrow[r, shift left, "\partial_n"] \arrow[r, shift right, "\Lap"']
    & \R^{\N \times V}
    \end{tikzcd}
\end{equation}
in the category $\Vect_\R$ of real vector spaces and linear maps.
This is a free diagram $D: \cat{J} \to \Vect_\R$ whose shape $\cat{J} := \{\bullet \rightrightarrows \bullet\}$ is the free category with two parallel arrows. The diagram does not commute, as the equation is not a tautology.
\end{example}

In the category of sets and functions, a solution to a system of equations presented by a diagram $D: \cat{J} \to \Set$ is a choice of elements $(x_j \in Dj)_{j\in\cat{J}}$ that is compatible with all morphisms in the diagram:
\begin{equation} \label{eq:diagram-lift-sets}
    (D f)(x_j) = x_k \qquad
    \text{ for all } j \xrightarrow{f} k \text{ in } \cat{J}.
\end{equation}
In general, the objects of a category do not have elements, so we speak instead of generalized elements when considering solutions for a diagram in an arbitrary category.

Recall that, for a category $\cat{C}$ and an object $S\in\cat{C}$, a \emph{generalized element of shape $S$ in $\cat{C}$} is simply a morphism $x: S \to X$ in $\cat{C}$. When $\cat{C}$ is a concrete category, it is usually possible to choose $S$ so that generalized elements of shape $S$ correspond to elements of the underlying sets. Alternatively, when $\cat{C}$ is a closed category, it is natural to take $S$ as the unit object of the closed structure. Exemplifying both situations, generalized elements $x: 1 \to X$ in $\Set$ correspond to elements $x \in X$, and generalized elements $v: \R \to V$ in $\Vect_\R$ correspond to vectors $v \in V$.

For any fixed category $\cat{C}$ and object $S \in \cat{C}$, generalized elements of shape $S$ are objects of the coslice category $S/\cat{C}$. The morphisms from $x: S \to X$ to $y: S \to Y$ in this category are morphisms $f: X \to Y$ in $\cat{C}$ forming a commutative triangle
\begin{equation*}
    \begin{tikzcd}
    	& S \\
    	X && Y
    	\arrow["x"', from=1-2, to=2-1]
    	\arrow["y", from=1-2, to=2-3]
    	\arrow["f", from=2-1, to=2-3]
    \end{tikzcd}.
\end{equation*}
The coslice category $S/\cat{C}$ will be suggestively denoted $\El_S (\cat{C})$ or, when the dependence on $S$ is clear from context, simply $\El(\cat{C})$.

Generalizing from diagrams in $\Set$, a solution (or family of solutions) to a system of equations presented by a diagram $D$ in $\cat{C}$ is a lift of the diagram $D$ to a diagram in $\El(\cat{C})$. The following definition makes this idea precise.

\begin{definition}[Lifting problem] \label{def:lifting-problem}
A \emph{lift} of a diagram $D: \cat{J} \to \cat{C}$ through a functor $\pi: \cat{E} \to \cat{C}$ is a diagram $\overline D: \cat{J} \to \cat{E}$ such that $\pi \circ \overline D = D$. Given a diagram $D$ in $\cat{C}$ and a functor $\pi: \cat{E} \to \cat{C}$, the \emph{lifting problem} is to find a lift of $D$ through $\pi$.
\begin{equation*}
    \begin{tikzcd}
    	& {\mathsf{E}} \\
    	{\mathsf{J}} & {\mathsf{C}}
    	\arrow["\pi", from=1-2, to=2-2]
    	\arrow["D"', from=2-1, to=2-2]
    	\arrow["{\overline D}", dashed, from=2-1, to=1-2]
    \end{tikzcd}
\end{equation*}
\end{definition}

\noindent In the present situation, $\cat{E} := \El_S(\cat{C})$ is a category of generalized elements of $\cat{C}$ and $\pi := \cod: \El_S(\cat{C}) \to \cat{C}$ is the forgetful functor sending a generalized element $x: S \to X$ to its codomain $X$.

To summarize, a diagram in $\cat{C}$ defines a system of equations involving the objects (``spaces'') and morphisms (``operators'') of $\cat{C}$. Solving the equations is the same as solving the lifting problem with respect to the functor $\pi = \cod: \El(\cat{C}) \to \cat{C}$.

\begin{example}[Solving the discrete heat equation]
A lift of the diagram (\ref{eq:discrete-heat-diagram})
\begin{equation*}
    \begin{tikzcd}
    \R^{\N \times V} \arrow[r, shift left, "\partial_n"] \arrow[r, shift right, "\Lap"']
    & \R^{\N \times V}
    \end{tikzcd}
\end{equation*}
in $\Vect_\R$ specifying the discrete heat equation to a diagram in $\El_\R(\Vect_\R)$ consists of two vectors $u, \dt u \in \R^{\N \times V}$ such that
\begin{equation*}
    \partial_n u = \dt u \qquad\text{and}\qquad \Lap u = \dt u.
\end{equation*}
Thus, we see that lifts of the diagram correspond to single vectors $u \in \R^{\N \times V}$ satisfying the discrete heat equation (\ref{eq:discrete-heat}). In this case, the solutions are easily characterized: once the initial data $u(0,-) \in \R^V$ is given, the solution is determined recursively by
\begin{equation*}
    u(n+1,x) = u(n,x) + \Lap u(n,x)
    = \frac{1}{\mu(x)} \sum_{e \in s^{-1}(x)} u(n, t(e))\, \mu(e).
\end{equation*}
Initial value problems are treated formally in \cref{sec:diagram-morphisms}.
\end{example}

We introduce a convention for drawing diagrams that makes the equations presented more transparent and more closely resemble the diagrams in the physics literature. Namely, nodes in the diagrams will be annotated with variables as well as the spaces they inhabit. For example, the discrete heat equation (\ref{eq:discrete-heat-diagram}) is redrawn as
\begin{equation} \label{eq:discrete-heat-diagram-el}
    \begin{tikzcd}
    u: \R^{\N \times V} \arrow[r, shift left, "\partial_n"] \arrow[r, shift right, "\Lap"']
    & \dt u: \R^{\N \times V}
    \end{tikzcd}.
\end{equation}
When presenting a physical theory diagrammatically, we will be deliberately ambiguous about whether the diagram is in $\cat{C}$ or $\El(\cat{C})$, just as a classical equation like ``$f(x) = y$'' is, without further context, ambiguous about whether $x$ and $y$ are indeterminates waiting to be filled or are specific values satisfying the equation. In the latter kind of diagram, the type-theoretic expression $x: X$ is interpreted as a generalized element $x: U \to X$ and an arrow $(x: X) \xrightarrow{f} (y:Y)$ as a morphism $X \xrightarrow{f} Y$ in $\cat{C}$ such that $x \cdot f = y$. For example, we can think of the drawing (\ref{eq:discrete-heat-diagram-el}) as defining a diagram $D: \cat{J} \to \Vect_\R$ with indexing category $\cat{J} = \{u \rightrightarrows \dt u\}$, having two objects named ``$u$'' and ``$\dt u$,'' and with assignments $D(u) = D(\dt u) = \R^{\N \times V}$; alternatively, the drawing (\ref{eq:discrete-heat-diagram-el}) can be seen as a diagram in $\El(\Vect_\R)$, involving two specific vectors $u$ and $\dt u$. We emphasize that this visual pun will not create in any ambiguity in the theoretical results, where it will always be clear to which categories the diagrams belong.

{\bf Connection with limits.} The limit of a diagram $D: \cat{J} \to \cat{C}$, when it exists, is a ``universal solution'' to the given equations. For many choices of category $\cat{C}$, including $\Set$ and $\Vect_\R$, the limit comprises \emph{all} solutions to the equations. One might expect that lifts and limits are related, and that is indeed the case, as we now explain.

Recall that a \emph{cone} over the diagram $D$ with \emph{apex} $S \in \cat{C}$ is a natural transformation $\lambda: \Delta_S \To D$, where $\Delta_S: \cat{J} \to \cat{C}$ is the constant diagram at $S$. Equivalently, a cone over $D$ is a natural transformation into $D$ whose domain diagram factors through the terminal category $\cat{1}$.
\begin{equation*}
    \begin{tikzcd}
    	& {\mathsf{J}} \\
    	{\mathsf{1}} && {\mathsf{C}}
    	\arrow[""{name=0, anchor=center, inner sep=0}, "{!}"', from=1-2, to=2-1]
    	\arrow["S"', from=2-1, to=2-3]
    	\arrow[""{name=1, anchor=center, inner sep=0}, "D", from=1-2, to=2-3]
    	\arrow["\lambda", shift right=2, shorten <=6pt, shorten >=6pt, Rightarrow, from=0, to=1]
    \end{tikzcd}
\end{equation*}
A \emph{limit} of the diagram $D$ is a cone over $D$ with the universal property of being terminal in the category of cones over $D$.

The following is immediate from the definitions.
\begin{proposition}[Lifts and cones] \label{prop:lifts-and-cones}
A lift of a diagram $D: \cat{J} \to \cat{C}$ through $\cod: \El_S(\cat{C}) \to \cat{C}$ is precisely a cone over $D$ with apex $S$.
\end{proposition}

This perspective is useful to keep in mind, even if in subsequent sections we consider generalizations of the lifting problem that do not fit so neatly with cones and limits.

\section{Differential equations as diagrams} \label{sec:diff-eqs}

Physical theories from continuum mechanics, electromagnetics, and other parts of classical physics are expressed in terms of physical quantities defined on geometric spaces. Thus, to formulate the theories diagrammatically, we need a category accommodating two distinct kinds of entities: geometric spaces, which we take to be smooth manifolds, and quantities on those spaces, which will be differential forms and vector fields. In this section, we construct a category fulfilling this purpose; it is neither the most general nor the most elegant setting conceivable, but it suffices for the examples we will consider while requiring a minimum of nonstandard concepts.

A smooth manifold $M$, usually with boundary $\partial M$, serves as the spatial domain. In non-static situations, the time domain is the half-open interval $[0,\infty)$, and so the full space-time domain is $M \times [0,\infty)$. Unless $M$ is a manifold without boundary, the product $M \times [0,\infty)$ is not a manifold with boundary but rather a manifold with corners \cite{melrose1996}. We therefore work in the category of smooth manifolds with corners, denoted $\Man$, and take the word ``manifold'' without qualification to mean a smooth manifold, possibly with boundary or corners. Although manifolds with corners are perhaps not as well known as they could be, using them poses no serious difficulties. For example, any manifold with corners $M$ embeds as a submanifold of an open manifold $\widetilde{M}$ of the same dimension, and any smooth function on $M$ extends to a smooth function on an open neighborhood of $M$ in $\widetilde{M}$ \cite{michor2020}. More generally, any differential form on $M$ extends to a differential form on an open neighborhood \cite{gurer2019}.

Physical quantities are taken to be sections of sheaves on the spatial domain $M$ or the space-time domain $M \times [0,\infty)$. Loosely speaking, a sheaf on a topological space is a coherent family of quantities defined on open subsets that can be glued together from smaller subsets whenever they agree on the overlaps. In this paper, sheaves will be used more as a convenient language than as a technical tool, and it will suffice to recall the following terminology \cite{wedhorn2016}.

Let $X$ be a topological space and let $\Open(X)$ be the poset of open subsets of $X$, ordered by inclusion. A \emph{presheaf} on $X$ is a functor $\sheaf{F}: \Open(X)^\op \to \Set$. Explicitly, a presheaf on $X$ consists of, for each open subset $U \subseteq X$, a set $\sheaf{F}(U)$, whose elements $s \in \sheaf{F}(U)$ are called \emph{sections over $U$}, and a functorial assignment of, for each inclusion $V \subseteq U$, a function $\sheaf{F}(U) \to \sheaf{F}(V)$, called the \emph{restriction map} and denoted $\res_V^U$. The action of the restriction map on a section $s \in \sheaf{F}(U)$ is abbreviated as $s|_V := \res_V^U(s)$. A \emph{global section} is a section over the open subset $X$ itself. Finally, a \emph{sheaf} on $X$ is a presheaf $\sheaf{F}$ on $X$ satisfying the \emph{sheaf condition}: if $(s_i \in \sheaf{F}(U_i))_{i \in I}$, is a family of sections that agree on all pairwise intersections, i.e., such that
\begin{equation*}
    s_i|_{U_i \cap U_j} = s_j|_{U_i \cap U_j} \quad
    \text{ for all } i, j \in I,
\end{equation*}
then there exists a unique section $s \in \sheaf{F}(U)$ over the union $U := \bigcup_{i \in I} U_i$ such that $s|_{U_i} = s_i$ for all $i \in I$. The inclusion functor $\Sh(X) \hookrightarrow \PSh(X)$ has a left adjoint $(-)^\sharp: \PSh(X) \to \Sh(X)$, called \emph{sheafification}, that freely associates a sheaf to any presheaf.

Sheaves can also be defined in categories besides $\Set$. For any category $\cat{S}$, an \emph{$\cat{S}$-valued presheaf} on a topological space $X$ is a functor $\sheaf{F}: \Open(X)^\op \to \cat{S}$. An \emph{$\cat{S}$-valued sheaf} is such a presheaf that additionally satisfies a sheaf condition abstracting the one above: for any sub-poset $P$ of $\Open(X)$ that is closed under intersections, the diagram $\sheaf{F}|_P: P^\op \to S$ has as limit the cone with apex $\sheaf{F}(U)$ and legs $\res_V^U$, $V \in P$, where $U := \bigcup_{V \in P} V$. The $\cat{S}$-valued presheaves and sheaves on $X$ each form categories, denoted $\PSh(X, \cat{S})$ and $\Sh(X, \cat{S})$, respectively. The morphisms in both categories are natural transformations. The most important case for us is that of $\Vect_\R$-valued sheaves, or \emph{sheaves of real vector spaces}. The categories of $\Vect_\R$-valued presheaves and sheaves on a space $X$ are abbreviated as $\PSh_\R(X)$ and $\Sh_\R(X)$.
Note that, for an arbitrary category $\cat{S}$, the sheafification functor $(-)^\sharp:\PSh(X,\cat{S})\to\Sh(X,\cat{S})$ might not exist, but it does for the category $\cat{S} = \Vect_\R$. The technical condition that ensures this is the \emph{IPC property} \cite[Definition~3.1.10]{kashiwara2006}.

As the term ``section'' suggests, a principal source of sheaves of vector spaces in differential geometry is sheaves of sections of smooth vector bundles, such as the tangent or cotangent bundles. For the precise relation between sheaves and vector bundles, see \cite[Proposition 8.45]{wedhorn2016}. The following sheaves of vector spaces, in which $M$ is a smooth manifold of dimension $m$, are essential for our purposes:
\begin{itemize}[noitemsep]
\item $\VectField_M \in \Sh_\R(M)$: smooth vector fields on $M$
\item $\VectField_{t,M} \in \Sh_\R(M \times [0,\infty))$: smooth \emph{time-dependent} vector fields on $M$
\item $\Omega^k_M \in \Sh_\R(M)$, for $0 \leq k \leq m$: differential $k$-forms on $M$
\item $\Omega^k_{t,M} \in \Sh_\R(M \times [0,\infty))$, for $0 \leq k \leq m$: \emph{time-dependent} differential $k$-forms on $M$ (i.e., the $k$-forms on $M \times [0,\infty)$ not involving the 1-form $dt$)
\item $\widetilde \Omega^k_M \in \Sh_\R(M)$ and $\widetilde \Omega^k_{t,M} \in \Sh_\R(M \times [0,\infty))$, for $0 \leq k \leq m$: \emph{twisted} (time-independent and time-dependent) differential $k$-forms on $M$.
\end{itemize}
Note that $\Omega^0_M = \mathcal{C}^\infty_M$, the sheaf of smooth functions on $M$. When the manifold $M$ is clear from context, we will drop the subscript in the notation.

All of these objects are standard in differential geometry, except perhaps the twisted differential forms, which are better known in physics \cite{burke1985,bossavit1998d}. Unlike an ordinary form, a twisted form changes sign under an orientation-reversing change of coordinates; formally, an \emph{(untwisted) $k$-form} is a section of the $k$th-order exterior bundle $\Lambda^k T^* M$, whereas a \emph{twisted $k$-form} \cite[\S I.7]{bott1982} is a section of the tensor product bundle
\begin{equation*}
    \widetilde \Lambda^k T^*M := (\Lambda^k T^* M) \otimes \Orient(M). 
\end{equation*}
Here $\Orient(M)$ is the \emph{orientation bundle} of $M$, the real line bundle whose transition functions are given by the sign of the Jacobian determinant of the coordinate transition functions \cite{bott1982,michor2008}. Particularly important are the top-dimensional twisted forms, called \emph{densities}, which can be integrated over an arbitrary manifold, oriented or not. Physical quantities such as a mass density or charge density are indeed modeled by densities (cf. \cref{subfig:maxwell-house-traditional-exterior}).

Another class of sheaves, the constant sheaves, are useful in working with generalized elements of sheaves. For any object $S \in \cat{S}$, the \emph{constant presheaf} on $X$ is the presheaf $S_X: \Open(X)^\op \to \cat{S}$ with $S_X(U) := S$ for all open sets $U \subseteq X$ and all restriction maps equal to the identity on $S$. In general, the constant presheaf is not a sheaf; the \emph{constant sheaf} on $X$ is the sheafification $S_X^\sharp$, which exists in particular when $\cat{S}=\Vect_\R$.

\begin{proposition}[Constant sheaves and generalized elements] \label{prop:elements-of-sheaves}
    Let $\cat{S}$ be a category for which the sheafification functor $(-)^\sharp: \PSh(X,\cat{S}) \to \Sh(X,\cat{S})$ exists. For any object $S \in \cat{S}$ and any $\cat{S}$-valued sheaf $\sheaf{F}$ on a space $X$, the sheaf morphisms $S_X^\sharp \to \sheaf{F}$ are in natural bijection with generalized elements of $\sheaf{F}(X)$ of shape $S$.
    
    In particular, when $\sheaf{F}$ is a sheaf of real vector spaces, the sheaf morphisms $\R_X^\sharp \to \sheaf{F}$ are in natural bijection with global sections of $\sheaf{F}$.
\end{proposition}
\begin{proof}
    The sheafification functor is defined to be left adjoint to the inclusion functor $\Sh(X,\cat{S}) \hookrightarrow \PSh(X,\cat{S})$. Thus, sheaf morphisms $S_X^\sharp \to \sheaf{F}$ are in natural bijection with presheaf morphisms $S_X \to \sheaf{F}$. The latter are natural transformations $\alpha: S_X \To \sheaf{F}$, where naturality implies that, for every open set $U \subseteq X$, the diagram
    \begin{equation*}
    \begin{tikzcd}
    	& S \\
    	{\sheaf{F}(X)} && {\sheaf{F}(U)}
    	\arrow["{\alpha_X}"', from=1-2, to=2-1]
    	\arrow["{\mathrm{res}_U^X}", from=2-1, to=2-3]
    	\arrow["{\alpha_U}", from=1-2, to=2-3]
    \end{tikzcd}
    \end{equation*}
    commutes. It follows that the natural transformation $\alpha$ is completely determined by its component $\alpha_X$, which is a generalized element of $\sheaf{F}(X)$ of shape $S$.
\end{proof}

We use this observation to formulate partial differential equations as lifting problems of diagrams.

\begin{example}[Maxwell-Faraday equations] \label{ex:maxwell-faraday}
Let $M$ be a smooth manifold. Phrased in exterior calculus, the \emph{Maxwell-Faraday equations} on the spatial domain $M$ are
\begin{align*}
    d E &= - \partial_t B \\
    d B &= 0,
\end{align*}
where $E \in \Omega^1_t(M)$ is the electric field as a time-dependent 1-form and $B \in \Omega^2_t(M)$ is the magnetic field as a time-dependent 2-form. This system of partial differential equations is encoded by the diagram
\begin{equation*}
\begin{tikzcd}
	{E: \Omega^1_t} & {\dt B: \Omega^2_t} \\
	& {B: \Omega^2_t} & {\Omega^3_t} & 0
	\arrow["{-d}", from=1-1, to=1-2]
	\arrow["d", from=2-2, to=2-3]
	\arrow["{\partial_t}"', from=2-2, to=1-2]
	\arrow[from=2-4, to=2-3]
\end{tikzcd}
\end{equation*}
in $\Sh_\R(M \times [0,\infty))$. In drawing the diagram, we follow the conventions described in \cref{sec:diagrams} and we also abbreviate the sheaf $\Omega_{t,M}^k$ on $M \times [0,\infty)$ as $\Omega_t^k$. The object labeled $0$ is the zero sheaf on $M \times [0,\infty)$, which is the zero object in the category of sheaves of vector spaces.

By \cref{prop:elements-of-sheaves}, a lift of the diagram to a diagram of generalized elements of shape $\R_{M \times [0,\infty)}^\sharp$ is a solution to the Maxwell-Faraday equations on $M$ for all time $t \geq 0$.
\end{example}

The Maxwell-Faraday equations comprise the non-metric half of Maxwell's equations, which make sense on any smooth manifold and are preserved by diffeomorphisms \cite{baez1994}. Most equations of mathematical physics, including the full Maxwell equations, require a metric and so must be defined on a Riemannian or semi-Riemannian manifold. The categorical constructions above carry over directly when the category of smooth manifolds is replaced by the category of Riemannian manifolds (with corners) and isometries, and in the following examples it should be clear from context whether we are working with smooth manifolds or Riemannian manifolds.

In the metric setting, we use the Hodge star operator in a form that is more common certain parts of physics than in mathematics \cite{burke1985,bossavit1998d,barham2021}. Let $M$ be a Riemmanian manifold of dimension $m$. For each $0 \leq k \leq m$, the Hodge star is an isomorphism of vector bundles $\Lambda^k T^* M \simto \widetilde \Lambda^{m-k} T^* M$ \cite[Definition~7.13]{ramanan2005}. Thus, the Hodge star induces an isomorphism of sheaves of sections
\begin{equation*}
    \star := \star_k: \Omega^k_M \simto \widetilde \Omega^{m-k}_M
\end{equation*}
that sends $k$-forms to twisted $(m-k)$-forms. There are several reasons to prefer this version of the Hodge star. Unlike the usual Hodge star that sends $k$-forms to (untwisted) $(m-k)$-forms, this one does not require the manifold to be oriented and so avoids invoking superfluous structure. More importantly, the distinction between straight and twisted forms is physically meaningful, which becomes especially apparent in computational physics. The Hodge star involving twisted forms permits a direct translation to the discrete exterior calculus \cite{hirani2003,desbrun2005}, where straight and twisted forms manifest as ``primal'' and ``dual'' forms on a discretized manifold.

Using the metric structure, we now present two of the most famous equations of mathematical physics---the diffusion and wave equations---as diagrams.

\begin{example}[Diffusion] \label{ex:diffusion}
Consider a substance diffusing on a three-dimensional Riemannian manifold $M$. Important physical quantities in this system are the concentration of the substance, $C \in \Omega_t^0(M)$, and the negative of the diffusion flux, $\phi \in \widetilde\Omega_t^2(M)$. According to \emph{Fick's first law} \cite[\S 1.2]{crank1975}, the negative diffusion flux is proportional to the derivative of the concentration:
\begin{equation*}
    \phi = k \star d C,
\end{equation*}
where the proportionality constant $k \in C^\infty(M)$, called the \emph{diffusivity}, is a property of the material. Applying a conservation of mass principle \cite[\S 1.3]{crank1975}, one obtains the \emph{diffusion equation}
\begin{equation*}
    \partial_t C = \star^{-1} d \phi = \star^{-1} d (k \star) d C,
\end{equation*}
which can be neatly expressed by the diagram
\begin{equation} \label{eq:diffusion-diagram}
    \begin{tikzcd}
    	{C: \Omega_t^0} & {\dt C: \Omega_t^0} & {d\phi: \widetilde \Omega_t^3} \\
    	{dC: \Omega_t^1} & {} & {\phi: \widetilde \Omega_t^2}
    	\arrow["d"', from=1-1, to=2-1]
    	\arrow["{\partial_t}", from=1-1, to=1-2]
    	\arrow["k\star"', from=2-1, to=2-3]
    	\arrow["d"', from=2-3, to=1-3]
    	\arrow["{\star^{-1}}"', from=1-3, to=1-2]
    \end{tikzcd}
\end{equation}
in $\Sh_\R(M \times [0,\infty))$.

When the diffusivity $k$ is constant throughout $M$, the diffusion equation reduces to 
\begin{equation*}
    \partial_t C = k \delta d C = k \Delta C,
\end{equation*}
where $\delta := \star^{-1} \circ d \circ \star$ is the \emph{codifferential} and $\Delta := \delta \circ d$ is the \emph{Laplace-Beltrami operator}. In this special case, the diffusion equation is also called the \emph{heat equation} (cf.\ \cref{ex:discrete-heat-equation}). It can be expressed as a simpler diagram:
\begin{equation} \label{eq:heat-diagram}
    \begin{tikzcd}
    	{C: \Omega_t^0} & {\dt C: \Omega_t^0}
    	\arrow["{k \Delta}"', shift right=1, from=1-1, to=1-2]
    	\arrow["{\partial_t}", shift left=1, from=1-1, to=1-2]
    \end{tikzcd}.
\end{equation}
The notion of a morphism of diagrams, introduced in the next section, will allow us to make precise the relationship between the two diagrams (\ref{eq:diffusion-diagram}) and (\ref{eq:heat-diagram}). They are not isomorphic but, provided that $k$ is constant, they are equivalent in a weaker sense which we will make precise.
\end{example}

\begin{example}[Waves]
The \emph{wave equation} on an $m$-dimensional Riemannian manifold $M$ is
\begin{equation*}
    \partial_t^2 u = c^2 \Delta u,
\end{equation*}
where $u \in \Omega^0_t(M)$ is the \emph{displacement}, $c^2$ is a positive constant (the \emph{squared wave speed}), and $\Delta := \delta\circ d$ is the Laplace-Beltrami operator \cite{vasy2008}. Introducing another variable, the \emph{velocity} $v \in \Omega^0_t(M)$, the wave equation reduces to a first-order system in time:
\begin{align*}
    \partial_t u &= v \\
    \partial_t v &= c^2 \Delta u.
\end{align*}
Expanding the Laplacian into its constituent operators, we can present the wave equation by the diagram
\begin{equation*}
    \begin{tikzcd}
    	{u: \Omega_t^0} & {v: \Omega_t^0} & {\dt v: \Omega_t^0} & {\widetilde \Omega_t^m} \\
    	{\Omega_t^1} &&& {\widetilde \Omega_t^{m-1}}
    	\arrow["{\partial_t}", from=1-1, to=1-2]
    	\arrow["d"', from=1-1, to=2-1]
    	\arrow["{\partial_t}", from=1-2, to=1-3]
    	\arrow["{c^2 \star}", from=2-1, to=2-4]
    	\arrow["d"', from=2-4, to=1-4]
    	\arrow["{\star^{-1}}"', from=1-4, to=1-3]
    \end{tikzcd}
\end{equation*}
in $\Sh_\R(M \times [0,\infty))$.
\end{example}

Initial and boundary value problems, formulated in the next section, involve not just quantities defined globally on a space-time manifold but also quantities restricted to the manifold's spatial or temporal boundary. In order to express such problems, the underlying geometric objects---smooth manifolds with corners---and the spaces of quantities on them---sheaves of vector spaces---must be assembled into a single category accommodating both kinds of entities. For this, indexed categories and the Grothendieck construction are useful devices.

Given a category $\cat{J}$, a \emph{$\cat{J}$-indexed category} is a (pseudo)functor $\cat{X}: \cat{J}^\op \to \Cat$, thought of as a family of categories $\cat{X}(j)$ contravariantly indexed by $j \in \cat{J}$. The Grothendieck construction assembles the categories $\cat{X}(j)$, $j \in \cat{J}$, into a single category $\int \cat{X}$ with objects $\Ob(\int \cat{X}) := \sum_{j \in \cat{J}} \Ob \cat{X}(j)$. The morphisms in this category can be defined covariantly or contravariantly, depending on the form of the construction \cite[\S 10.5]{peschke2020}, \cite{spivak2019}. We take the covariant Grothendieck construction of $\cat{X}$, viewed as a covariant functor from $\cat{J}^\op$ to $\Cat$. The resulting category $\int \cat{X}$ has, as objects, pairs of objects $(j,x)$, where $j \in \cat{J}$ and $x \in \cat{X}(j)$, and, as morphisms $(j,x) \to (k,y)$, pairs of morphisms $k \xrightarrow{f} j$ in $\cat{J}$ and $\cat{X}(f)(x) \xrightarrow{\phi} y$ in $\cat{X}(k)$. Moreover, the category $\int \cat{X}$ is equipped with a canonical projection $\int X \to \cat{J}^\op$.

To apply this construction to our setting, let $M$ be a manifold and let $\Sub(M)$ be the poset of embedded, not necessarily open, submanifolds of $M$. Given a suitable category $\cat{S}$, such as $\Set$ or $\Vect_\R$, an indexed category $\Sh_{M,\cat{S}}: \Sub(M)^\op \to \Cat$ is defined on submanifolds $U \subseteq M$ by $\Sh_{M,\cat{S}}(U) := \Sh(U,\cat{S})$ and on inclusions $V \hookrightarrow U$ by restriction of sheaves,
\begin{equation*}
    \Sh_{M,\cat{S}}(V \hookrightarrow U): \sheaf{F} \in \Sh(U,\cat{S}) \mapsto \sheaf{F}|_V \in \Sh(V,\cat{S}).
\end{equation*}
Restriction of a sheaf is a particular case of the inverse image construction \cite[Definition 3.46]{wedhorn2016}. In fact, a greatly enlarged indexed category $\Sh_\cat{S}: \Man^\op \to \Cat$ sends each manifold $M$ to the category of sheaves $\Sh(M,\cat{S})$ and each smooth map $f: M \to N$ to the inverse image functor $f^{-1}: \Sh(N,\cat{S}) \to \Sh(M,\cat{S})$. The morphisms in the resulting category $\int \Sh_{\cat{S}}$ have been called \emph{sheaf cohomomorphisms} \cite[\S I.4]{bredon1997}. For our purposes, it is slightly more convenient to work with the $\Sub(M)$-indexed category $\Sh_{M,\cat{S}}$ than the $\Man$-indexed category $\Sh_{\cat{S}}$.

A good setting in which to formulate boundary value problems is the Grothendieck construction of the indexed category $\Sh_{M,\R} := \Sh_{M,\Vect_\R}$. This category is described explicitly as follows.

\begin{definition}[Category of sheaves on submanifolds] \label{def:category-sheaves-submanifolds}
    For any manifold $M$, the category $\int \Sh_{M,\R}$ has
    \begin{itemize}[nosep]
        \item as objects, an embedded submanifold $U$ of $M$, not necessarily open, together with a sheaf $\sheaf{F} \in \Sh_\R(U)$ of real vector spaces on $U$; and
        \item as morphisms $(U, \sheaf{F}) \to (V, \sheaf{G})$, an inclusion $V \hookrightarrow U$ together with a morphism of sheaves $\sheaf{F}|_V \to \sheaf{G}$.
    \end{itemize}
    Composition and identities are defined as usual in the Grothendieck construction.
\end{definition}

Two classes of morphisms in $\int \Sh_{M,\R}$ will be used routinely. First are the morphisms of form $(\id_U, \phi): (U, \sheaf{F}) \to (U, \sheaf{G})$, where the submanifold $U \subseteq M$ is fixed and $\phi$ is a morphism of sheaves; these are abbreviated as $\phi: \sheaf{F} \to \sheaf{G}$. Very often the map $\phi$ is a linear differential operator, which are characterized within the class of the sheaf morphisms by the \emph{Peetre theorem} \cite{peetre1960,kolar2013}. Second, for any submanifolds $V \subseteq U$ of $M$, pulling back differential $k$-forms along the inclusion defines a morphism $(V \hookrightarrow U, \Omega_U^k\big|_V \to \Omega_V^k): (U, \Omega_U^k) \to (V, \Omega_V^k)$ in $\int \Sh_{M,\R}$, which may be denoted $\res_V^U$ or $\res_V$. Time-dependent and twisted differential forms can be restricted similarly.

As before, generalized elements in the category $\int \Sh_{M,\R}$ correspond to sections of sheaves:

\begin{proposition} \label{prop:elements-of-restricted-sheaves}
    For any manifold $M$, submanifold $U \subseteq M$, and sheaf $\sheaf{F} \in \Sh_\R(U)$, the morphisms $(M,\R_M^\sharp) \to (U,\sheaf{F})$ in $\int \Sh_{M,\R}$ are in natural bijection with sections of $\sheaf{F}$ on $U$.
\end{proposition}
\begin{proof}
    Inverse images send constant sheaves to constant sheaves \cite[Example 3.56]{wedhorn2016}. In particular, the restriction to $U$ of the constant sheaf $\R_M^\sharp$ is the constant sheaf $\R_U^\sharp$. The statement now follows from \cref{prop:elements-of-sheaves}.
\end{proof}

Consequently, differential equations can be cast as lifting problems for diagrams in this category. Lifting a diagram $D: \cat{J} \to \int \Sh_{M,\R}$ to a diagram of generalized elements of shape $(M, \R_M^\sharp)$ amounts to choosing, for each $j \in \cat{J}$, a section
\begin{align*}
    s_j \in \sheaf{F}_j(U_j) &\qquad\text{where}\qquad
    (U_j, \sheaf{F}_j) := Dj
\intertext{such that, for every morphism $f: j \to k$ in $\cat{J}$, we have}
    \phi_{U_k}(s_j|_{U_k}) = s_k &\qquad\text{where}\qquad
    \big(U_k \hookrightarrow U_j, \sheaf{F}_j\big|_{U_k} \xrightarrow{\phi} \sheaf{F}_k\big) := Df.
\end{align*}
This system of equations should be compared with the simpler one for set-theoretic diagrams in \cref{eq:diagram-lift-sets}.

\section{Morphisms of diagrams and boundary value problems} \label{sec:diagram-morphisms}

Much of the value of formalizing diagrams in physics becomes apparent only after introducing a notion of \emph{morphism} between diagrams. It has long been known (even if not widely appreciated) that diagrams in a given category themselves form a category. In fact, they form a category in two different ways, reflecting the duality between limits and colimits.

\begin{definition}[Categories of diagrams] \label{def:diagram-categories}
For any category $\cat{C}$, the category $\DiagOp(\cat{C})$ has as objects, the diagrams $(\cat{J},D)$ in $\cat{C}$ (namely $D:\cat{J}\to\cat{C}$), and as morphisms $(\cat{J},D)\to(\cat{J}',D')$, a functor $R: \cat{J}' \to \cat{J}$ in the backward direction together with a natural transformation $\rho: D \circ R \To D'$.
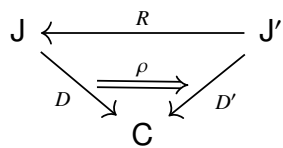
\begin{equation*}
    \begin{tikzcd}
    	{\mathsf{J}} && {\mathsf{J}'} \\
    	& {\mathsf{C}}
    	\arrow["R"', from=1-3, to=1-1]
    	\arrow[""{name=0, anchor=center, inner sep=0}, "D"', from=1-1, to=2-2]
    	\arrow[""{name=1, anchor=center, inner sep=0}, "{D'}", from=1-3, to=2-2]
    	\arrow["\rho", shorten <=6pt, shorten >=6pt, Rightarrow, from=0, to=1]
    \end{tikzcd}
\end{equation*}
Composition of morphisms $(\cat{J}, D) \xrightarrow{(R,\rho)} (\cat{J}',D') \xrightarrow{(S,\tau)} (\cat{J}'', D'')$ in $\DiagOp(\cat{C})$ is defined by the pasting diagram
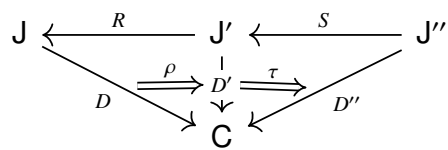
\begin{equation*}
\begin{tikzcd}
	{\mathsf{J}} && {\mathsf{J}'} && {\mathsf{J}''} \\
	&& {\mathsf{C}}
	\arrow["R"', from=1-3, to=1-1]
	\arrow[""{name=0, anchor=center, inner sep=0}, "D"', from=1-1, to=2-3]
	\arrow[""{name=1, anchor=center, inner sep=0}, "{D'}"{description}, from=1-3, to=2-3]
	\arrow["S"', from=1-5, to=1-3]
	\arrow[""{name=2, anchor=center, inner sep=0}, "{D''}", from=1-5, to=2-3]
	\arrow["\rho", shorten <=6pt, shorten >=6pt, Rightarrow, from=0, to=1]
	\arrow["\tau", shorten <=6pt, shorten >=6pt, Rightarrow, from=1, to=2]
\end{tikzcd}
\end{equation*}
and the identity morphism on $(\cat{J}, D)$ is $(\id_{\cat{J}}, \id_D)$.

Another category of diagrams, denoted $\Diag(\cat{C})$, is defined similarly, except that the morphisms from $D: \cat{J} \to \cat{C}$ to $D': \cat{J}' \to \cat{C}$ consist of a functor $R: \cat{J} \to \cat{J}'$ in the forward direction together with a natural transformation $\rho: D \To D' \circ R$.
\begin{equation*}
    \begin{tikzcd}
    	{\mathsf{J}} && {\mathsf{J}'} \\
    	& {\mathsf{C}}
    	\arrow["R", from=1-1, to=1-3]
    	\arrow[""{name=0, anchor=center, inner sep=0}, "D"', from=1-1, to=2-2]
    	\arrow[""{name=1, anchor=center, inner sep=0}, "{D'}", from=1-3, to=2-2]
    	\arrow["\rho", shorten <=6pt, shorten >=6pt, Rightarrow, from=0, to=1]
    \end{tikzcd}
\end{equation*}
There are forgetful functors $\DiagOp(\cat{C}) \to \Cat^\op$ and $\Diag(\cat{C}) \to \Cat$, namely the domain functors, that discard all but the indexing category of a diagram.
\end{definition}

The two categories of diagrams are motivated by their fundamental connection with limits and colimits.
\begin{proposition}[Functorality of limits] \label{prop:functorality-limits}
If $\cat{C}$ is a complete category, endowed with a choice of limit cone for every diagram, then the operation of taking limits extends to a functor
\begin{equation*}
    \lim = \lim_{\leftarrow}: \DiagOp(\cat{C}) \to \cat{C}.
\end{equation*}
Dually, if $\cat{C}$ is a cocomplete category, endowed with a choice of colimits, then the operation of taking colimits extends to a functor
\begin{equation*}
    \colim = \lim_{\rightarrow}: \Diag(\cat{C}) \to \cat{C}.
\end{equation*}
\end{proposition}

\noindent This statement appears already in Eilenberg and Mac Lane's original document about category theory \cite[\S 23]{eilenberg1945}, albeit restricted to ``inverse limits'' and ``direct limits'' since the terminology surrounding limits and colimits had not yet reached its modern form. We sketch the proof in the case of limits because it is useful in its own right.

\begin{proof}
A morphism $(R,\rho): (\cat{J}, D) \to (\cat{J}', D')$ in $\DiagOp(\cat{C})$ sends cones over $D$ to cones over $D'$, preserving the apex: if $\lambda$ is a cone over $D$ with apex $S \in \cat{C}$, then a cone $\lambda'$ over $D'$ with apex $S$ is defined by the pasting diagram:
\begin{equation} \label{eq:diagram-hom-cones}
    \begin{tikzcd}
    	& {\mathsf{J}} && {\mathsf{J}'} \\
    	{\mathsf{1}} && {\mathsf{C}}
    	\arrow[""{name=0, anchor=center, inner sep=0}, "D"{description}, from=1-2, to=2-3]
    	\arrow["R"', from=1-4, to=1-2]
    	\arrow[""{name=1, anchor=center, inner sep=0}, from=1-2, to=2-1]
    	\arrow["S"', from=2-1, to=2-3]
    	\arrow[""{name=2, anchor=center, inner sep=0}, "{D'}", from=1-4, to=2-3]
    	\arrow["\lambda"', shorten <=6pt, shorten >=6pt, Rightarrow, from=1, to=0]
    	\arrow["\rho", shorten <=6pt, shorten >=6pt, Rightarrow, from=0, to=2]
    \end{tikzcd}
    =:
    \begin{tikzcd}
    	& {\mathsf{J}'} \\
    	{\mathsf{1}} && {\mathsf{C}}
    	\arrow[""{name=0, anchor=center, inner sep=0}, from=1-2, to=2-1]
    	\arrow["S"', from=2-1, to=2-3]
    	\arrow[""{name=1, anchor=center, inner sep=0}, "{D'}", from=1-2, to=2-3]
    	\arrow["{\lambda'}"', shorten <=6pt, shorten >=6pt, Rightarrow, from=0, to=1]
    \end{tikzcd}.
\end{equation}
In particular, the morphism $(R,\rho)$ acts on the limit cone $\lambda$ over $D$, having apex $\lim D$, to give a cone $\lambda'$ over $D'$ with apex $\lim D$. The universal property of the limit of $D'$ then yields a canonical morphism $\lim D \to \lim D'$ in $\cat{C}$. This construction defines a functor $\lim: \DiagOp(\cat{C}) \to \cat{C}$, whose functorality follows from the universal property.
\end{proof}

In view of the connection between lifts and cones (\cref{prop:lifts-and-cones}), the preceding \cref{prop:functorality-limits} shows that, of the two categories of diagrams, the category $\DiagOp(\cat{C})$ should be preferred when interpreting diagrams as systems of equations. Specifically, given an object $S \in \cat{C}$, \cref{eq:diagram-hom-cones} in the proof shows that a morphism $D \to D'$ in $\DiagOp(\cat{C})$ pushes forward any lift of $D$ through $\pi: \El_S(\cat{C}) \to \cat{C}$ to a lift of $D'$ through $\pi$. In other words, morphisms in $\DiagOp(\cat{C})$ carry solutions of one system of equations to solutions of another. Morphisms in $\Diag(\cat{C})$ generally do not have this property. Thus, when we speak of the ``category of diagrams'' or a ``morphism of diagrams'' without qualification, we refer to the category $\DiagOp(\cat{C})$.

Certain classes of diagram morphisms are worth singling out. A morphism $(R,\rho)$ in $\Diag(\cat{C})$ or $\DiagOp(\cat{C})$ is called \emph{strong} if the natural transformation $\rho$ is a natural isomorphism and \emph{strict} if it is an identity. Thus, the strict morphisms in $\Diag(\cat{C})$ comprise the slice category $\Cat/\cat{C}$ and the strict morphisms in $\DiagOp(\cat{C})$ comprise the opposite category $(\Cat / \cat{C})^\op$.

The next several examples illustrate how morphisms of diagrams can be used to formalize relations between physical theories presented diagrammatically.

\begin{example}[Static Maxwell-Faraday equations] \label{ex:static-maxwell-faraday}
In the static (time-independent) case, the Maxwell-Faraday equations on a three-dimensional manifold $M$ take the simple form
\begin{equation} \label{eq:maxwell-faraday-static}
    \begin{tikzcd}
    	{E: \Omega^1} & {\Omega^2} & 0 \\
    	& {B: \Omega^2} & {\Omega^3} & 0
    	\arrow["{-d}", from=1-1, to=1-2]
    	\arrow["d", from=2-2, to=2-3]
    	\arrow[from=2-4, to=2-3]
    	\arrow[from=1-3, to=1-2]
    \end{tikzcd}
\end{equation}
as a diagram in $\Sh_\R(M)$. The decoupling of the electric and magnetic fields in the static case is visually apparent in the two components of the diagram. Formally, the diagram is a product in the category of diagrams \cite[\S 2.2]{peschke2020}. (It is not a coproduct, as one might initially expect, simply because of the contravariance in the forgetful functor $\DiagOp(\cat{C}) \to \Cat^\op$.)

If an electric potential $\phi$ and a magnetic potential $A$ exist, then the static Maxwell-Faraday equations can be rewritten as the diagram
\begin{equation} \label{eq:maxwell-faraday-static-potentials}
    \begin{tikzcd}
    	{\phi: \Omega^0} & {E: \Omega^1} & {\Omega^2} \\
    	& {A: \Omega^1} & {B: \Omega^2} & {\Omega^3}
    	\arrow["{-d}", from=1-1, to=1-2]
    	\arrow["{-d}", from=1-2, to=1-3]
    	\arrow["d", from=2-2, to=2-3]
    	\arrow["d", from=2-3, to=2-4]
    \end{tikzcd}
\end{equation}
in $\Sh_\R(M)$. The equations $dE = 0$ and $dB = 0$ need not be explicitly stated since they are implied by the property $d^2 = 0$ of the exterior derivative.

A morphism from the diagram \eqref{eq:maxwell-faraday-static-potentials} including potentials to the diagram \eqref{eq:maxwell-faraday-static} excluding them is specified by the following picture.
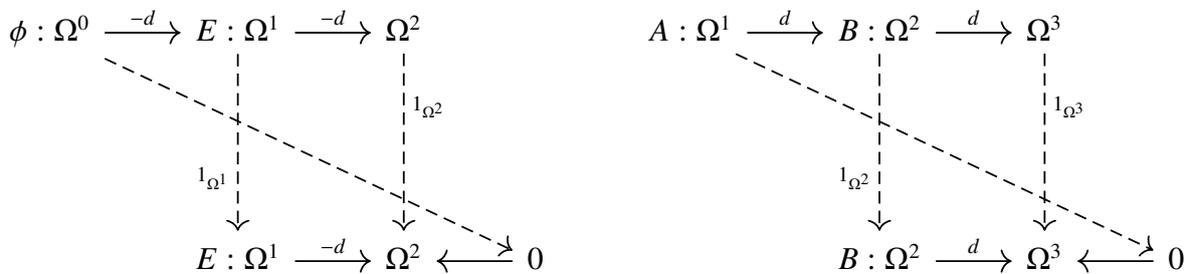
\begin{equation*}
    \begin{tikzcd}
    	{\phi: \Omega^0} & {E: \Omega^1} & {\Omega^2} && {A: \Omega^1} & {B: \Omega^2} & {\Omega^3} \\
    	\\
    	\\
    	& {E: \Omega^1} & {\Omega^2} & 0 && {B: \Omega^2} & {\Omega^3} & 0
    	\arrow["{-d}", from=1-1, to=1-2]
    	\arrow["{-d}", from=1-2, to=1-3]
    	\arrow["d", from=1-5, to=1-6]
    	\arrow["d", from=1-6, to=1-7]
    	\arrow["{-d}", from=4-2, to=4-3]
    	\arrow[from=4-4, to=4-3]
    	\arrow[dashed, from=1-1, to=4-4]
    	\arrow["d", from=4-6, to=4-7]
    	\arrow[from=4-8, to=4-7]
    	\arrow["{1_{\Omega^1}}"'{pos=0.7}, dashed, from=1-2, to=4-2]
    	\arrow["{1_{\Omega^2}}"{pos=0.3}, dashed, from=1-3, to=4-3]
    	\arrow["{1_{\Omega^2}}"'{pos=0.7}, dashed, from=1-6, to=4-6]
    	\arrow["{1_{\Omega^3}}"{pos=0.3}, dashed, from=1-7, to=4-7]
    	\arrow[dashed, from=1-5, to=4-8]
    \end{tikzcd}
\end{equation*}
In drawing the diagram morphism $(R,\rho): (\cat{J}, D) \to (\cat{J}', D')$, the components $\rho_{j'}$, for $j' \in \cat{J}'$, of the natural transformation $\rho: D \circ R \to D'$ are shown as dashed arrows, while the functor $R: \cat{J}' \to \cat{J}$ between indexing categories is not explicitly shown. The object map of $R$ is determined by the dashed arrows of the transformation components. When the indexing category $\cat{J}$ is thin, as in this example, the morphism map of $R$ is determined by the object map. For example, the indexing morphism over $0 \to \Omega^2$ maps to the composite morphism over $\Omega^0 \xrightarrow{-d} \Omega^1 \xrightarrow{-d} \Omega^2$ and the indexing morphism over $0 \to \Omega^3$ maps to the composite morphism over $\Omega^1 \xrightarrow{d} \Omega^2 \xrightarrow{d} \Omega^3$. Crucially, the corresponding naturality squares commute due to the fact that $d^2 = 0$.

It is not a coincidence that the diagram morphism goes from the system with potentials to the system without them, rather than the other way around. The presence of electric and magnetic potentials ensure that $d E = 0$ and $d B = 0$, but those equations do not imply that potentials exist. In mathematical jargon, the diagram \eqref{eq:maxwell-faraday-static} presents a pair of closed forms and the diagram \eqref{eq:maxwell-faraday-static-potentials} presents a pair of exact forms. Exact forms are closed, but whether or not all closed forms are exact depends on the manifold $M$ \cite{baez1994}.
Loosely speaking, the diagram morphism goes in the direction of increasing generality.
\end{example}

\begin{example}[Steady states in diffusion]
Physical intuition suggests that a diffusing substance should approach a steady-state solution in the long-run limit. Steady states of the diffusion equation can be extracted by a morphism of diagrams. To express this rigorously, the mathematical setup in \cref{ex:diffusion} must be altered to ensure that the needed limits exist and are well-behaved.

Let $M$ be a three-dimensional Riemannian manifold and let $\Omega_\infty^k = \Omega_{\infty,M}^k$ be the sheaf of time-dependent $k$-forms of $M$ that have well-defined limits as $t \to \infty$ and whose time derivatives are uniformly continuous in $t$. Note that $\Omega_\infty^k$ is a sheaf on $M$, not on $M \times [0,\infty)$ as usual. Replacing each object $\Omega_t^k$ in diagram \eqref{eq:diffusion-diagram} with $\Omega_\infty^k$ yields a diagram in $\Sh_\R(M)$ for diffusion under the assumption of steady long-run behavior.

On the other hand, the equation for an equilibrium concentration of the substance is given by the diagram
\begin{equation*}
    \begin{tikzcd}
    	& 0 \\
    	{C: \Omega^0} & {\Omega^0} & {d\phi: \widetilde \Omega^3} \\
    	{dC: \Omega^1} & {} & {\phi: \widetilde \Omega^2}
    	\arrow["d"', from=2-1, to=3-1]
    	\arrow["k\star"', from=3-1, to=3-3]
    	\arrow["d"', from=3-3, to=2-3]
    	\arrow["{\star^{-1}}"', from=2-3, to=2-2]
    	\arrow[from=1-2, to=2-2]
    \end{tikzcd}
\end{equation*}
in $\Sh_\R(M)$ and does not involve time at all. When the diffusivity $k \in C^\infty(M)$ is constant throughout $M$, this equation is equivalent to Laplace's equation, $\Lap C = 0$.

A morphism $(R,\rho)$ from the first diagram to the second can be depicted as
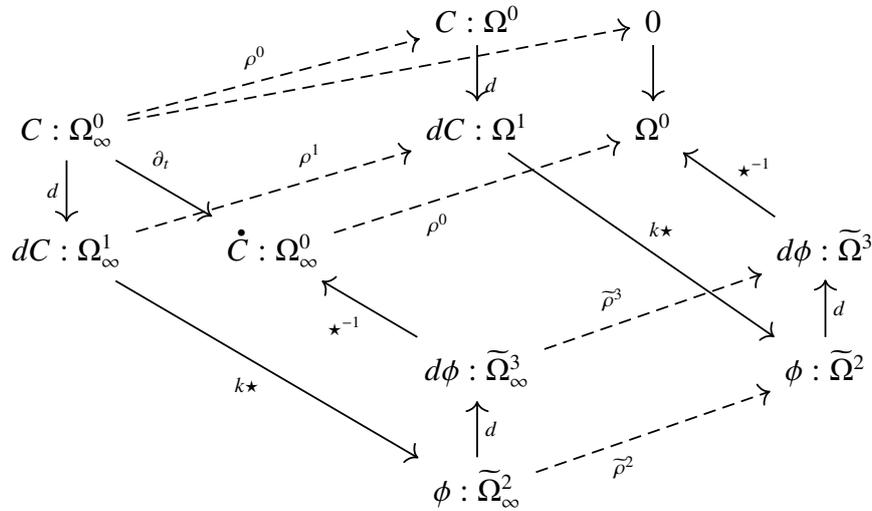
\begin{equation*}
    \begin{tikzcd}
    	&& {C: \Omega^0} & 0 \\
    	{C: \Omega^0_\infty} && {dC: \Omega^1} & {\Omega^0} \\
    	{dC: \Omega^1_\infty} & {\dt C: \Omega^0_\infty} &&& {d\phi: \widetilde\Omega^3} \\
    	&& {d\phi: \widetilde\Omega^3_\infty} && {\phi: \widetilde\Omega^2} \\
    	&& {\phi: \widetilde\Omega^2_\infty}
    	\arrow["d"{pos=0.7}, from=1-3, to=2-3]
    	\arrow["d"', from=4-5, to=3-5]
    	\arrow["{\star^{-1}}"', from=3-5, to=2-4]
    	\arrow["k\star", from=2-3, to=4-5]
    	\arrow["{\star^{-1}}", from=4-3, to=3-2]
    	\arrow["{\rho^0}", dashed, from=2-1, to=1-3]
    	\arrow["d"', from=2-1, to=3-1]
    	\arrow["k\star"', from=3-1, to=5-3]
    	\arrow["d"', from=5-3, to=4-3]
    	\arrow["{\rho^1}"{pos=0.7}, dashed, from=3-1, to=2-3]
    	\arrow["{\widetilde\rho^3}"{pos=0.4}, dashed, from=4-3, to=3-5]
    	\arrow["{\widetilde\rho^2}"'{pos=0.3}, dashed, from=5-3, to=4-5]
    	\arrow["{\rho^0}"'{pos=0.3}, dashed, from=3-2, to=2-4]
    	\arrow["{\partial_t}"{pos=0.3}, from=2-1, to=3-2]
    	\arrow[from=1-4, to=2-4]
    	\arrow[dashed, from=2-1, to=1-4]
    \end{tikzcd}
\end{equation*}
where, as before, the functor $R$ between indexing categories is implicit. The components of the natural transformation $\rho$ are the limit operators $\rho^k := \lim_{t \to \infty}: \Omega_\infty^k \to \Omega^k$ or, in one case, the zero map. The naturality squares corresponding to the spatial operators commute because limits in time commute with both derivatives and Hodge stars in space. More subtly, the naturality square corresponding to the time derivative commutes because if $C \in \Omega^0_\infty$, then $\lim_{t \to \infty} \partial_t C = 0$, by the uniform continuity assumption and Barb{\u{a}}lat's lemma \cite{farkas2016}.
\end{example}

\begin{remark}[Collage of a diagram morphism]
The style of drawing morphisms of diagrams used above can be formalized by a procedure that reduces a diagram morphism to a single diagram encompassing both domain and codomain and the mapping itself. We call the resulting diagram the \emph{collage} of the diagram morphism in analogy with the collage of a profunctor. Let $(R,\rho): (\cat{J}, D) \to (\cat{J}', D')$ be a morphism in $\DiagOp(\cat{C})$. Define its collage $(\cat{K}, F)$ as follows. Letting $\cat{2} := \{0 \xrightarrow{i} 1\}$ be the interval category, take the indexing category $\cat{K}$ to be the pushout
\begin{equation*}
    \begin{tikzcd}
    	{\sf{J}'} & {\sf{J}' \times \sf{2}} \\
    	{\sf{J}} & {\sf{K}}
    	\arrow["R"', from=1-1, to=2-1]
    	\arrow["{(\mathrm{id},0)}", from=1-1, to=1-2]
    	\arrow[from=2-1, to=2-2]
    	\arrow[from=1-2, to=2-2]
    	\arrow["\lrcorner"{anchor=center, pos=0.125, rotate=180}, draw=none, from=2-2, to=1-1]
    \end{tikzcd}
\end{equation*}
in $\Cat$. The naturality equations are already encoded by the commutative diagrams
\begin{equation*}
    \begin{tikzcd}
    	Rj & {(j,0)} && {(j,1)} \\
    	Rk & {(k,0)} && {(k,1)}
    	\arrow[Rightarrow, no head, from=1-1, to=1-2]
    	\arrow["Rf"', from=1-1, to=2-1]
    	\arrow["{(\mathrm{id}_j, i)}", from=1-2, to=1-4]
    	\arrow["{(f,\mathrm{id}_0)}", from=1-2, to=2-2]
    	\arrow[Rightarrow, no head, from=2-1, to=2-2]
    	\arrow["{(f,\mathrm{id}_1)}", from=1-4, to=2-4]
    	\arrow["{(\mathrm{id}_k, i)}", from=2-2, to=2-4]
    \end{tikzcd}
\end{equation*}
in $\cat{K}$, for each morphism $f: j \to k$ in $\cat{J}'$. Next, recall that the natural transformation $\rho: D \circ R \to D'$ is equivalent to a functor $\rho: \cat{J}' \times \cat{2} \to \cat{C}$ such that $\rho(-,0) = D \circ R$ and $\rho(-,1) = D'$. This is the classic interpretation of natural transformations as ``categorical homotopies.'' Thus, the collage diagram $F: \cat{K} \to \cat{C}$ can be defined by the universal property of the pushout:
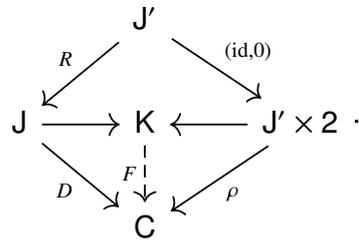
\begin{equation*}
    \begin{tikzcd}
    	& {\sf{J}'} \\
    	{\sf{J}} & {\sf{K}} & {\sf{J}' \times \sf{2}} \\
    	& {\sf{C}}
    	\arrow["R"', from=1-2, to=2-1]
    	\arrow["{(\mathrm{id},0)}", from=1-2, to=2-3]
    	\arrow[from=2-1, to=2-2]
    	\arrow[from=2-3, to=2-2]
    	\arrow["F"', dashed, from=2-2, to=3-2]
    	\arrow["D"', from=2-1, to=3-2]
    	\arrow["\rho", from=2-3, to=3-2]
    \end{tikzcd}.
\end{equation*}
\end{remark}

The next example is simpler than the preceding ones but brings out an important feature of diagram morphisms.

\begin{example}[Variants of the heat equation] \label{ex:heat-equation-morphisms}
We saw in \cref{ex:diffusion} that the heat equation, or the diffusion equation with constant diffusivity, can be presented by at least two different diagrams, namely \eqref{eq:diffusion-diagram} and \eqref{eq:heat-diagram}. As expected, there is a strict morphism of diagrams in one direction
\begin{equation*}
    \left(
    \begin{tikzcd}
    	{C: \Omega_t^0} & {\dt C: \Omega_t^0} & {d\phi: \widetilde \Omega_t^3} \\
    	{dC: \Omega_t^1} & {} & {\phi: \widetilde \Omega_t^2}
    	\arrow["d"', from=1-1, to=2-1]
    	\arrow["{\partial_t}", from=1-1, to=1-2]
    	\arrow["k\star"', from=2-1, to=2-3]
    	\arrow["d"', from=2-3, to=1-3]
    	\arrow["{\star^{-1}}"', from=1-3, to=1-2]
    \end{tikzcd}
    \right)
    \quad\longrightarrow\quad
    \left(
    \begin{tikzcd}
    	{C: \Omega_t^0} & {\dt C: \Omega_t^0}
    	\arrow["{k \Delta}"', shift right=1, from=1-1, to=1-2]
    	\arrow["{\partial_t}", shift left=1, from=1-1, to=1-2]
    \end{tikzcd}
    \right)
\end{equation*}
whose backward map on indexing categories sends the morphism over $k \Lap$ to the composite morphism over
\begin{equation*}
    \Omega_t^0 \xrightarrow{d} \Omega_t^1 \xrightarrow{k \star} \widetilde \Omega_t^2
      \xrightarrow{d} \widetilde \Omega_t^3 \xrightarrow{\star^{-1}} \Omega_t^0.
\end{equation*}
Lifts of the first diagram can therefore be pushed forward to lifts of the second.

A puzzle now presents itself. The two presentations of the heat equation are evidently equivalent inasmuch as any solution of one gives a solution of the other. However, the diagram morphism above is not an isomorphism. Nor is it a 2-categorical equivalence when the category of diagrams assumes its proper status as a 2-category (see \cite{perrone2021} and \cref{rem:2-categories}). So how is this notion of equivalence accomodated by the diagrammatic formalism? The problem is resolved by recognizing that a special class of diagram morphisms, encompassing the one above, are ``weak equivalences'' in that they are not necessarily isomorphisms but still establish a one-to-one correspondence between lifts. Weak equivalences of diagrams are studied in \cref{sec:weak-equivalences}.
\end{example}

\begin{remark}[Notions of theoretical equivalence]
The question of when two physical theories should be considered equivalent is a subtle one that has long been studied by philosophers of science, sometimes even using category-theoretic methods \cite{halvorson2017}. The discussion above takes for granted that two systems of equations should be considered equivalent when they have interchangeable solutions, but the stronger, more syntactical notion of equivalence supplied by the category of diagrams is also meaningful. As a physical theory, the diffusion equation is derived by combining a phenomenological principle (Fick's first law) with a conservation principle, as recalled in \cref{ex:diffusion}. Textbooks on diffusion, such as Crank's \cite{crank1975}, regard Fick's first law as significant enough to be stated on its own, before introducing the diffusion equation. From this perspective, the diagram morphism in \cref{ex:heat-equation-morphisms} discards information about the derivation of the heat equation and may rightfully be considered a non-equivalence.
\end{remark}

{\bf Boundary value problems.}
Another use of diagram morphisms is in formulating boundary value problems. To solve a boundary value problem, one must solve a system of differential equations while simultaneously satisfying a set of side constraints, called the \emph{boundary conditions}. For our purposes, boundary conditions encompass both the spatial and the temporal boundary of the space-time domain, and so include \emph{initial conditions} as a special case. We have seen that solving a system of equations amounts to finding a lift of the diagram presenting the equations (\cref{def:lifting-problem}). Meanwhile, satisfying the boundary conditions requires extending the boundary data to the whole domain. These two problems are combined in the notion of an extension-lifting problem of diagrams.

\begin{definition}[Extension-lifting problem] \label{def:extension-lifting-problem}
An \emph{extension} of a diagram $D_0: \cat{J}_0 \to \cat{C}$ along a functor $R: \cat{J}_0 \to \cat{J}$ is a morphism of diagrams of the form $(R,\rho): (\cat{J}, D) \to (\cat{J}_0, D_0)$, i.e., a diagram $D: \cat{J} \to \cat{C}$ together with a natural transformation $\rho: D \circ R \To D_0$.
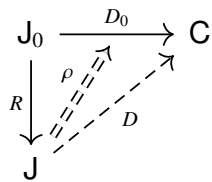
\begin{equation*}
    \begin{tikzcd}[row sep=large, column sep=large]
    	{\mathsf{J}_0} & {\mathsf{C}} \\
    	{\mathsf{J}}
    	\arrow["R"', from=1-1, to=2-1]
    	\arrow[""{name=0, anchor=center, inner sep=0}, "{D_0}", from=1-1, to=1-2]
    	\arrow["D"', dashed, from=2-1, to=1-2]
    	\arrow["\rho", shorten <=4pt, shorten >=4pt, Rightarrow, dashed, from=2-1, to=0]
    \end{tikzcd}
\end{equation*}
Now let $(R,\rho): (\cat{J}, D) \to (\cat{J}_0, D_0)$ be a morphism of diagrams in a category $\cat{C}$ and let $\overline D_0$ be a lift of the diagram $D_0$ through a functor $\pi: \cat{E} \to \cat{C}$. The \emph{extension-lifting problem} associated with this data
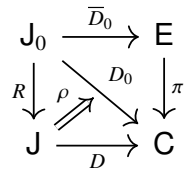
\begin{equation*}
    \begin{tikzcd}
    	{\mathsf{J}_0} & {\mathsf{E}} \\
    	{\mathsf{J}} & {\mathsf{C}}
    	\arrow["R"', from=1-1, to=2-1]
    	\arrow["{\overline D_0}", from=1-1, to=1-2]
    	\arrow["D"', from=2-1, to=2-2]
    	\arrow["\pi", from=1-2, to=2-2]
    	\arrow[""{name=0, anchor=center, inner sep=0}, "{D_0}", from=1-1, to=2-2]
    	\arrow["\rho", shorten >=2pt, Rightarrow, from=2-1, to=0]
    \end{tikzcd}
\end{equation*}
is to find an extension $(R, \overline\rho): (\cat{J}, \overline D) \to (\cat{J}_0, \overline D_0)$ of $\overline D_0$ along $R$ whose domain $\overline D$ is also a lift of $D$ through $\pi$, such that the 2-cells are compatible:
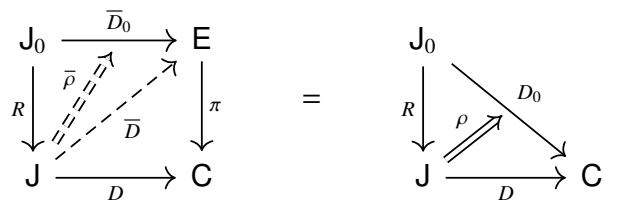
\begin{equation*}
    \begin{tikzcd}[row sep=large, column sep=large]
    	{\mathsf{J}_0} & {\mathsf{E}} \\
    	{\mathsf{J}} & {\mathsf{C}}
    	\arrow["R"', from=1-1, to=2-1]
    	\arrow[""{name=0, anchor=center, inner sep=0}, "{\overline D_0}", from=1-1, to=1-2]
    	\arrow["{\overline D}"', dashed, from=2-1, to=1-2]
    	\arrow["\pi", from=1-2, to=2-2]
    	\arrow["D"', from=2-1, to=2-2]
    	\arrow["\overline\rho", shorten <=4pt, shorten >=4pt, Rightarrow, dashed, from=2-1, to=0]
    \end{tikzcd}
    \qquad=\qquad
    \begin{tikzcd}[row sep=large, column sep=large]
    	{\mathsf{J}_0} \\
    	{\mathsf{J}} & {\mathsf{C}}
    	\arrow["R"', from=1-1, to=2-1]
    	\arrow["D"', from=2-1, to=2-2]
    	\arrow[""{name=0, anchor=center, inner sep=0}, "{D_0}", from=1-1, to=2-2]
    	\arrow["\rho", shorten >=2pt, Rightarrow, from=2-1, to=0]
    \end{tikzcd}.
\end{equation*}
\end{definition}

When presenting physical theories, the diagram $D$ represents the whole system, the diagram $D_0$ represents the boundary of the system, and the morphism $D \to D_0$ projects the system onto its boundary. If this seems backward, note that, geometrically, the boundary is included in the space-time domain, but for physical quantities defined on the domain, the inclusion of spaces induces a projection of quantities contravariantly. As before, the lifting is through the projection functor $\cat{E} := \El(\cat{C}) \to \cat{C}$ associated with a category of generalized elements. Finally, a lift $\overline D_0$ of $D_0$ is a choice of boundary data. The extension-lifting problem can then be interpreted as a boundary value problem.

In algebraic topology, it is standard for extension and extension-lifting problems to be \emph{strict}, meaning that all 2-cells involved must be identities \cite{arkowitz2011,riehl2014}. On the other hand, in category theory, a right Kan extension is a ``lax'' extension, in the sense of \cref{def:extension-lifting-problem}, that satisfies a universal property among all such extensions \cite[Chapter 6]{riehl2016}. As the examples below show, nontrivial 2-cells are essential for boundary value problems, so we allow extension-lifting problems to be non-strict.

\begin{example}[Discrete Dirichlet problem] \label{ex:discrete-dirichlet-problem}
Like the classical heat equation (\cref{ex:discrete-heat-equation}), the classical Dirichlet problem has a discrete analogue on graphs. Let $G = (V,E)$ be a locally finite symmetric weighted graph with no isolated vertices and let $\Lap: \R^{\Omega} \to \R^{\Omega}$ be the discrete Laplacian on $G$. For any subset of vertices $\Omega \subseteq V$, define the \emph{one-step closure} of $\Omega$ to be
\begin{equation*}
    \overline{\Omega} := \{y \in V: x \xrightarrow{e} y \text{ for some } x \in \Omega, e \in E\}
\end{equation*}
and define the \emph{vertex boundary} of $\Omega$ to be $\partial \Omega := \overline{\Omega} \setminus \Omega$. Given boundary data $g \in \R^{\partial \Omega}$, the \emph{discrete Dirichlet problem} is to find $u \in \R^{\overline{\Omega}}$ such that
\begin{equation*}
    \begin{cases}
        \Lap u(x) = 0 & \text{for all $x \in \Omega$} \\
        u(x) = g(x) & \text{for all $x \in \partial\Omega$}.
    \end{cases}
\end{equation*}
It can be shown that, if the graph $G$ is connected and $\Omega$ is finite with non-empty complement $\Omega^c$, then a solution of the problem exists and is unique \cite[Theorem 1.38]{grigoryan2018}.

As an extension-lifting problem, the discrete Dirichlet problem is defined by the morphism of diagrams
\begin{equation*}
    \begin{tikzcd}
    	0 & {\R^{\Omega}} & {u: \mathbb{R}^{\overline{\Omega}}} && {u_b: \R^{\partial \Omega}}
    	\arrow["{\Delta|_{\Omega}}"', from=1-3, to=1-2]
    	\arrow[from=1-1, to=1-2]
    	\arrow["{\mathrm{res}_{\partial \Omega}}", dashed, from=1-3, to=1-5]
    \end{tikzcd},
\end{equation*}
in $\Vect_\R$, whose domain is a cospan and whose codomain is a singleton diagram, together with the lift of the codomain diagram through $\El_\R(\Vect_\R) \to \Vect_\R$ that assigns the boundary data $g \in \R^{\partial\Omega}$ to $u_b$. In the domain diagram, $\Lap|_{\Omega}$ denotes the restriction of the Laplacian to vertices in $\Omega$, which due to boundary effects depends on vertices in $\overline{\Omega}$. In the diagram morphism, the sole component of the 2-cell is the restriction map $\res_{\partial\Omega} = \res_{\partial\Omega}^{\overline{\Omega}}$. Solutions of the extension-lifting problem coincide with solutions of the discrete Dirichlet problem.
\end{example}

For smooth boundary value problems, we work in the category of sheaves on submanifolds of the spatial or space-time domain (\cref{def:category-sheaves-submanifolds}).

\begin{example}[Diffusion with Dirichlet condition]
Continuing \cref{ex:diffusion}, the diffusion equation with \emph{Dirichlet boundary conditions} can be formulated as the extension-lifting problem for the morphism of diagrams
\begin{equation*}
    \begin{tikzcd}
    	{dC: \Omega_t^1(M)} & {C: \Omega_t^0(M)} && {C_0: \Omega^0(M)} \\
    	& {\dt C: \Omega_t^0(M)} && {C_b: \Omega_t^0(\partial M)} \\
    	{\phi: \widetilde\Omega_t^2(M)} & {d\phi: \widetilde\Omega_t^3(M)}
    	\arrow["{\partial_t}", from=1-2, to=2-2]
    	\arrow["d"', from=1-2, to=1-1]
    	\arrow["{k \star}"', from=1-1, to=3-1]
    	\arrow["d"', from=3-1, to=3-2]
    	\arrow["{\star^{-1}}"', from=3-2, to=2-2]
    	\arrow["{\mathrm{res}_{t=0}}", dashed, from=1-2, to=1-4]
    	\arrow["{\mathrm{res}_{\partial M}}"', dashed, from=1-2, to=2-4]
    \end{tikzcd}
\end{equation*}
in the category $\int \Sh_{M \times [0,\infty), \R}$, together with initial and boundary data corresponding to $C_0: \Omega^0(M)$ and $C_b: \Omega_t^0(\partial M)$, which would be written in classical notation as
\begin{equation*}
    \begin{cases}
        C(x,0) = C_0(x) & \text{for all $x \in M$} \\
        C(x,t) = C_b(x,t) & \text{for all $x \in \partial M$ and $t \geq 0$}.
    \end{cases}
\end{equation*}
We abuse notation slightly in abbreviating the objects $(M, \Omega^k_M)$ and $(M \times [0,\infty), \Omega^k_{t,M})$, for example, as $\Omega^k(M)$ and $\Omega^k_t(M)$. The two components $\res_{t=0}$ and $\res_{\partial M}$ are the restriction maps corresponding to the inclusions $M \times \{0\} \hookrightarrow M \times [0,\infty)$ and $\partial M \times [0,\infty) \hookrightarrow M \times [0,\infty)$.

Subtleties lurk behind this seemingly innocent formulation. The boundary data $C_0$ and $C_b$ are not independent of one another since they both specify the concentration at the intersection of the spatial and temporal boundary. Thus, the extension-lifting has a solution only if the compatibility condition $C_0|_{\partial M} = C_b|_{t=0}$ is satisfied. This constraint is incorporated by the extension-lifting problem for the morphism of diagrams
\begin{equation*}
    \begin{tikzcd}
    	{dC: \Omega_t^1(M)} & {C: \Omega_t^0(M)} && {C_0: \Omega^0(M)} & {} & {\Omega^0(\partial M)} \\
    	& {\dt C: \Omega_t^0(M)} && {C_b: \Omega_t^0(\partial M)} \\
    	{\phi: \widetilde\Omega_t^2(M)} & {d\phi: \widetilde\Omega_t^3(M)}
    	\arrow["{\partial_t}", from=1-2, to=2-2]
    	\arrow["d"', from=1-2, to=1-1]
    	\arrow["{k \star}"', from=1-1, to=3-1]
    	\arrow["d"', from=3-1, to=3-2]
    	\arrow["{\star^{-1}}"', from=3-2, to=2-2]
    	\arrow["{\mathrm{res}_{t=0}}"'{pos=0.7}, dashed, from=1-2, to=1-4]
    	\arrow["{\mathrm{res}_{\partial M}}"', dashed, from=1-2, to=2-4]
    	\arrow["{\mathrm{res}_{\partial M, t=0}}", curve={height=-18pt}, dashed, from=1-2, to=1-6]
    	\arrow["{\mathrm{res}_{t=0}}"', from=2-4, to=1-6]
    	\arrow["{\mathrm{res}_{\partial M}}"'{pos=0.3}, from=1-4, to=1-6]
    \end{tikzcd}
\end{equation*}
where a lift of the codomain diagram must satisfy the compatibility condition. The naturality conditions for the diagram morphism are satisfied because the inclusions
\begin{equation*}
    \begin{tikzcd}
      {\partial M \times \{0\}} & {M \times \{0\}} \\
      {\partial M \times [0,\infty)} & {M \times [0,\infty)}
      \arrow[hook, from=1-1, to=1-2]
      \arrow[hook, from=1-1, to=2-1]
      \arrow[hook, from=1-2, to=2-2]
      \arrow[hook, from=2-1, to=2-2]
    \end{tikzcd}
\end{equation*}
induce a commuting diagram of restrictions:
\begin{equation*}
    \begin{tikzcd}
    	{\Omega^k(\partial M)} & {\Omega^k(M)} \\
    	{\Omega_t^k(\partial M)} & {\Omega_t^k(M)}
    	\arrow["{\mathrm{res}_{t=0}}"', from=2-2, to=1-2]
    	\arrow["{\mathrm{res}_{\partial M}}"', from=1-2, to=1-1]
    	\arrow["{\mathrm{res}_{t=0}}", from=2-1, to=1-1]
    	\arrow["{\mathrm{res}_{\partial M}}", from=2-2, to=2-1]
    	\arrow["{\mathrm{res}_{\partial M, t=0}}"{description}, from=2-2, to=1-1]
    \end{tikzcd}.
\end{equation*}

In fact, the condition above is not enough to ensure the existence of a smooth solution near the corner of the space-time domain. For further discussion, see \cite{chen2011} and references therein. We merely note that the diagrammatic formalism can be useful in making assumptions about the boundary data, whatever they may be, clear and explicit.
\end{example}

\begin{example}[Diffusion with Neumann condition]
As a variation on the previous example, the diffusion equation with \emph{Neumann boundary conditions} is the extension-lifting problem for the morphism of diagrams
\begin{equation*}
    \begin{tikzcd}
    	{dC: \Omega_t^1(M)} & {C: \Omega_t^0(M)} && {C_0: \Omega^0(M)} \\
    	& {\dt C: \Omega_t^0(M)} \\
    	{\phi: \widetilde\Omega_t^2(M)} & {d\phi: \widetilde\Omega_t^3(M)} && {\phi_b: \widetilde \Omega_t^2(\partial M)}
    	\arrow["{\partial_t}", from=1-2, to=2-2]
    	\arrow["d"', from=1-2, to=1-1]
    	\arrow["{k \star}"', from=1-1, to=3-1]
    	\arrow["d", from=3-1, to=3-2]
    	\arrow["{\star^{-1}}"', from=3-2, to=2-2]
    	\arrow["{\mathrm{res}_{t=0}}", dashed, from=1-2, to=1-4]
    	\arrow["{\mathrm{res}_{\partial M}}"', curve={height=12pt}, dashed, from=3-1, to=3-4]
    \end{tikzcd}
\end{equation*}
in $\int \Sh_{M \times [0,\infty), \R}$. The negative diffusion flux $\phi$ is specified at the spatial boundary for all time. For example, this boundary condition describes a closed system with impermeable boundary when the flux through the boundary is fixed to be zero. Compatibility conditions between the initial data $C_0$ and the boundary data $\phi_b$ can also be imposed.
\end{example}

\section{Lifting morphisms of diagrams} \label{sec:lifting-properties}

In this section, we temporarily pause the application-driven development to study how morphisms of diagrams in a category $\cat{C}$ interact with lifts through a functor $\pi: \cat{E} \to \cat{C}$. The right level of generality for this study takes the functor $\pi$ to be a discrete opfibration, a concept we now review.

\begin{definition}[Discrete opfibration] \label{def:discrete-opfibration}
A functor $\pi: \cat{E} \to \cat{C}$ is a \emph{discrete opfibration} if, for every morphism $f: x \to y$ in $\cat{C}$ and every object $\overline x \in \cat{E}$ with $\pi(\overline x) = x$, there exists a unique morphism $\overline f: \overline x \to \overline y$ in $\cat{E}$ such that $\pi(\overline f) = f$.
\begin{equation*}
    \begin{tikzcd}[row sep=small]
    	{\overline x} & {} & \bullet \\
    	\\
    	x & {} & y
    	\arrow["f"', from=3-1, to=3-3]
    	\arrow["{\overline f}", dashed, from=1-1, to=1-3]
    	\arrow["\pi"', shorten <=3pt, shorten >=6pt, maps to, from=1-2, to=3-2]
    \end{tikzcd}
\end{equation*}
As the picture suggests, the codomain of the lift $\overline f$ is not given but is part of the existence and uniqueness statement.
\end{definition}

Our motivating example of a discrete opfibration is the codomain projection $\El \cat{C} \to \cat{C}$ associated with a category of generalized elements of $\cat{C}$.

The relevance of discrete opfibrations to extension-lifting problems of diagrams is strongly suggested by the fact that the former definition can be stated in terms of extension-lifting problems. Conflating objects in a category with diagrams of shape $\{\bullet\}$ and morphisms with diagrams of shape $\{0 \to 1\}$, a functor $\pi: \cat{E} \to \cat{C}$ is a discrete opfibration exactly when, for every object $\overline x \in \cat{E}$ and every morphism $f$ in $\cat{C}$ forming a commutative square
\begin{equation*}
    \begin{tikzcd}[row sep=large, column sep=large]
    	{\{\bullet\}} & {\mathsf{E}} \\
    	{\{0 \to 1\}} & {\mathsf{C}}
    	\arrow["0"', from=1-1, to=2-1]
    	\arrow["f"', from=2-1, to=2-2]
    	\arrow["{\overline x}", from=1-1, to=1-2]
    	\arrow["\pi", from=1-2, to=2-2]
    	\arrow["{\overline f}", dashed, from=2-1, to=1-2]
    \end{tikzcd}
\end{equation*}
there exists a unique morphism $\overline f$ in $\cat{E}$ (indicated by the dashed line) making both triangles commute. These are extension-lifing problems, a strict version of \cref{def:extension-lifting-problem}, involving diagrams of shape $\cat{J}_0 := \{\bullet\}$ and $\cat{J} := \{0 \to 1\}$.

Alternatively, discrete opfibrations over $\cat{C}$ can be seen as actions of the category $\cat{C}$ in the following way. An \emph{action} of a category $\cat{C}$, also called a \emph{$\cat{C}$-set} or more commonly a \emph{copresheaf on $\cat{C}$}, is a functor $X: \cat{C} \to \Set$. Informally, a $\cat{C}$-set $X$ is a family of sets $X(c)$, indexed by $c \in \cat{C}$, on which the morphisms of $\cat{C}$ act functorially. Given a $\cat{C}$-set $X$, a discrete opfibration is defined by the projection $\pi_X: \El(X) \to \cat{C}$, where $\El(X)$ is the category of elements of $X$, a special case of the Grothendieck construction (cf.\ \cref{sec:diff-eqs}). The objects of the category $\El(X)$ are pairs $(c, x)$, where $c \in \cat{C}$ and $x \in X(c)$, and morphisms $(c,x) \to (d,y)$ of $\El(X)$ are morphisms $f: c \to d$ of $\cat{C}$ such that $X(f)(x) = y$. Conversely, given a discrete opfibration $\pi: \cat{E} \to \cat{C}$, one can construct a $\cat{C}$-set $X$ defined on objects by $X(c) := \pi^{-1}(c)$. Moreover, the two constructions extend to an equivalence of categories
\begin{equation} \label{eq:dopf-cset-equivalence}
    \DOpf(\cat{C}) \simeq \catSet{C},
\end{equation}
where $\DOpf(\cat{C})$ is the full subcategory of $\Cat/\cat{C}$ spanned by discrete opfibrations and $\catSet{C}$ is the functor category $\Set^{\cat{C}}$ \cite[Proposition 3.9]{spivak2014}.

Under this correspondence, the discrete opfibrations $\cod: \El_S(\cat{C}) \to \cat{C}$ associated with generalized elements of some shape $S \in \cat{C}$ correspond to the \emph{representable} $\cat{C}$-sets, i.e., those of the form $\Hom_{\cat{C}}(S,-): \cat{C} \to \Set$. More general $\cat{C}$-sets can also be useful in defining lifting problems. For example, the recourse to sheafification and the constant sheaf in \cref{prop:elements-of-sheaves} can be avoided by simply taking, for a given manifold $M$, the $\Sh(M)$-set that sends each sheaf on $M$ to its set of global sections.

Lifting properties of diagrams can be stated more economically using the functorality of the diagram category construction. Namely, a functor $\DiagOp: \Cat \to \Cat$ sends each category $\cat{C}$ to the diagram category $\DiagOp(\cat{C})$, as in \cref{def:diagram-categories}, and each functor $F: \cat{C} \to \cat{D}$ to the functor $\DiagOp(F): \DiagOp(\cat{C}) \to \DiagOp(\cat{D})$ that acts on diagrams by post-composition,
\begin{equation*}
    \cat{J} \xrightarrow{D} \cat{C} \qquad\leadsto\qquad
    \cat{J} \xrightarrow{D} \cat{C} \xrightarrow{F} \cat{D},
\end{equation*}
and on diagram morphisms by post-whiskering, sending a morphism $(R,\rho)$ to $(R, F*\rho)$.
\begin{equation*}
    \begin{tikzcd}
    	{\mathsf{J}} && {\mathsf{J}'} \\
    	& {\mathsf{C}}
    	\arrow["R"', from=1-3, to=1-1]
    	\arrow[""{name=0, anchor=center, inner sep=0}, "D"', from=1-1, to=2-2]
    	\arrow[""{name=1, anchor=center, inner sep=0}, "{D'}", from=1-3, to=2-2]
    	\arrow["\rho", shorten <=6pt, shorten >=6pt, Rightarrow, from=0, to=1]
    \end{tikzcd}
    \qquad\leadsto\qquad
    \begin{tikzcd}
    	{\mathsf{J}} && {\mathsf{J}'} \\
    	& {\mathsf{C}} \\
    	& {\mathsf{D}}
    	\arrow["R"', from=1-3, to=1-1]
    	\arrow[""{name=0, anchor=center, inner sep=0}, "D"', from=1-1, to=2-2]
    	\arrow[""{name=1, anchor=center, inner sep=0}, "{D'}", from=1-3, to=2-2]
    	\arrow["F"', from=2-2, to=3-2]
    	\arrow["\rho", shorten <=6pt, shorten >=6pt, Rightarrow, from=0, to=1]
    \end{tikzcd}
\end{equation*}
A functor $\Diag: \Cat \to \Cat$ is defined in the same way. In fact, these two endofunctors on $\Cat$ extend to strict 2-functors, and even to pseudomonads on $\Cat$ \cite{kock1967,guitart1973,perrone2021}, although we will not use this fact here.

With these definitions, the previous \cref{def:lifting-problem,def:extension-lifting-problem} can be rephrased as:
\begin{itemize}
\item Given a diagram $D$ in a category $\cat{C}$ and a functor $\pi: \cat{E} \to \cat{C}$, the \emph{lifting problem} is to find a diagram $\overline D$ in $\cat{E}$ such that $\DiagOp(\pi)(\overline D) = D$.
\item Given a morphism of diagrams $(R, \rho): D \to D_0$ in $\cat{C}$ and a lift $\overline D_0$ of $D_0$ through $\pi$, the \emph{extension-lifting problem} is to find a morphism of diagrams $(R, \overline\rho): \overline D \to \overline D_0$ such that $\DiagOp(\pi)(R,\overline\rho) = (R,\rho)$.
\end{itemize}
We can also succinctly state a generalization of the result \eqref{eq:diagram-hom-cones} that diagram morphisms carry lifts of their domain to lifts of their codomain.

\begin{theorem}[Lifting morphisms of diagrams] \label{thm:diagram-opfibration}
If $\pi: \cat{E} \to \cat{C}$ is a discrete opfibration, then so too is the functor $\DiagOp(\pi): \DiagOp(\cat{E}) \to \DiagOp(\cat{C})$.
\end{theorem}
\begin{proof}
We must show that for any morphism of diagrams $(R, \rho): (\cat{J}, D) \to (\cat{J}', D')$ and any lift $\overline D$ of $D$ through $\pi$, there exists a unique morphism $(R, \overline\rho): (\cat{J}, \overline D) \to (\cat{J}', \overline D')$ such that $\DiagOp(\pi)(R, \overline\rho) = (R, \rho)$.
\begin{equation*}
    \begin{tikzcd}[row sep=large, column sep=large]
    	{\mathsf{J}} & {\mathsf{E}} \\
    	{\mathsf{J}'} & {\mathsf{C}}
    	\arrow["{\overline D}", from=1-1, to=1-2]
    	\arrow["\pi", from=1-2, to=2-2]
    	\arrow["{D'}"', from=2-1, to=2-2]
    	\arrow["R", from=2-1, to=1-1]
    	\arrow[""{name=0, anchor=center, inner sep=0}, "{\overline D'}"', dashed, from=2-1, to=1-2]
    	\arrow["{\overline \rho}", shorten >=2pt, Rightarrow, dashed, from=1-1, to=0]
    \end{tikzcd}
    \qquad=\qquad
    \begin{tikzcd}
    	{\mathsf{J}} \\
    	{\mathsf{J}'} & {\mathsf{C}}
    	\arrow["R", from=2-1, to=1-1]
    	\arrow[""{name=0, anchor=center, inner sep=0}, "D", from=1-1, to=2-2]
    	\arrow[""{name=1, anchor=center, inner sep=0}, "{D'}"', from=2-1, to=2-2]
    	\arrow["\rho"', shorten <=2pt, shorten >=2pt, Rightarrow, from=0, to=1]
    \end{tikzcd}
\end{equation*}
Since $\pi : \cat{E} \to \cat{C}$ is a discrete opfibration, for each $j \in \cat{J}'$, the object $\overline{D} R j \in \cat{E}$ along with the morphism 
\begin{equation*}
    \pi (\overline{D} R j) = D R j \xrightarrow{\rho_j} D' j
\end{equation*}
in $\cat{C}$ has a unique lift through $\pi$ to a morphism $\overline{D} R j \xrightarrow{\overline \rho_j} e_j$ in $\cat{E}$. Thus, we must define $\overline D': \cat{J}' \to \cat{E}$ on objects by $\overline D' j := e_j$.
Next, for every morphism $j \xrightarrow{f} k$ in $\cat{J}'$, the object $\overline D' j \in \cat{E}$ along with the morphism
\begin{equation*}
    \pi(\overline D' j) = D' j \xrightarrow{D' f} D' k
\end{equation*}
in $\cat{C}$ has a unique lift through $\pi$ to a morphism $\overline D' j \xrightarrow{\overline f} \bullet$ in $\cat{E}$. Now, the (possibly ill-formed) square
\begin{equation*}
    \begin{tikzcd}
    	{\overline{D} Rj} & {\overline{D} Rk} \\
    	{\overline{D}' j} & {\bullet \overset{?}{=} \overline{D}' k}
    	\arrow["{\overline\rho_j}"', from=1-1, to=2-1]
    	\arrow["{\overline f}", from=2-1, to=2-2]
    	\arrow["{\overline{D} Rf}", from=1-1, to=1-2]
    	\arrow["{\overline\rho_k}", shift left=3, from=1-2, to=2-2]
    \end{tikzcd}
\end{equation*}
in $\cat{E}$ projects under $\pi$ to the square
\begin{equation*}
    \begin{tikzcd}
    	DRj & DRk \\
    	{D'j} & {D'k}
    	\arrow["{\rho_k}"', from=1-1, to=2-1]
    	\arrow["DRf", from=1-1, to=1-2]
    	\arrow["{\rho_k}", from=1-2, to=2-2]
    	\arrow["{D'f}", from=2-1, to=2-2]
    \end{tikzcd}
\end{equation*}
in $\cat{C}$, which commutes by the naturality of $\rho$. From the uniqueness of lifts through $\pi$, it follows that the codomain of $\overline f$ is $\overline D' k$ and also that the square in $\cat{E}$ commutes. In particular, we can and must define $\overline D'$ on morphisms by $\overline D' f := \overline f$.

We have shown that the data of the diagram $\overline D': \cat{J}' \to \cat{E}$ and the transformation $\overline\rho: \overline D \circ R \to \overline D'$ exist and are unique. The argument already shows that $\overline\rho$ is natural, while the functorality of $\overline D'$ follows from the uniqueness of lifts through $\pi$.
\end{proof}

Given a $\cat{C}$-set $X$, lifts of diagrams and diagram morphisms through the discrete opfibration $\pi_X: \El(X) \to \cat{C}$ can also be expressed directly in terms of $X$. Let $1 = \{*\}$ be the terminal set and denote the functor $\cat{J} \xrightarrow{!} \cat{1} \xrightarrow{1} \Set$ by $1 = 1^{\cat{J}}$ also, since it is the terminal object in $\catSet{J}$. Lifts of a diagram $D: \cat{J} \to \cat{C}$ through $\pi_X$ correspond to cones $\lambda$ over $X \circ D: \cat{J} \to \Set$ with apex $1$, or equivalently to natural transformations of form
\begin{equation*}
    \begin{tikzcd}
    	& {\mathsf{J}} \\
    	{\mathsf{Set}} && {\mathsf{C}}
    	\arrow[""{name=0, anchor=center, inner sep=0}, "1"', from=1-2, to=2-1]
    	\arrow[""{name=1, anchor=center, inner sep=0}, "D", from=1-2, to=2-3]
    	\arrow["X", from=2-3, to=2-1]
    	\arrow["\lambda"', shorten <=6pt, shorten >=6pt, Rightarrow, from=0, to=1]
    \end{tikzcd}.
\end{equation*}
Moreover, the action on lifts by a diagram morphism $(R, \rho): (\cat{J}, D) \to (\cat{J}', D')$ in \cref{thm:diagram-opfibration} takes the form:
\begin{equation}
    \begin{tikzcd}
    	& {\mathsf{J}} && {\mathsf{J}'} \\
    	{\mathsf{Set}} && {\mathsf{C}}
    	\arrow[""{name=0, anchor=center, inner sep=0}, "D"{description}, from=1-2, to=2-3]
    	\arrow["X", from=2-3, to=2-1]
    	\arrow["R"', from=1-4, to=1-2]
    	\arrow[""{name=1, anchor=center, inner sep=0}, "{D'}", from=1-4, to=2-3]
    	\arrow[""{name=2, anchor=center, inner sep=0}, "1"', from=1-2, to=2-1]
    	\arrow["\rho", shorten <=6pt, shorten >=6pt, Rightarrow, from=0, to=1]
    	\arrow["\lambda"', shorten <=6pt, shorten >=6pt, Rightarrow, from=2, to=0]
    \end{tikzcd}
    \quad=:\quad
    \begin{tikzcd}
    	& {\mathsf{J}'} \\
    	{\mathsf{Set}} && {\mathsf{C}}
    	\arrow[""{name=0, anchor=center, inner sep=0}, "1"', from=1-2, to=2-1]
    	\arrow[""{name=1, anchor=center, inner sep=0}, "{D'}", from=1-2, to=2-3]
    	\arrow["X", from=2-3, to=2-1]
    	\arrow["{\lambda'}"', shorten <=6pt, shorten >=6pt, Rightarrow, from=0, to=1]
    \end{tikzcd}.
\end{equation}
This equation should be compared with \cref{eq:diagram-hom-cones}.

In summary, we have seen, from the perspectives of both discrete opfibrations and category actions, how lifts of diagrams in a category $\cat{C}$ are pushed forward by morphisms in $\DiagOp(\cat{C})$. Importantly, morphisms in $\Diag(\cat{C})$ generally do \emph{not} have this property, which is why we regard the category $\DiagOp(\cat{C})$ as primary in this work. In \cref{sec:weak-equivalences}, we explore special conditions under which morphisms in $\Diag(\cat{C})$ push forward lifts of diagrams.

\section{Diagrams involving cartesian products} \label{sec:cartesian-diagrams}

The physical theories considered so far, such as the diffusion and wave equations, have been fairly simple. That is no accident. Although the basic vocabulary of category theory---categories, functors, and natural transformations---goes surprisingly far, more structure is needed to express richer physical theories. In this section and the next, we extend the diagrammatic formalism from bare categories to categories with extra structure, beginning with cartesian products, i.e., products in the categorical sense.

The objects of interest will now be cartesian diagrams in a cartesian category. For us, a \emph{cartesian category} is a category with finite products. Be warned that some authors (notably Johnstone \cite{johnstone2002}) use ``cartesian category'' to mean a category having all finite limits. A \emph{cartesian functor} between cartesian categories $\cat{C}$ and $\cat{D}$ is a functor $F: \cat{C} \to \cat{D}$ that preserves finite products, in the sense that if $\prod_{i \in I} x_i$ is a finite product of objects $x_i \in \cat{C}$, $i \in I$, with projections $\pi_i: \prod_{i \in I} x_i \to x_i$, then the images
\begin{equation*}
    F(\pi_i): F\Big(\prod_{i \in I} x_i\Big) \to F(x_i), \qquad
    i \in I,
\end{equation*}
constitute a product of the objects $F(x_i) \in \cat{D}$, $i \in I$. Natural transformations between cartesian functors require no extra conditions, as they automatically preserve products. To see this, suppose $\alpha: F \To G$ is a natural transformation between cartesian functors $F, G: \cat{C} \to \cat{D}$. Then, in the case of a binary product of $x,y \in \cat{C}$, we have
\begin{equation*}
    \alpha_{x \times y}
    = (\alpha_{x \times y} \cdot G(\pi_x), \alpha_{x \times y} \cdot G(\pi_y))
    = (F(\pi_x) \cdot \alpha_x, F(\pi_y) \cdot \alpha_y)
    = \alpha_x \times \alpha_y,
\end{equation*}
where the middle equation follows from the naturality of the transformation $\alpha$ applied to the projections $\pi_x: x \times y \to x$ and $\pi_y: x \times y \to y$ in $\cat{C}$.

The resulting terminology for diagrams is not standard but should be self-explanatory.

\begin{definition}[Cartesian diagram] \label{def:cartesian-diagram}
A \emph{cartesian diagram} in a cartesian category $\cat{C}$ is a cartesian functor $D: \cat{J} \to \cat{C}$, where the indexing category $\cat{J}$ is a small cartesian category.
\end{definition}
\noindent Similarly, morphisms of cartesian diagrams are defined as in \cref{def:diagram-categories}, replacing ``category'' with ``cartesian category'' and ``functor'' with ``cartesian functor.''

\begin{remark}[Diagrams, 2-categorically] \label{rem:2-categories}
As this replacement procedure suggests, the present paper's most fundamental definitions involve only \emph{formal category theory} in that they do not depend on the specifics of categories, functors, and natural transformations, but only on the fact that these form a 2-category. For example, given a 2-category $\catTwo{K}$, the 2-category $\Diag(\cat{C})$ of diagrams in $\cat{C} \in \catTwo{K}$ can be defined as the lax slice 2-category $\catTwo{K} \nearrow \cat{C}$; similarly, the 2-category $\DiagOp(\cat{C})$ can be defined as the opposite of the colax slice 2-category $\catTwo{K} \swarrow \cat{C}$, using Johnstone's notation \cite{johnstone1993}. Taking $\catTwo{K} = \Cat$ recovers the classical notion of diagrams, while taking $\catTwo{K} = \Cart$, the 2-category of cartesian categories, cartesian functors, and natural transformations, yields the setting of this section. At the expense of introducing mild logical redundancy, we deemphasize the 2-categorical perspective in order to avoid unnecessary formalism. More importantly, the theory in \cref{sec:lifting-properties} \emph{does} rely on the specifics of the 2-category $\Cat$ and so does not extend to other 2-categories without further argument.
\end{remark}

Any category of generalized elements of a cartesian category is itself a cartesian category. Given a cartesian category $\cat{C}$, the product of generalized elements $x_i: S \to X_i$, $i \in I$, in $\El_S(\cat{C})$ is the pairing $(x_i)_{i \in I}: S \to \prod_{i \in I} X_i$ obtained by the universal property of the product in $\cat{C}$. The projections $\pi_i: (x_i)_{i \in I} \to x_i$ in $\El_S(\cat{C})$ are simply the projections $\pi_i: \prod_{i \in I} X_i \to X_i$ in $\cat{C}$. Moreover, the codomain projection $\pi := \cod: \El_S(\cat{C}) \to \cat{C}$ is clearly a cartesian functor. One can therefore consider lifting problems and extension-lifting problems of cartesian diagrams through the cartesian functor $\pi$. The foundational results concerning extension-lifting problems from \cref{sec:lifting-properties} will be verified for cartesian diagrams at the end of this section. For now, we take them for granted in presenting examples of differential equations as lifting problems of cartesian diagrams.

In categories involving vector spaces, (finite) products are typically direct sums of vector spaces, which are also set-theoretic products. In particular, this is the case in the prototypical such category, $\Vect_\R$. In the category $\Sh_\R(M)$ of sheaves of vector spaces on a manifold $M$, products are constructed by taking direct sums of vector spaces locally, for each open subspace of $M$ \cite[\S 3.6]{wedhorn2016}.

\begin{example}[Diffusion with source] \label{ex:diffusion-with-source}
Extending \cref{ex:diffusion}, suppose that a substance is diffusing on a three-dimensional Riemannian manifold $M$ that has an external source or sink, which adds or removes substance at an instantaneous rate $S \in \Omega_t^0(M)$. The system is described by the \emph{nonhomogeneous diffusion equation},
\begin{equation*}
    \partial_t C = \star^{-1} d \phi + S,
\end{equation*}
where, as before, $\phi = k \star dC$ is the negative diffusion flux. Alternatively, the system is presented by the cartesian diagram
\begin{equation} \label{eq:diffusion-with-source-diagram}
    \begin{tikzcd}
    	& {(\dt C_\mathrm{flux}, S): (\Omega_t^0)^2} & {S: \Omega_t^0} \\
    	{C: \Omega_t^0} & {\dt C: \Omega_t^0} & {\dt C_{\mathrm{flux}}: \Omega_t^0} & {d\phi: \widetilde\Omega_t^3} \\
    	{dC: \Omega_t^1} &&& {\phi: \widetilde\Omega_t^2}
    	\arrow["{\pi_2}", from=1-2, to=1-3]
    	\arrow["{\pi_1}"', from=1-2, to=2-3]
    	\arrow["{+}"', from=1-2, to=2-2]
    	\arrow["{\partial_t}", from=2-1, to=2-2]
    	\arrow["d"', from=2-1, to=3-1]
    	\arrow["{\star^{-1}}"', from=2-4, to=2-3]
    	\arrow["d"', from=3-4, to=2-4]
    	\arrow["{k \star}"', from=3-1, to=3-4]
    \end{tikzcd}
\end{equation}
in $\Sh_\R(M \times [0,\infty))$, where $\dt C_{\mathrm{flux}}$ denotes the contribution to the total time derivative of $C$ by the flux of substance due to diffusion. While more spacious than the corresponding equation, the diagram highlights the conceptual roles played by the different parts of the equation by decomposing it into constituents like $\dt C$ and $\dt C_{\mathrm{flux}}$.

Any solution of the (homogoneous) diffusion equation gives a solution of the nonhomogeneous diffusion equation for which the source $S$ is identically zero. Encoding this fact, there is a morphism  from the diagram \eqref{eq:diffusion-diagram}, now construed as a cartesian diagram, to the cartesian diagram \eqref{eq:diffusion-with-source-diagram} above, which is the identity except as shown below.
\begin{equation*}
    \begin{tikzcd}
    	0 && {\dt C_{\mathrm{flux}}: \Omega_t^0} && {S: \Omega_t^0} \\
    	{(\dt C,0): \Omega_t^0 \oplus 0} &&& {(\dt C_{\mathrm{flux}}, S): \Omega_t^0 \oplus \Omega_t^0} \\
    	{\dt C: \Omega_t^0} &&& {\dt C: \Omega_t^0}
    	\arrow["{+}", from=2-4, to=3-4]
    	\arrow["{\pi_1}"', from=2-4, to=1-3]
    	\arrow["{\pi_2}", from=2-4, to=1-5]
    	\arrow["{\pi_1}"', dotted, from=2-1, to=3-1]
    	\arrow["{\pi_2}", dotted, from=2-1, to=1-1]
    	\arrow["{\mathrm{id}}", dashed, from=3-1, to=3-4]
    	\arrow["{\pi_1 \cdot \iota_1}"'{pos=0.7}, dashed, from=2-1, to=2-4]
    	\arrow["{\mathrm{id}}"'{pos=0.3}, dashed, from=3-1, to=1-3]
    	\arrow[curve={height=-18pt}, dashed, from=1-1, to=1-5]
    \end{tikzcd}
\end{equation*}
In this picture, the product on the left side, indicated by dotted lines, is implicitly part of the domain diagram because the diagram is cartesian, even if it would not usually be drawn. The map $\iota_1: \Omega_t^0 \to \Omega_t^0 \oplus \Omega_t^0$ is the inclusion $x \mapsto (x,0)$ associated with the direct sum. The transformation involved in the diagram morphism is natural, due to equations such as
\begin{equation*}
    \begin{tikzcd}
    	{\Omega_t^0} & {\Omega_t^0 \oplus \Omega_t^0} \\
    	& {\Omega_t^0}
    	\arrow["{\iota_1}", from=1-1, to=1-2]
    	\arrow["{\pi_1 \text{ or } +}", from=1-2, to=2-2]
    	\arrow["{\mathrm{id}}"', from=1-1, to=2-2]
    \end{tikzcd}
    \qquad\text{and}\qquad
    \begin{tikzcd}
    	{\Omega_t^0} & {\Omega_t^0 \oplus \Omega_t^0} \\
    	& {\Omega_t^0}
    	\arrow["{\iota_1}", from=1-1, to=1-2]
    	\arrow["{\pi_2}", from=1-2, to=2-2]
    	\arrow["0"', from=1-1, to=2-2]
    \end{tikzcd}.
\end{equation*}
\end{example}

Building on the conventions in \cref{sec:diagrams} for drawing diagrams, type-theoretic notation is used to distinguish between those objects in the indexing category of a cartesian diagram that are products and those that merely map to products. Consider the two cartesian diagrams
\begin{equation*}
    x: X \xleftarrow{\pi_X} (x,y): X \times Y \xrightarrow{\pi_Y} y: Y
    \qquad\text{and}\qquad
    x: X \xleftarrow{\pi_X} z: X \times Y \xrightarrow{\pi_Y} y: Y.
\end{equation*}
Both diagrams map to the product cone $X \xleftarrow{\pi_X} X \times Y \xrightarrow{\pi_Y} Y$, but their indexing categories are respectively the free cartesian category on two objects and the free cartesian category containing a span, i.e.,
\begin{equation*}
    j_1 \xleftarrow{\pi_1} j_1 \times j_2 \xrightarrow{\pi_2} j_2
    \qquad\text{and}\qquad
    j_1 \leftarrow j_0 \rightarrow j_2.
\end{equation*}
Thus, the first diagram specifies the product in the indexing category whereas the second does not. When a product is intended, it is advisable to include it in the indexing category since that ensures, for example, that any lift of a cartesian diagram through a functor $\pi$ will have the intended products, irrespective of the properties of $\pi$ (but see also \cref{lem:reflect-products}). Failing to do so can also have subtle but significant consequences for \emph{morphisms} of the diagram, as the next example shows.

\begin{example}[Maxwell-Faraday equations: dynamic and static]
Any solution of the static Maxwell-Faraday equations (\cref{ex:static-maxwell-faraday}) can be construed as a solution of the dynamic Maxwell-Faraday equations (\cref{ex:maxwell-faraday}) that happens to be constant in time. One would expect this to define a morphism of diagrams from the static equations to the dynamic ones. This is indeed the case \emph{if the diagrams are taken to be cartesian}, with terminal objects specified in the indexing category.

Writing $\rho^k: \Omega^k \to \Omega_t^k$ to mean the operation that sends $k$-forms to time-dependent $k$-forms that are constant in time, there is a morphism of cartesian diagrams
\begin{equation*}
    \begin{tikzcd}
    	{E: \Omega^1} &&& {E: \Omega_t^1} \\
    	{\Omega^2} & {B: \Omega^2} && {\dt B: \Omega_t^2} & {B:\Omega_t^2} \\
    	0 & {\Omega^3} &&& {\Omega_t^3} \\
    	& 0 &&& 0
    	\arrow["{-d}"', from=1-1, to=2-1]
    	\arrow[from=3-1, to=2-1]
    	\arrow["d"', from=2-2, to=3-2]
    	\arrow[dotted, from=2-2, to=3-1]
    	\arrow[from=4-2, to=3-2]
    	\arrow["{-d}", from=1-4, to=2-4]
    	\arrow["{\partial_t}"', from=2-5, to=2-4]
    	\arrow["d", from=2-5, to=3-5]
    	\arrow[from=4-5, to=3-5]
    	\arrow["{\rho^1}", dashed, from=1-1, to=1-4]
    	\arrow["{\rho^2}", curve={height=-12pt}, dashed, from=2-1, to=2-4]
    	\arrow["{\rho^2}"', curve={height=12pt}, dashed, from=2-2, to=2-5]
    	\arrow["{\rho^3}"', dashed, from=3-2, to=3-5]
    	\arrow[dashed, from=4-2, to=4-5]
    \end{tikzcd}
\end{equation*}
in $\Sh_\R(M)$, where the morphism indexing the time derivative $\partial_t$ on the right side is sent to the composite indexing the map $\Omega^2 \xrightarrow{!} 0 \xrightarrow{0} \Omega^2$ on the left side. The indexing morphism into the terminal object, indicated by the dotted arrow, would not usually be drawn at all, but it is implicitly contained in the indexing category because this category is cartesian. The corresponding naturality equation says that the time derivative of a time-dependent form that is constant in time is zero. The other naturality equations hold because spatial derivatives commute with the operations $\rho^k$.

This example shows that cartesian diagrams should be preferred to bare diagrams when using product structure, even when doing so in seemingly minor ways.
\end{example}

The use of products to specify sums, as in \cref{ex:diffusion-with-source}, is common enough to merit its own abbreviated notation: for any integer $n \geq 0$, the cartesian diagram
\begin{equation*}
    \begin{tikzcd}[row sep=tiny,every arrow/.append style={touch src}]
    	{x_1: X} \\
    	\vdots & {\left(x_i\right)_{i=1}^n: X^n} & {x:X} \\
    	{x_n: X}
    	\arrow["{\pi_1}"', from=2-2, to=1-1]
    	\arrow["{\pi_n}", from=2-2, to=3-1]
    	\arrow["{+}", from=2-2, to=2-3]
    \end{tikzcd}
\end{equation*}
is abbreviated as
\begin{equation*}
    \begin{tikzcd}[row sep=tiny,every arrow/.append style={touch src}]
    	{x_1: X} \\
    	\vdots & \bigcirc & {x:X} \\
    	{x_n: X}
    	\arrow[from=2-2, to=1-1]
    	\arrow[from=2-2, to=3-1]
    	\arrow["{+}", from=2-2, to=2-3]
    \end{tikzcd},
\end{equation*}
where it is not necessary to label the projections since addition is commutative. For readers familiar with string diagrams, the notation is deliberately reminiscent of the depiction of commutative monoid objects in symmetric monoidal categories \cite[\S 6]{selinger2010}, but note the reversal of the arrows corresponding to the projections.

With this notation, Maxwell's equations in matter admit the following presentation.

\begin{example}[Maxwell's house] \label{ex:maxwell-house}
Let $M$ be a three-dimensional Riemannian manifold. Phrased in exterior calculus and with units such that the speed of light is 1, \emph{Maxwell's equations} on the spatial domain $M$ are
\begin{align*}
    d E &= - \partial_t B \\
    d B &= 0 \\
    \delta E &= \rho \\
    \delta B &= \partial_t E + J,
\end{align*}
where $E \in \Omega_t^1(M)$ is the electric field, $B \in \Omega_t^2(M)$ is the magnetic field, $\rho: \Omega_t^0(M)$ is the electric charge density, and $J \in \Omega_t^1(M)$ is the current density \cite{baez1994}. As before, the operator $\delta := \star^{-1} \circ d \circ \star$ is the codifferential. If we assume the existence of an electric potential $\phi \in \Omega_t^0(M)$ and a magnetic potential $A \in \Omega_t^1(M)$, we can add the two further equations
\begin{equation} \label{eq:maxwell-potentials}
    \begin{aligned}
        E &= -d\phi - \partial_t A \\
        B &= d A.
    \end{aligned}
\end{equation}
Since $d^2 = 0$, these latter two equations imply the first half of Maxwell's equations, the Maxwell-Faraday equations (\cref{ex:maxwell-faraday}).

When studying electromagnetism in a material medium, it is conventional to introduce the electric displacement $\widetilde D \in \Omega_t^2(M)$ and magnetic intensity $\widetilde H \in \Omega_t^1(M)$, related to the electric and magnetic fields by the constitutive relations
\begin{equation} \label{eq:maxwell-constitutive}
    \widetilde D := \epsilon \star E
    \qquad\text{and}\qquad
    \widetilde H := \frac{1}{\mu} \star B,
\end{equation}
where the positive constants $\epsilon, \mu \in C^\infty(M)$ are properties of the medium. The second half of Maxwell's equation are then replaced by
\begin{equation} \label{eq:maxwell-matter}
    \begin{aligned}
        d \widetilde{D} &= \widetilde \rho \\
        d \widetilde{H} &= \partial_t \widetilde{D} + \widetilde J,
    \end{aligned}
\end{equation}
where $\widetilde\rho \in \widetilde\Omega_t^3(M)$ and $\widetilde J \in \widetilde\Omega_t^2(M)$ are twisted proxies for the charge and current densities. Finally, \emph{Ohm's law}
\begin{equation} \label{eq:ohm}
    \widetilde J = \sigma \star E,
\end{equation}
where the conductivity $\sigma \in C^\infty(M)$ is another property of the medium, is an equation independent of Maxwell's that is often added.

A cartesian diagram for Maxwell's equations in matter, assuming Ohm's law and the existence of an electric and magnetic potential, is shown in \cref{fig:maxwell-house}. It faithfully encodes all of \cref{eq:maxwell-potentials,eq:maxwell-constitutive,eq:maxwell-matter,eq:ohm}. It fulfills the promise of the Introduction to give a rigorous interpretation to Maxwell's house, which was shown in its traditional, informal style in \cref{fig:maxwell-house-traditional}. The inconsistency pointed out in the Introduction is avoided through the use of products.

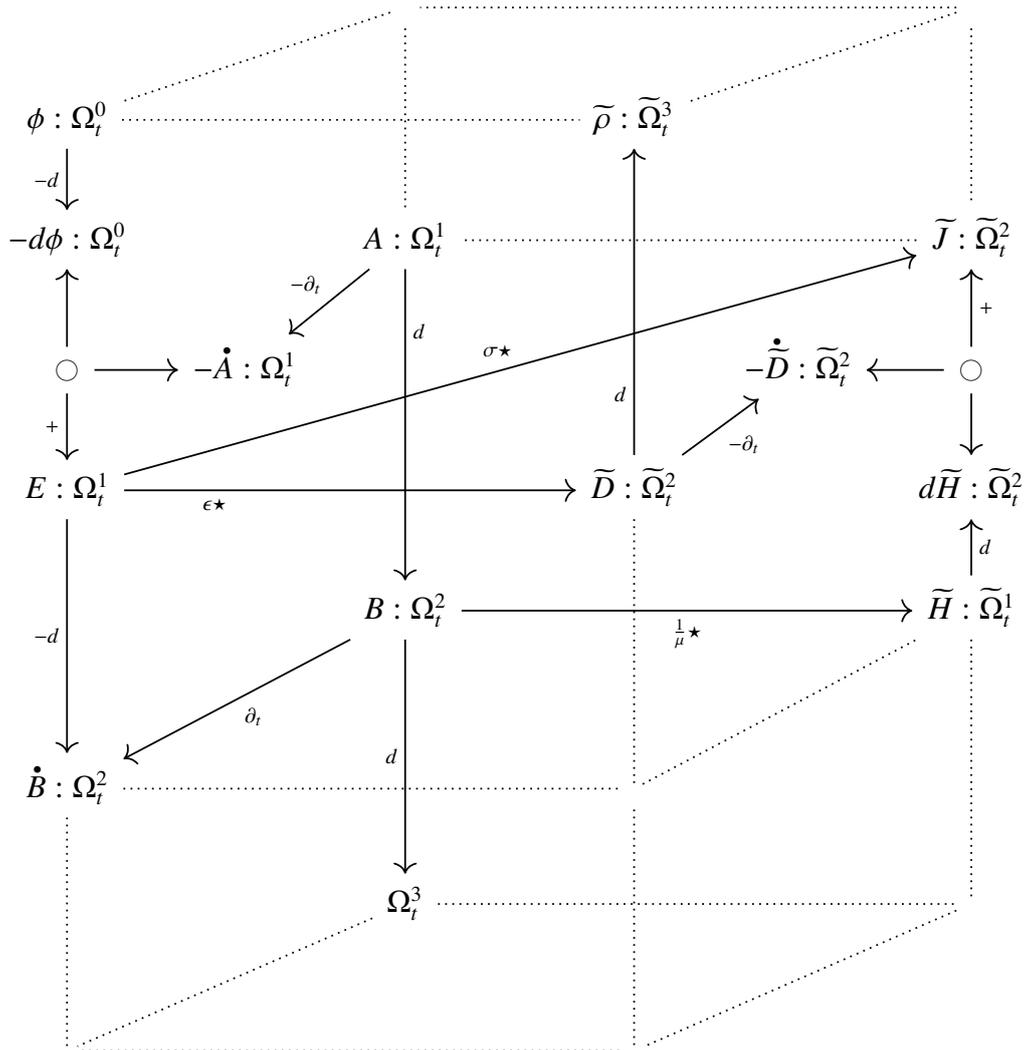
\begin{figure}[H]
    \centering
    \begin{equation*}
        \begin{tikzcd}[column sep=small]
        	&& {} &&&&& {} \\
        	{\phi: \Omega_t^0} &&&&& {\widetilde\rho: \widetilde\Omega_t^3} \\
        	{-d\phi: \Omega_t^0} && {A: \Omega_t^1} &&&&& {\widetilde{J}: \widetilde\Omega_t^2} \\
        	\bigcirc & {- \dt A: \Omega_t^1} &&&&& {-\dt{\widetilde{D}}: \widetilde\Omega_t^2} & \bigcirc \\
        	{E: \Omega_t^1} &&&&& {\widetilde{D}: \widetilde\Omega_t^2} && {d \widetilde{H}: \widetilde\Omega_t^2} \\
        	&& {B: \Omega_t^2} &&&&& {\widetilde{H}: \widetilde\Omega_t^1} \\
        	\\
        	{\dt B: \Omega_t^2} &&&&& {} \\
        	&& {\Omega_t^3} &&&&& {} \\
        	\\
        	{} &&&&& {}
        	\arrow["{+}"', from=4-1, to=5-1]
        	\arrow[from=4-1, to=3-1]
        	\arrow["{-d}"', from=2-1, to=3-1]
        	\arrow[from=4-1, to=4-2]
        	\arrow["{-\partial_t}"', from=3-3, to=4-2]
        	\arrow["{-d}"', from=5-1, to=8-1]
        	\arrow["d"{pos=0.2}, from=3-3, to=6-3]
        	\arrow["{\partial_t}", from=6-3, to=8-1]
        	\arrow["{\epsilon \star}"'{pos=0.2}, from=5-1, to=5-6]
        	\arrow["{\frac{1}{\mu} \star}"', from=6-3, to=6-8]
        	\arrow["d"', from=6-3, to=9-3]
        	\arrow["{-\partial_t}"', from=5-6, to=4-7]
        	\arrow["d"', from=6-8, to=5-8]
        	\arrow["d"{pos=0.2}, from=5-6, to=2-6]
        	\arrow[from=4-8, to=5-8]
        	\arrow[from=4-8, to=4-7]
        	\arrow["{+}"', from=4-8, to=3-8]
        	\arrow[dotted, no head, from=8-1, to=11-1]
        	\arrow[dotted, no head, from=8-1, to=8-6]
        	\arrow[dotted, no head, from=8-6, to=6-8]
        	\arrow[dotted, no head, from=2-1, to=2-6]
        	\arrow[dotted, no head, from=11-1, to=11-6]
        	\arrow[dotted, no head, from=8-6, to=11-6]
        	\arrow[dotted, no head, from=5-6, to=8-6]
        	\arrow[dotted, no head, from=6-8, to=9-8]
        	\arrow[dotted, no head, from=9-8, to=11-6]
        	\arrow[dotted, no head, from=9-3, to=9-8]
        	\arrow[dotted, no head, from=9-3, to=11-1]
        	\arrow[dotted, no head, from=1-3, to=2-1]
        	\arrow[dotted, no head, from=3-8, to=3-3]
        	\arrow[dotted, no head, from=1-8, to=1-3]
        	\arrow[dotted, no head, from=1-8, to=2-6]
        	\arrow[dotted, no head, from=1-8, to=3-8]
        	\arrow[dotted, no head, from=1-3, to=3-3]
        	\arrow["{\sigma \star}", from=5-1, to=3-8]
        \end{tikzcd}
    \end{equation*}
    \caption{Maxwell's house. Maxwell's equations in matter, formulated in exterior calculus on a three-dimensional Riemannian manifold $M$, are presented as a cartesian diagram in the category $\Sh_\R(M \times [0,\infty))$. The dotted lines exist only to bring out the three-dimensional shape of the diagram and have no formal meaning. The remainder of the picture constitutes a well-defined mathematical object, as explained in this and preceding sections. Compare with the traditional, informal version of Maxwell's house in \cref{fig:maxwell-house-traditional}.}
    \label{fig:maxwell-house}
\end{figure}

\end{example}

{\bf Lifting morphisms of cartesian diagrams.}
To conclude this section, we extend the main lifting result of \cref{sec:lifting-properties} to the cartesian setting. The reader willing to take this result for granted may skip the following without loss of continuity.

\begin{theorem}
Suppose that $\pi: \cat{E} \to \cat{C}$ is a cartesian discrete opfibration between cartesian categories. Then the functor $\DiagOp(\pi): \DiagOp(\cat{E}) \to \DiagOp(\cat{C})$ between categories of cartesian diagrams is also a discrete opfibration.
\end{theorem}
\begin{proof}
Suppose $(R,\rho): D \to D'$ is a morphism of cartesian diagrams and $\overline D$ is a cartesian diagram lifting $D$ through $\pi$. Then, in particular, $(R,\rho)$ is a morphism of (bare) diagrams and $\overline D$ is a diagram lifting $D$ through $\pi$. By \cref{thm:diagram-opfibration}, there exists a unique morphism of diagrams $(R, \overline\rho): \overline D \to \overline D'$ lifting $(R,\rho)$. Thus, to prove that $\DiagOp(\pi)$ is a discrete opfibration, we must show that the lift $\overline D'$ is necessarily a \emph{cartesian} diagram. That is a consequence of the following lemma.
\end{proof}

\begin{lemma} \label{lem:cartesian-diagram-lifts}
If $D: \cat{J} \to \cat{C}$ is a cartesian diagram and $\overline D: \cat{J} \to \cat{E}$ is a diagram lifting $D$ through a cartesian discrete opfibration $\pi: \cat{E} \to \cat{C}$, then $\overline D$ is also cartesian.
\end{lemma}
\begin{proof}
For any product cone $\{\pi_i: j \to j_i\}_{i \in I}$ in $\cat{J}$, its image $\{D(\pi_i): Dj\to Dj_i\}_{i \in I}$ is a product cone in $\cat{C}$ and its image $\{\overline{D}(\pi_i): \overline{D}j \to \overline{D}j_i\}_{i \in I}\}$ is a cone in $\cat{E}$. Since $D = \overline D \circ \pi$, the latter cone will be a product cone if the functor $\pi$ reflects products, which is true by the final lemma below.
\end{proof}

\begin{lemma} \label{lem:reflect-products}
A discrete opfibration (or discrete fibration) between cartesian categories that preserves finite products also reflects finite products.
\end{lemma}
\begin{proof}
In general, a functor $F: \cat{C} \to \cat{D}$ creates limits for a class of diagrams, and in particular reflects them, whenever (i) $\cat{C}$ has those limits and $F$ preserves them and (ii) $F$ reflects isomorphisms \cite[Exercise 3.3.iv]{riehl2016}. Since discrete (op)fibrations reflect isomorphisms, the result follows.
\end{proof}

\section{Diagrams involving tensor products} \label{sec:monoidal-diagrams}

Any physical system modeled by \emph{linear} partial differential equations can be formalized using cartesian diagrams in a category of vector spaces or sheaves of vector spaces, as in the preceding section. The classical physics of diffusion, waves, and electromagnetism belong to this setting, as does the quantum mechanics of continuous systems, described by Schr\"odinger's equation. But, of course, many other physical systems exhibit nonlinearity. 

Nonlinear physics is accomodated by the diagrammatic formalism in at least two different ways. First, cartesian diagrams can be adapted to nonlinear equations by simply abandoning the category of (sheaves of) vector spaces and linear maps in favor of the category of (sheaves of) sets and functions. One could also consider vector spaces and possibly nonlinear maps as a compromise. In either case, categorical products are given by set-theoretic products. This approach is perfectly viable and, for strongly nonlinear systems like the sine-Gordon equation \cite{rubinstein1970} (which involves the sine of an unknown function) may even be necessary. However, most equations of mathematical physics exhibit more structure than an arbitrary nonlinear system in that, while they may be nonlinear, they are at least \emph{multilinear} in the unknown functions and their derivatives. In such cases, it is natural to remain in a category with linear maps, using the familiar device of tensor products to reduce multilinearity to linearity.

In this section, we study diagrams involving tensor products. The tensor product of vector spaces is an example of a \emph{symmetric monoidal product} in a category: a product-like operation on a category that does not necessarily satisfy the universal property of a product or coproduct, yet is still associative, unital, and commutative up to coherent natural isomorphism. So, to be more precise, we study diagrams in symmetric monoidal categories. We assume a working knowledge of symmetric monoidal categories, as can be gleaned from the surveys by Baez and Stay \cite{baez2010}, Coecke and Parquette \cite{coecke2010}, and Selinger \cite{selinger2010}.

While it may appear that the passage from diagrams with cartesian products to diagrams with monoidal products is shorter than that from bare diagrams to those with cartesian products, the opposite is the case. The projections associated with a cartesian product, together with their universal property, permit cartesian diagrams to be presented in essentially the same style as ordinary diagrams. Although we have been agnostic about the syntax of cartesian diagrams in the logician's strict sense, a syntactical, graph-theoretic description can be given using finite product sketches \cite{wells1993}. Because a generic monoidal product does not have projections, just as a tensor product of vector spaces does not, monoidal diagrams require an entirely different visual and combinatorial mode of presentation. This section therefore marks a more significant departure than the previous one, and certain aspects of the theory will be less fully developed.

The objects of interest are now suitably structured diagrams in a symmetric monoidal category (SMC).

\begin{definition}[Monoidal diagram] \label{def:monoidal-diagram}
A \emph{symmetric monoidal diagram} in an SMC $\cat{C}$ is a symmetric monoidal functor $D: \cat{J} \to \cat{C}$, where the indexing SMC $\cat{J}$ is small.
\end{definition}

Morphisms of symmetric monoidal diagrams are defined similarly to \cref{def:diagram-categories}, replacing ``category'' with ``symmetric monoidal category'', ``functor'' with ``symmetric monoidal functor,'' and ``natural transformation'' with ``monoidal natural transformation.'' Formally, the 2-categories of symmetric monoidal diagrams are the diagram 2-categories relative to the 2-category $\SMC$, in the sense of \cref{rem:2-categories}.

Recall that a symmetric monoidal functor $(F,\Phi,\phi): (\cat{C}, \otimes_{\cat{C}}, I_{\cat{C}}) \to (\cat{D}, \otimes_{\cat{D}}, I_{\cat{D}})$ between SMCs consists of a functor $F: \cat{C} \to \cat{D}$ between the underlying categories, together with a natural family of structure morphisms
\begin{equation*}
    \Phi_{x,y}: Fx \otimes_{\cat{D}} Fy \to F(x \otimes_{\cat{C}} y),
    \qquad x, y \in \cat{C},
\end{equation*}
and $\phi: I_{\cat{D}} \to F(I_{\cat{C}})$, subject to certain axioms. In general, a symmetric monoidal functor is \emph{lax}. When the structure morphisms are all isomorphisms, it is \emph{strong}; when they are identities, it is \emph{strict}. In this paper, including in \cref{def:monoidal-diagram}, symmetric monoidal functors are taken to be strong unless otherwise stated. Moreover, we can and will assume that the shapes of symmetric monoidal diagrams are strict SMCs and that the maps between shapes are strict symmetric monoidal functors.

If free categories, which are the shapes of free diagrams, are generated by graphs, then what kind of structure generates free strict SMCs, which are the shapes of free symmetric monoidal diagrams? In an influential paper \cite{meseguer1990}, Meseguer and Montanari established a connection between Petri nets and free strict SMCs, which has been fully clarified only recently through work by Baez et al \cite{baez2021nets} and independently by Kock \cite{kock2020}. Free strict SMCs are generated by a Petri-net-like structure called \emph{$\Sigma$-nets} \cite{baez2021nets}, of which a special case are \emph{whole-grain Petri nets} \cite{kock2020}. We refer to the cited papers for the precise definitions, but idea is simple enough. A \emph{net} is a kind of directed bipartite graph, with vertices of one type---drawn as ellipses---being the generating objects of the SMC and the vertices of the other type---drawn as rectangles---being the generating morphisms. Input and output arcs between the two types of vertices specify the domain and codomain of the generating morphisms. For example, the free strict SMC generated by objects $w, x, y, z$ and morphisms $f: w \to x \otimes y$ and $g: x \otimes y \to z$ is presented by the net
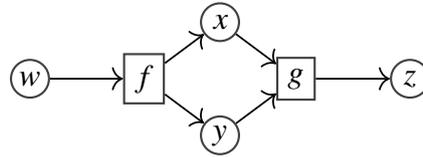
\begin{equation*}
    \begin{tikzpicture}
      \node[place] (x) at (0,0.75) {$x$};
      \node[place] (y) at (0,-0.75) {$y$};
      \node[place] (w) at (-2.5,0) {$w$};
      \node[place] (z) at (2.5,0) {$z$};
      \node[transition] (f) at (-1,0) {$f$};
      \node[transition] (g) at (1,0) {$g$};
      \path[commutative diagrams/.cd, every arrow, every label]
        (w) edge (f)
        (f) edge (x)
        (f) edge (y)
        (x) edge (g)
        (y) edge (g)
        (g) edge (z);
    \end{tikzpicture}.
\end{equation*}
To save space and emphasize connections with graph-based diagrams, we adopt the nonstandard convention of drawing morphism generators that have exactly one input and one output as a labeled arrow. Thus, the pictures
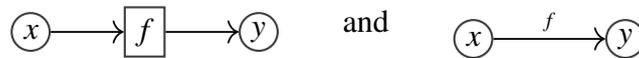
\begin{equation*}
    \begin{tikzpicture}[baseline=(current bounding box.center)]
      \node[place] (x) at (-1.5,0) {$x$};
      \node[place] (y) at (1.5,0) {$y$};
      \node[transition] (f) at (0,0) {$f$};
      \path[commutative diagrams/.cd, every arrow, every label]
        (x) edge (f)
        (f) edge (y);
    \end{tikzpicture}
    \qquad\text{and}\qquad
    \begin{tikzpicture}[baseline=(current bounding box.center)]
      \node[place] (x) at (-1,0) {$x$};
      \node[place] (y) at (1,0) {$y$};
      \path[commutative diagrams/.cd, every arrow, every label]
        (x) edge node {$f$} (y);
    \end{tikzpicture}
\end{equation*}
both represent the generating morphism $f: x \to y$. In summary, using  nets, we can visually present free strict SMCs and thus also free symmetric monoidal diagrams.

The category of generalized elements of shape $S$ in an SMC $(\cat{C}, \otimes, I)$ becomes an SMC itself when $S$ has the additional structure of commutative comonoid object $(S, \delta, \epsilon)$ in $\cat{C}$. Here $\delta: S \to S \otimes S$ is the comultiplication and $\epsilon: S \to I$ is the counit. The monoidal product on $\El_S(\cat{C})$ is defined on objects $x: S \to X$ and $y: S \to Y$ by
\begin{equation*}
    \begin{tikzcd}
    	S && {X \otimes Y} \\
    	& {S \otimes S}
    	\arrow["{x \otimes_{\mathrm{El}(\mathsf{C})} y}", dashed, from=1-1, to=1-3]
    	\arrow["\delta"', from=1-1, to=2-2]
    	\arrow["{x \otimes_{\mathsf{C}} y}"', from=2-2, to=1-3]
    \end{tikzcd},
\end{equation*}
with monoidal unit $\epsilon: S \to I$, and on morphisms as in $\cat{C}$. Defining the associators, unitors, and braidings of $\El_S(\cat{C})$ by post-composition with the associators, unitors, and braidings of $\cat{C}$, it is straightforward to check that the former are well-defined and satisfy the SMC axioms. For example, the braidings are well-defined because the composites
\begin{equation*}
    S \xrightarrow{\delta} S \otimes S \xrightarrow{x \otimes y} X \otimes Y \xrightarrow{\sigma_{X,Y}} Y \otimes X
    \qquad\text{and}\qquad
    S \xrightarrow{\delta} S \otimes S \xrightarrow{y \otimes x} Y \otimes X
\end{equation*}
are equal by the naturality of the family $\sigma_{X,Y}$, $X,Y \in \cat{C}$, and by the commutativity of the comonoid object. Moreover, the codomain projection $\pi := \cod: \El_S(\cat{C}) \to \cat{C}$ is a \emph{strict} symmetric monoidal functor.

Two special cases of this construction are particularly relevant. First, when a category $\cat{C}$ has finite products, a \emph{cartesian monoidal category} $(\cat{C}, \times, 1)$ is defined by making a choice of products for each pair of objects and a choice of terminal object. For any object $S \in \cat{C}$, there is a unique commutative comonoid object $(S, \Delta_S, !_S)$ with carrier $S$, in which $\Delta_S$ is the diagonal map $(\id_S, \id_S): S \to S \times S$ and $!_S$ is the unique map $S\to 1$. The category of generalized elements $\El_S(\cat{C})$ then becomes a cartesian monoidal category itself. Thus, the setting of the previous \cref{sec:cartesian-diagrams} is a special case of this one, apart from the need to make a choice of products in order to obtain a monoidal structure. The cartesian case has enough special features to merit its own treatment, however.

The other case of interest is when the category of generalized elements has shape $S = I$ equal to the monoidal unit of the SMC $(\cat{C}, \otimes, I)$. The commutative comonoid object $(I, \delta, \epsilon)$ is given by unitor isomorphisms, with the comultiplication and counit being
\begin{equation*}
    \delta := \lambda_I^{-1} = \rho_I^{-1}: I \xrightarrow{\cong} I \otimes I
    \qquad\text{and}\qquad
    \epsilon := \id_I: I \xrightarrow{\cong} I.
\end{equation*}
In this case, the SMC is often a symmetric \emph{closed} monoidal category, as in the motivating example of $(\Vect_\R, \otimes, \R)$, the category of real vector spaces with its tensor product. In the category $\Sh_\R(M)$ of sheaves of vector spaces on a manifold $M$, the tensor product is defined by taking tensor products of vector spaces locally, for each open subspace of $M$, followed by sheafification.

A final complication is that tensor products and cartesian products/direct sums are often needed in conjunction. The structure of a category with two symmetric monoidal products, where one distributes over the other up to coherent isomorphism, is captured by the notion of a \emph{symmetric rig category} or \emph{symmetric bimonoidal category} \cite[Chapter 2]{johnson2021}. In the following examples, we make free use of both tensor products and cartesian products. This is justified by the theory developed in the final parts of the present section and the previous one, although we omit the details of the bimonoidal structure.

\begin{example}[Incompressible Euler equations] \label{ex:incompressible-euler}
The basic physical quantities for fluid flow on a Riemannian manifold $M$ are the
\begin{itemize}[noitemsep]
    \item velocity field for the fluid, $\bm{u} \in \VectField_t(M)$,
    \item 1-form proxy for the velocity field, $u := \bm{u}^\flat \in \Omega_t^1(M)$,
    \item pressure, $p \in \Omega_t^0(M)$, and
    \item mass density, $\rho \in \Omega_t^0(M)$.
\end{itemize}
Here $\VectField_t(M)$ denotes the space of time-dependent smooth vector fields on $M$, and the \emph{musical isomorphism} $(-)^\flat: \VectField_t(M) \rightleftarrows \Omega_t^1(M): (-)^\sharp$ uses the Riemannian metric to pass between vector fields and 1-forms (covector fields).

For a homogeneous, incompressible fluid, the density is a constant in space and time, $\rho \equiv \rho_0$, which we may as well normalize to 1. The Euler equations are then
\begin{align*}
    \partial_t u + \nabla_{\bm u} u + dp &= 0 \\
    \delta u &= 0,
\end{align*}
where $\nabla$ is the covariant derivative on the Riemannian manifold. Using the formula $\mathcal{L}_{\bm u} u = \nabla_{\bm u} u + \frac{1}{2} d \norm{\bm{u}}^2$ that relates the Lie and covariant derivatives of $u$ with respect to $\bm u$ itself \cite{dupre1978}, we can rewrite the first equation as
\begin{equation*}
    \partial_t u + \mathcal{L}_{\bm u} u - 
        \frac{1}{2} d \norm{\bm{u}}^2 + dp = 0.
\end{equation*}
This version of the Euler equation is useful when translating to the discrete exterior calculus, which defines a discrete Lie derivative but not a discrete covariant derivative \cite{hirani2003}. Boyland compares the Lie, covariant, and material derivatives in the context of fluid mechanics \cite[\S 4.5]{boyland2001}. For a textbook derivation of the Euler equations on a Riemannian manifold, see \cite[\S 8.2, p.\ 588]{abraham1988}.

\cref{subfig:euler} shows the incompressible Euler equations as a symmetric monoidal diagram in $\Sh_\R(M \times [0,\infty))$.
\end{example}

\begin{figure}[H]
    \centering
    \begin{subfigure}{\textwidth}
    \[\begin{tikzpicture}
      \node[place] (u) at (0,0) {$u: \Omega_t^1$};
      \node[place] (uvec) at (2,0) {$\bm{u}: \VectField_t$};
      \node[place] (divu) at (-2.5,0) {$\delta u: \Omega_t^0$};
      \node[place] (uzero) at (-4,0) {$0$};
      \node[place] (udot) at (-3,-1.5) {$\dt u: \Omega_t^1$};

      \node[transition] (covd) at (1,-1.25) {$\nabla$};
      \node[place] (ucovd) at (1,-2.25) {$\nabla_{\bm u} u: \Omega_t^1$};

      \node[place] (p) at (4,0) {$p: \Omega_t^0$};
      \node[place] (dp) at (4,-1.5) {$dp: \Omega_t^1$};

      \node[sum place] (sum) at (1,-3.25) {};
      \node[place] (sumrhs) at (1,-4.25) {$\Omega_t^1$};
      \node[place] (sumzero) at (1,-5.25) {$0$};

      \path[commutative diagrams/.cd, every arrow, every label]
        (u) edge[swap] node {$\delta$} (divu)
        (uzero) edge (divu)
        (u) edge[bend left] node {$\sharp$} (uvec)
        (u) edge[swap] node {$\partial_t$} (udot)
        (uvec) edge[swap,bend left] node {$\flat$} (u)

        (u) edge[bend right,in=110] ([xshift=-1mm]covd.north east)
        (uvec) edge[bend left,in=250] ([xshift=1mm]covd.north west)
        (covd) edge (ucovd)

        (p) edge node {$d$} (dp)

        (sum) edge (ucovd)
        (sum) edge[bend left] (udot)
        (sum) edge[bend right] (dp)
        (sum) edge node {$+$} (sumrhs)
        (sumzero) edge (sumrhs);
    \end{tikzpicture}\]
    \caption{Incompressible Euler equations}
    \label{subfig:euler}
    \end{subfigure}
    \par\bigskip
    \begin{subfigure}{\textwidth}
    \[\begin{tikzpicture}
      \node[place] (u) at (0,0) {$u: \Omega_t^1$};
      \node[place] (uvec) at (2,0) {$\bm{u}: \VectField_t$};
      \node[place] (divu) at (-2.5,0) {$\delta u: \Omega_t^0$};
      \node[place] (uzero) at (-4,0) {$0$};
      \node[place] (udot) at (-3,-1.5) {$\dt u: \Omega_t^1$};

      \node[place] (stress) at (-1,-1.5) {$\tau: \Omega_t^2$};
      \node[place] (dstress) at (-1,-2.5) {$\delta \tau: \Omega_t^1$};

      \node[transition] (covd) at (1,-1.25) {$\nabla$};
      \node[place] (ucovd) at (1,-2.5) {$\nabla_{\bm u} u: \Omega_t^1$};

      \node[place] (p) at (4,0) {$p: \Omega_t^0$};
      \node[place] (dp) at (4,-1.5) {$dp: \Omega_t^1$};

      \node[sum place] (sum) at (1,-3.5) {};
      \node[place] (sumrhs) at (1,-4.5) {$\Omega_t^1$};
      \node[place] (sumzero) at (1,-5.5) {$0$};

      \path[commutative diagrams/.cd, every arrow, every label]
        (u) edge[swap] node {$\delta$} (divu)
        (uzero) edge (divu)
        (u) edge[bend left] node {$\sharp$} (uvec)
        (u) edge[swap] node {$\partial_t$} (udot)
        (uvec) edge[swap,bend left] node {$\flat$} (u)

        (u) edge node {$\mu\, d$} (stress)
        (stress) edge node {$\delta$} (dstress)

        (u) edge[bend right,in=110] ([xshift=-1mm]covd.north east)
        (uvec) edge[bend left,in=250] ([xshift=1mm]covd.north west)
        (covd) edge (ucovd)

        (p) edge node {$d$} (dp)

        (sum) edge (ucovd)
        (sum) edge[bend left] (udot)
        (sum) edge[bend left] (dstress)
        (sum) edge[bend right] (dp)
        (sum) edge node {$+$} (sumrhs)
        (sumzero) edge (sumrhs);
    \end{tikzpicture}\]
    \caption{Incompressible Navier-Stokes equations}
    \label{subfig:navier-stokes}
    \end{subfigure}
    \caption{Fundamental equations for incompressible fluid flow on a Riemannian manifold $M$, presented as symmetric monoidal diagrams in $\Sh_\R(M \times [0,\infty))$.}
    \label{fig:fluid-mechanics}
\end{figure}
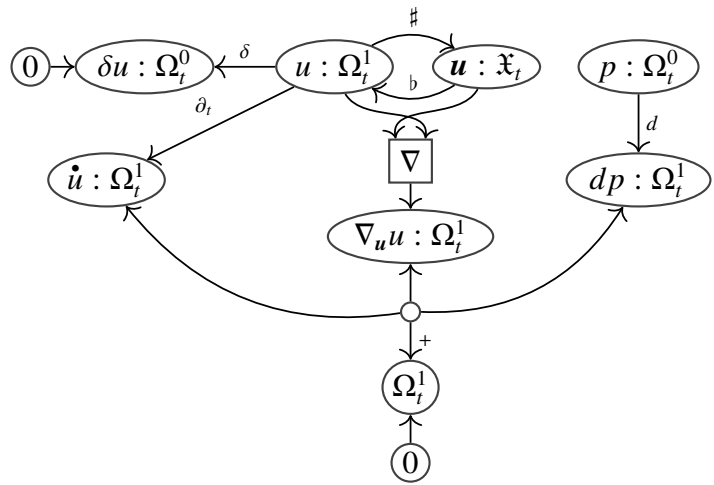
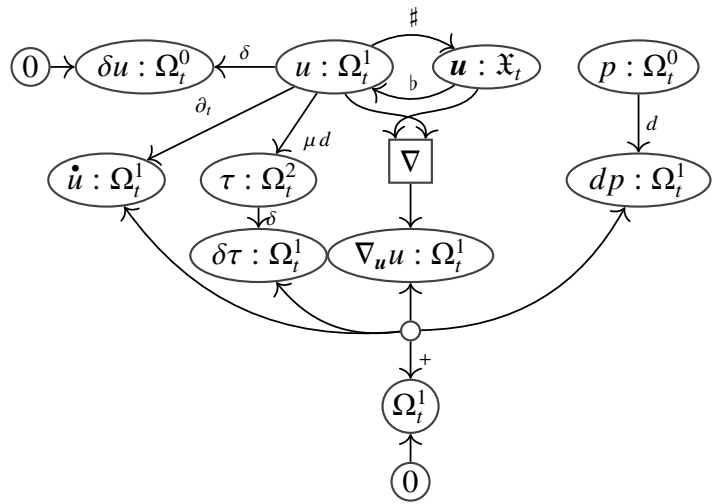

\begin{example}[Incompressible Navier-Stokes]
  The Euler equations model the flow of an inviscous fluid. The Navier-Stokes equations extend the Euler equations with an extra term to account for viscosity. For a homogeneous, incompressible fluid, the Navier-Stokes equations on a Riemannian manifold are
  \begin{align*}
    \partial_t u + \nabla_{\bm u} u + \mu\, \delta d u + dp &= 0 \\
    \delta u &= 0,
  \end{align*}
  where $\mu$, assumed to be constant, is the dynamic viscosity \cite[Equation 5]{mohamed2016}. The extra term $\mu\, \delta d u$ can be rewritten as $\delta \tau$, where $\tau$ is the stress and $\tau = \mu\, d u$ is a constitutive equation relating the stress and the gradient of the flow velocity. With this notation, \cref{subfig:navier-stokes} shows the incompressible Navier-Stokes equations as a symmetric monoidal diagram.
\end{example}

{\bf Lifting morphisms of monoidal diagrams.}
Just as discrete opfibrations provide a useful level of generality at which to study the lifting properties of diagram morphisms, \emph{monoidal} discrete opfibrations, introduced more recently, are useful to study the lifting properties of morphisms of monoidal diagrams. In the remainder of this section, we define monoidal discrete opfibrations, explain their connection with monoidal category actions, and extend the main lifting result of \cref{sec:lifting-properties} to monoidal diagrams. This material may be omitted without loss of continuity.

\begin{definition}[Monoidal discrete opfibration]
A \emph{symmetric monoidal discrete opfibration} between SMCs is a strict symmetric monoidal functor that is also a discrete opfibration between the underlying categories.
\end{definition}

The definition specializes that of a (non-discrete) monoidal opfibration, introduced by Shulman \cite[Definition 12.1]{shulman2008}. We emphasize that the monoidal functor is required to be strict even when the monoidal categories are not strict, which is atypical for monoidal functors. As a result, a monoidal discrete opfibration $\pi: \cat{E} \to \cat{C}$ satisfies the important property that, if $\overline f: \overline x \to \overline w$ and $\overline g: \overline y \to \overline z$ in $\cat{E}$ are lifts of $f: x \to w$ and $g: y \to z$ in $\cat{C}$ through $\pi$, then $\overline f \otimes \overline g$ is a lift of $f \otimes g$ through $\pi$ with domain $\overline x \otimes \overline y$, and hence is the unique such lift.

The equivalence between discrete opfibrations over a category $\cat{C}$ and actions of the category $\cat{C}$, recalled in \cref{eq:dopf-cset-equivalence}, extends to monoidal categories. An \emph{action} of an SMC $(\cat{C}, \otimes, I)$ is a \emph{lax} symmetric monoidal functor $(X,\Phi,\phi): (\cat{C}, \otimes, I) \to (\Set, \times, 1)$, whose structure morphisms
\begin{equation*}
    \Phi_{c,d}: F(c) \times F(d) \to F(c \otimes d), \qquad
    c, d \in \cat{C},
\end{equation*}
and $\phi: 1 = \{*\} \to F(I)$ are interpreted as an action of the monoidal product on the family of sets $F(c)$ indexed by $c \in \cat{C}$. Given any symmetric monoidal category action $(X,\Phi,\phi)$, the category of elements $\El(X)$ becomes an SMC with monoidal product
\begin{equation*}
    (c, x) \otimes (d, y) := (c \otimes d, \Phi_{c,d}(x,y)),
    \qquad c,d \in \cat{C},\ x \in F(c),\ y \in F(d)
\end{equation*}
on objects and monoidal unit $(I, \phi(*))$, where $\phi(*) \in F(I)$. Monoidal products on morphisms, as well as associators, unitors, and braidings, are inherited from $\cat{C}$. With this definition, the canonical projection $\pi_X: \El(X) \to \cat{C}$ is a symmetric monoidal discrete opfibration. Moreover, the category of elements construction extends to an equivalence between symmetric monoidal discrete opfibrations over $\cat{C}$ and actions of the SMC $\cat{C}$. We will not state this equivalence precisely, but it is a special case of the equivalence between monoidal opfibrations and lax monoidal pseudofunctors $(\cat{C}, \otimes, I) \to (\Cat, \times, 1)$, established by Moeller and Vasilakopoulou using the monoidal Grothendieck construction \cite{moeller2020}.

For example, when $(S, \delta, \epsilon)$ is a commutative comonoid object in an SMC $(\cat{C}, \otimes, I)$, the representable category action $\Hom_{\cat{C}}(S,-): \cat{C} \to \Set$ becomes an action of $(\cat{C}, \otimes, I)$ via the structure morphisms
\begin{equation*}
    \Phi_{c,d}(S \xrightarrow{x} c, S \xrightarrow{y} d) :=
      (S \xrightarrow{\delta} S \otimes S \xrightarrow{x \otimes y} c \otimes d),
    \qquad c,d \in \cat{C},
\end{equation*}
and $\phi(*) := (S \xrightarrow{\epsilon} I)$. Taking its category of elements, we recover the symmetric monoidal discrete opfibration $\cod: \El_S(\cat{C}) \to \cat{C}$ associated with the category of generalized elements of shape $S$.

The basic lifting result in \cref{thm:diagram-opfibration} extends to monoidal diagrams as follows.

\begin{theorem}
Suppose that $\pi: \cat{E} \to \cat{C}$ is a symmetric monoidal discrete opfibration between SMCs. Then the functor $\DiagOp(\pi): \DiagOp(\cat{E}) \to \DiagOp(\cat{C})$ between categories of symmetric monoidal diagrams is a discrete opfibration.
\end{theorem}
\begin{proof}
Suppose $(R,\rho): D \to D'$ is a morphism of symmetric monoidal diagrams and $(\overline D, \overline\Phi, \overline\phi)$ is a symmetric monoidal functor lifting $(D, \Phi, \phi)$ through $\pi$. Then, in particular, $(R,\rho)$ is a morphism of diagrams and $\overline D$ is a diagram lifting $D$ through $\pi$. By \cref{thm:diagram-opfibration}, there exists a unique morphism of diagrams $(R, \overline\rho): \overline D \to \overline D'$ lifting $(R,\rho)$. Thus, to prove that $\DiagOp(\pi)$ is a discrete opfibration, we must show that there exists a unique structure of a symmetric monoidal functor $(\overline{D}', \overline\Phi', \overline\phi')$ on $\overline D'$ that lifts $D'$ through $\pi$ \emph{as a symmetric monoidal functor}, i.e., such that $\pi \circ (\overline{D}', \overline\Phi', \overline\phi') = (D', \Phi', \phi')$. We must also show that $\overline\rho$ is a monoidal natural transformation with respect to this structure.

The uniqueness is clear, since if the structure morphisms exist, then the composition law for monoidal functors and the strictness of $\pi$ implies that $\pi(\overline\Phi_{j,k}') = \Phi_{j,k}'$ for all $j,k \in \cat{J}'$ and also $\pi(\overline\Phi') = \phi'$. The uniqueness of the structure morphisms $\overline\Phi_{j,k}'$ and $\overline\phi'$ then follows from the uniqueness of lifts through $\pi$.

We now prove the existence of the structure morphisms. For any $j,k \in \cat{J}'$, since $\pi$ is strict, $\overline{D}' j \otimes \overline{D}' k$ is a lift of $D' j \otimes D' k$ through $\pi$. Let $\overline\Phi_{j,k}': \overline{D}' j \otimes \overline{D}' k \to \bullet$ be the corresponding lift of $\Phi_{j,k}': D' j \otimes D' k \to D'(j \otimes k)$ through $\pi$. Then the morphisms
\begin{equation*}
    \overline{D} Rj \otimes \overline{D} Rk
      \xrightarrow{\overline\rho_j \otimes \overline\rho_k}
    \overline{D}' j \otimes \overline{D}' k
      \xrightarrow{\overline\Phi_{j,k}'} \bullet
\end{equation*}
and
\begin{equation*}
    \overline{D} Rj \otimes \overline{D} Rk \xrightarrow{\overline\Phi_{Rj,Rk}}
    \overline{D} (Rj \otimes Rk) = \overline{D} R(j \otimes k) \xrightarrow{\overline\rho_{j \otimes k}}
    \overline{D}' (j \otimes k)
\end{equation*}
in $\cat{E}$ project under $\pi$ to the morphisms
\begin{equation*}
    DRj \otimes DRk \xrightarrow{\rho_j \otimes \rho_k}
    D'j \otimes D'k \xrightarrow{\Phi_{j,k}'} D'(j \otimes k)
\end{equation*}
and
\begin{equation*}
    DRj \otimes DRk \xrightarrow{\Phi_{Rj,Rk}}
    D(Rj \otimes Rk) = DR(j \otimes k) \xrightarrow{\rho_{j \otimes k}} D'(j \otimes k)
\end{equation*}
in $\cat{C}$, respectively. But the latter two morphisms are equal, since $\rho: D \circ R \To D'$ is a monoidal natural transformation, so the former two morphisms are also equal, by the uniqueness of lifts through $\pi$. In particular, the codomain of $\overline\Phi_{j,k}'$ is $\overline{D}'(j \otimes k)$, as required. Similary, there is a lift $\overline\phi': I_{\cat{E}} \to \overline{D}'(I_{\cat{J}'})$ of $\phi': I_{\cat{C}} \to D'(I_{\cat{J}'})$ through $\pi$. Thus, we have the data $(\overline{D}', \overline\Phi', \overline\phi')$ of a symmetric monoidal functor. That it \emph{is} a strong symmetric monoidal functor follows straightforwardly from $\pi$ being discrete opfibration, and the argument above already shows that $\overline\rho$ is a monoidal natural transformation.
\end{proof}

\section{Composition of diagrams and multiphysics} \label{sec:composing-diagrams}

Multiphysics is the study of physical systems that couple together distinct physical phenomena \cite{keyes2013}. For example, reaction-diffusion systems combine chemical reactions between different substances with the diffusion of those substances throughout a medium \cite{fife1979,socolofsky2005}. Of course, what counts as a single physics or as multiple physics is largely a matter of perspective, and even systems that are usually considered to involve a single physics, such as a diffusion process, may be derived by combining several distinct physical principles. In this section, as an application of the diagrammatic formalism, we systematize the compositional and hierarchical construction of physical, possibly multiphysical, theories from first principles and experimental laws.

Besides categories of diagrams, the key mathematical tool enabling this formalization are the structured cospans introduced by Baez and Courser \cite{baez2020} as a variant of Fong's decorated cospans \cite{fong2015}. Decorated and structured cospans belong to a now well established tradition of modeling open systems usings spans or cospans \cite{katis1997}. Open systems of a geometrical nature are modeled as cospans, where the legs of a cospan give inclusions of the system's boundaries into the system itself. Composition is by pushout (gluing along boundaries). Dually, open systems consisting of quantitative state variables on a geometrical space are modeled as spans, where the legs of a span give projections of the system's state space onto the state spaces of its boundaries. Composition is by pullback or variants thereof, such as in Baez, Weisbart, and Yassine's study of open systems in Lagrangian and Hamiltonian mechanics \cite{baez2021mechanics}.

Having seen how diagrams can present physical quantities on geometrical spaces and how diagram morphisms can serve as projections defining boundary value problems, it would be natural to define ``open diagrams'' in a category $\cat{C}$ as spans in the diagram category $\DiagOp(\cat{C})$. For the purposes of this section, however, it suffices to consider spans in $(\Cat/\cat{C})^\op$, the wide subcategory of $\DiagOp(\cat{C})$ having \emph{strict} diagram morphisms. Composing open diagrams by pullback then amounts to computing pushouts in the slice category $\Cat/\cat{C}$ and so ultimately in $\Cat$. Although pushouts of categories are generally complicated, pushouts of free categories (the shapes of free diagrams) are simple because they reduce to pushouts of the generating graphs.

To be more precise, we take open diagrams to be structured spans, which are the formal duals of structured cospans \cite{baez2020}. Given a functor $R: \cat{A} \to \cat{X}$, an \emph{$R$-structured span} is a pair of objects $a,b \in \cat{A}$ together with a span of form $Ra \leftarrow x \rightarrow Rb$ in $\cat{X}$. Assuming the categories $\cat{A}$ and $\cat{X}$ have finite limits and the functor $R$ is a right adjoint, a hypergraph category is formed by the objects of $A$ and isomorphism classes of $R$-structured spans, where composition is given by pullback in $\cat{X}$ \cite[Theorem 3.12]{baez2020}. For our applications, fix a category $\cat{C}$ and consider the discrete diagram functor $\Disc_{\cat{C}}: \Set/\Ob(\cat{C}) \to \Cat/\cat{C}$, which sends an indexed set of objects in $\cat{C}$ to the discrete diagram with those objects. This functor is left adjoint to the underlying objects functor $\Ob_{\cat{C}}: \Cat/\cat{C} \to \Set/\Ob(\cat{C})$. Consequently, its opposite functor $\Disc_{\cat{C}}^\op: (\Set/\Ob(\cat{C}))^\op \to (\Cat/\cat{C})^\op$ is a right adjoint. We provisionally define an \emph{open diagram} in $\cat{C}$ to be a $\Disc_{\cat{C}}^\op$-structured span, or equivalently a span in $\DiagOp(\cat{C})$ whose legs are strict diagram morphisms and whose feet are discrete diagrams.

\begin{example}[Constituents of diffusion]
In \cref{ex:diffusion}, the diffusion equation on a three-dimensional Riemannian manifold $M$ was loosely described as the conjunction of two physical principles, an experimental law for the flux of the diffusing substance (Fick's first law) and a conservation principle relating the instantaneous changes in concentration in time and space (conservation of mass). We can now begin to make this precise by formulating the two principles as open diagrams. Fick's first law is presented by the open diagram
\begin{equation} \label{eq:diffusion-flux-diagram}
    \begin{tikzcd}[
    /tikz/execute at end picture={
        \node (apex) [boxUWD, inner xsep=1.5em, fit=(C1) (dC) (F1)] {};
        \node (foot1) [boxUWD, inner xsep=1.5em, fit=(C2)] {};
        \node (foot2) [boxUWD, inner xsep=1.5em, fit=(F2)] {};
    }]
    	|[alias=C1]| {C: \Omega_t^0} &&& |[alias=C2]| {C: \Omega_t^0} \\
    	|[alias=dC]| {dC: \Omega_t^1} & |[alias=F1]| {\phi: \widetilde\Omega_t^2} && |[alias=F2]| {\phi: \widetilde\Omega_t^2} \\
    	&& {}
    	\arrow["d"', from=1-1, to=2-1]
    	\arrow["{k \star}"', from=2-1, to=2-2]
    	\arrow[dashed, from=1-1, to=1-4]
    	\arrow[dashed, from=2-2, to=2-4]
    \end{tikzcd}
\end{equation}
in $\Sh_\R(M \times [0,\infty))$. Diagram morphisms are drawn here as in \cref{sec:diagram-morphisms}; because they are strict, the dashed lines are identities. Similarly, conservation of mass is presented by the open diagram
\begin{equation} \label{eq:transport-conservation-diagram}
    \begin{tikzcd}[
    /tikz/execute at end picture={
        \node (apex) [boxUWD, inner xsep=1.5em, fit=(C1) (dotC) (dF) (F1)] {};
        \node (foot1) [boxUWD, inner xsep=1.5em, fit=(C2)] {};
        \node (foot2) [boxUWD, inner xsep=1.5em, fit=(F2)] {};
    }]
    	|[alias=C2]| {C: \Omega_t^0} && |[alias=C1]| {C: \Omega_t^0} & |[alias=dotC]| {\dt C: \Omega_t^0} & |[alias=dF]| {d\phi: \widetilde \Omega_t^3} \\
    	|[alias=F2]| {\phi: \widetilde\Omega_t^2} &&&& |[alias=F1]| {\phi: \widetilde \Omega_t^2} \\
    	& {}
    	\arrow["{\partial_t}", from=1-3, to=1-4]
    	\arrow["d"', from=2-5, to=1-5]
    	\arrow["{\star^{-1}}"', from=1-5, to=1-4]
    	\arrow[dashed, from=1-3, to=1-1]
    	\arrow[dashed, from=2-5, to=2-1]
    \end{tikzcd}\ .
\end{equation}
\end{example}

As one would expect, the two open diagrams above can be composed, using the operations available in a hypergraph category, to form the diffusion equation \eqref{eq:diffusion-diagram}. However, experience of the authors and collaborators shows that it is more convenient to employ an alternative syntax for composition that does not artificially insist on dividing the boundary of an open system into exactly two parts \cite{libkind2021,libkind2022}. We thus mildly generalize the definitions of structured spans and open diagrams as follows. Given a functor $R: \cat{A} \to \cat{X}$, an \emph{$R$-structured multispan} is a list of objects $a_1, \dots, a_n \in \cat{A}$, where $n \geq 0$, together with a multispan of form
\begin{equation*}
    \begin{tikzcd}
    	&& x \\
    	{R a_1} & {R a_2} & \cdots & {Ra_{n-1}} & {R a_n}
    	\arrow[from=1-3, to=2-1]
    	\arrow[from=1-3, to=2-4]
    	\arrow[from=1-3, to=2-2]
    	\arrow[from=1-3, to=2-5]
    \end{tikzcd}
\end{equation*}
in the category $\cat{X}$.

\begin{definition}[Open diagram] \label{def:open-diagram}
An \emph{open diagram} in a category $\cat{C}$ is a $\Disc_{\cat{C}}^\op$-structured multispan, or equivalently a multispan in $\DiagOp(\cat{C})$ whose legs are strict diagram morphisms and whose feet are discrete diagrams.
\end{definition}

Although it is not assumed by the theory, the feet of an open diagram will be singleton diagrams (diagrams indexed by the terminal category) in our examples. Open diagrams can then be defined more succinctly by giving the diagram at the apex of the multispan and stating which indexing objects are \emph{exposed} through the legs of the multispan. For example, the open diagrams \eqref{eq:diffusion-flux-diagram} and \eqref{eq:transport-conservation-diagram} both expose the concentration $C$ and the negative diffusion flux $\phi$ and no other quantities.

The definition of an open diagram in a category is straightforwardly modified for open cartesian diagrams in a cartesian category (\cref{sec:cartesian-diagrams}) or open symmetric monoidal diagrams in a symmetric monoidal category (\cref{sec:monoidal-diagrams}). We omit the details.

Composites of open diagrams can be specified in an appealing visual manner using undirected wiring diagrams (UWDs). Given a category $\cat{C}$, there is an operad of typed, undirected wiring diagrams~\cite{spivak2013,libkind2021} with types being discrete diagrams in $\cat{C}$. Isomorphism classes of open diagrams in $\cat{C}$ form an algebra of this operad, where an undirected wiring diagram acts on a list of structured multispans by evaluating a suitable limit, generalizing the composition of structured spans by pullback. The existence of this operad algebra follows from the equivalences between hypergraph categories and cospan algebras \cite{fong2019} and between cospan algebras and algebras of the operad of UWDs \cite[Examples 2.1.3 and 2.1.10]{leinster2004}, \cite[Examples 2.4 and 2.7]{libkind2021}. Details of this construction may be extracted from the given references, but should not be necessary to understand the following examples.

\begin{example}[Diffusion, compositionally] \label{ex:diffusion-compositional}
Modeling mass transport of a substance through a medium, such as through diffusion or advection, requires a description of the transport process itself along with a mass conservation principle, both of which relate the concentration of the substance $C$ to the (negative) diffusion flux $\phi$. This pattern of composition is represented by the undirected wiring diagram
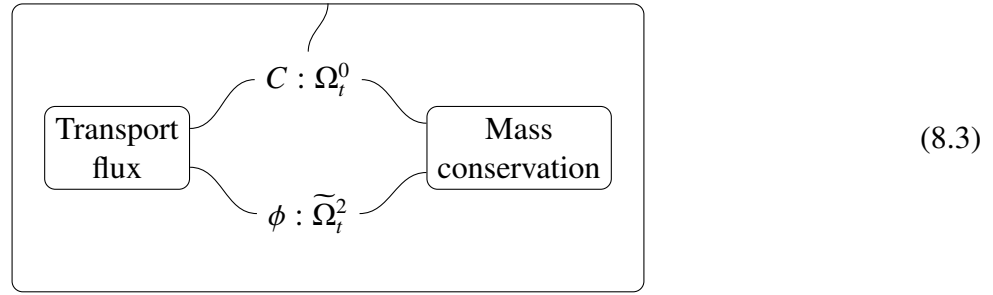
\begin{equation} \label{eq:transport-uwd}
    \begin{tikzpicture}[baseline=(current bounding box.center),tips=proper]
        \node (C) {$C: \Omega_t^0$};
        \node (F) [below=of C] {$\phi: \widetilde\Omega_t^2$};
        \node (CF) [fit=(C) (F), inner sep=0] {};
        \node (transport) [boxUWD, left=2em of CF, align=center] {Transport \\ flux};
        \node (conservation) [boxUWD, right=2em of CF, align=center] {Mass \\ conservation};
        \node (outer) [boxUWD, inner xsep=1em, inner ysep=1.5em, fit=(C) (F) (transport) (conservation)] {};
        \path
            (transport.15) edge[out=0, in=180] (C)
            (transport.345) edge[out=0, in=180] (F)
            (conservation.165) edge[out=180, in=0] (C)
            (conservation.195) edge[out=180, in=0] (F)
            (outer.north) edge[out=270, in=90] (C.north);
    \end{tikzpicture}
\end{equation}
which has two boxes (each with two ports), two junctions, and one outer port. The ports and junctions are all typed by discrete diagrams---here singleton diagrams---but for brevity the types are shown only on the junctions. Note that the textual labels on the boxes have no formal meaning, but in all other respects the picture denotes a meaningful mathematical object. Applying the UWD \eqref{eq:transport-uwd} to the open diagrams \eqref{eq:diffusion-flux-diagram} and \eqref{eq:transport-conservation-diagram} yields an open diagram for the diffusion equation \eqref{eq:diffusion-diagram}, with the concentration $C$ as the only exposed quantity.
\end{example}

The modular construction of a physical theory using UWDs and open diagrams allows certain components of the theory to be replaced or extended without affecting the others, as demonstrated by the next example.

\begin{example}[Advection] \label{ex:advection}
To construct an advection system following the pattern of \cref{ex:diffusion-compositional}, we need only replace the open diagram plugged into the ``transport flux'' box, previously given by Fick's first law of diffusion, with an open diagram describing the flux due to advection along a moving fluid.

Let $\bm{v} \in \VectField_t(M)$ be a fixed time-dependent velocity field for the fluid advecting the substance. If $\widetilde C := \star C \in \widetilde\Omega_t^3$ is the density corresponding to the 0-form $C \in \Omega_t^0$, then the flux of $C$ due to advection is simply $\iota_{\bm v} \widetilde C$, where $\iota_{\bm v}: \widetilde\Omega_t^3 \to \widetilde\Omega_t^2$ is the \emph{interior product} with respect to $\bm v$ \cite[Definition 6.4.7]{abraham1988}. In Euclidean space $M = \R^3$, using the various isomorphisms between forms and vector fields, this term can be identified with the vector-scalar product $\bm{v} C$, which is possibly more familiar. In any case, an advection equation on a three-dimensional Riemannian manifold $M$ is obtained by applying the UWD \eqref{eq:transport-uwd} to the open diagram
\begin{equation} \label{eq:advection-flux-diagram}
    \begin{tikzcd}
    	{C: \Omega_t^0} & {\widetilde C: \widetilde\Omega_t^3} \\
    	& {\phi: \widetilde\Omega_t^2}
    	\arrow["\star", from=1-1, to=1-2]
    	\arrow["{-\iota_{\mathbf{v}}}", from=1-2, to=2-2]
    \end{tikzcd}
\end{equation}
with $C$ and $\phi$ exposed, along with the open diagram \eqref{eq:transport-conservation-diagram} for mass conservation. In traditional notation, the resulting differential equation is
\begin{equation*}
    \partial_t C = -\star^{-1} d \iota_{\bm v} \star C
\end{equation*}
or equivalently
\begin{equation*}
    \partial_t \widetilde C = -d \iota_{\bm v} \widetilde C = - \mathcal{L}_{\bm v} \widetilde C,
\end{equation*}
where $\mathcal{L}_{\bm v}$ denotes the Lie derivative with respect to $\bm v$, having used \emph{Cartan's magic formula} $\mathcal{L}_{\bm v} = d \circ \iota_{\bm v} + \iota_{\bm v} \circ d$ \cite[Theorem 6.4.8 (iv)]{abraham1988}. The equation $\partial_t \widetilde C + \mathcal{L}_{\bm v} \widetilde C = 0$ has been called the \emph{Lie advection equation} \cite{mullen2011}.
\end{example}

The UWD algebra of open diagrams enables not only modular but also hierarchical construction of physical theories, as illustrated by the following two-level construction of the advection-diffusion equation.

\begin{example}[Advection-diffusion]
An advection-diffusion system combines diffusion of a substance with advection of the substance by a moving fluid. The fluxes involved are described the UWD
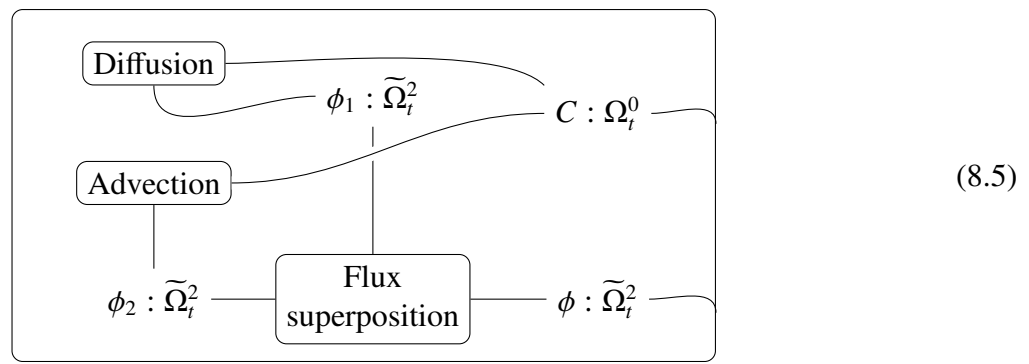
\begin{equation} \label{eq:advection-diffusion-flux-uwd}
    \begin{tikzpicture}[baseline=(current bounding box.center),tips=proper]
        \node (C) {$C: \Omega_t^0$};
        \node (F) [below=4em of C] {$\phi: \widetilde\Omega_t^2$};
        \node (CF) [fit=(C) (F), inner sep=0] {};
        \node (superposition) [boxUWD, left=of F, align=center] {Flux \\ superposition};
        \node (F1) [above=4em of superposition] {$\phi_1: \widetilde\Omega_t^2$};
        \node (F2) [left=2em of superposition] {$\phi_2: \widetilde\Omega_t^2$};
        \node (advection) [boxUWD, above=2em of F2] {Advection};
        \node (diffusion) [boxUWD, above=of advection] {Diffusion};
        \node (outer) [boxUWD, inner xsep=2em, inner ysep=1em, fit=(C) (F) (diffusion) (advection)] {};
        \path
            (diffusion.east) edge[out=0, looseness=0.5] (C.north west)
            (diffusion.south) edge[out=270, in=180] (F1.west)
            (advection.south) edge (F2)
            (superposition) edge (F)
            (superposition) edge (F1)
            (superposition) edge (F2)
            (outer.10) edge[out=90, in=0] (C.east)
            (outer.340) edge[out=90, in=0] (F.east);
        \path[draw,line width=5pt,white]
            (advection.east) edge[out=0, in=180] (C.west);
        \path
            (advection.east) edge[out=0, in=180] (C.west);
    \end{tikzpicture}
\end{equation}
where the box labeled ``flux superposition'' is intended to be filled by the open cartesian diagram
\begin{equation*}
    \begin{tikzcd}[row sep=tiny]
    	{\phi_1: \widetilde\Omega_t^2} \\
    	& {(\phi_1, \phi_2): \widetilde\Omega_t^2 \oplus \widetilde\Omega_t^2} & {\phi: \widetilde\Omega_t^2} \\
    	{\phi_2: \widetilde\Omega_t^2}
    	\arrow["{+}", from=2-2, to=2-3]
    	\arrow["{\pi_1}"', from=2-2, to=1-1]
    	\arrow["{\pi_2}", from=2-2, to=3-1]
    \end{tikzcd}
\end{equation*}
with all of $\phi_1$, $\phi_2$, and $\phi$ exposed. It may seem excessive to glorify mere addition with the name ``superposition'' but the addivity of the fluxes does have the physically meaningful interpretation that the two sources of flux are independent of each other.

Applying the UWD \eqref{eq:advection-diffusion-flux-uwd} to the superposition diagram as well as to the diffusion and advection flux diagrams \eqref{eq:diffusion-flux-diagram} and \eqref{eq:advection-flux-diagram}, all now regarded as open cartesian diagrams, we obtain an open cartesian diagram exposing $C$ and $\phi$ and encoding the equation
\begin{equation*}
    \phi = \phi_1 + \phi_2 = k \star dC + \iota_{\bm v} \star C
\end{equation*}
for the total negative flux. We can then apply the original UWD \eqref{eq:transport-uwd} to this diagram and the mass conservation diagram \eqref{eq:transport-conservation-diagram} to obtain an open cartesian diagram exposing $C$ and encoding the advection-diffusion equation. Alternatively, we can operadically compose the original UWD \eqref{eq:transport-uwd} with the UWD \eqref{eq:advection-diffusion-flux-uwd} to obtain the enlarged UWD
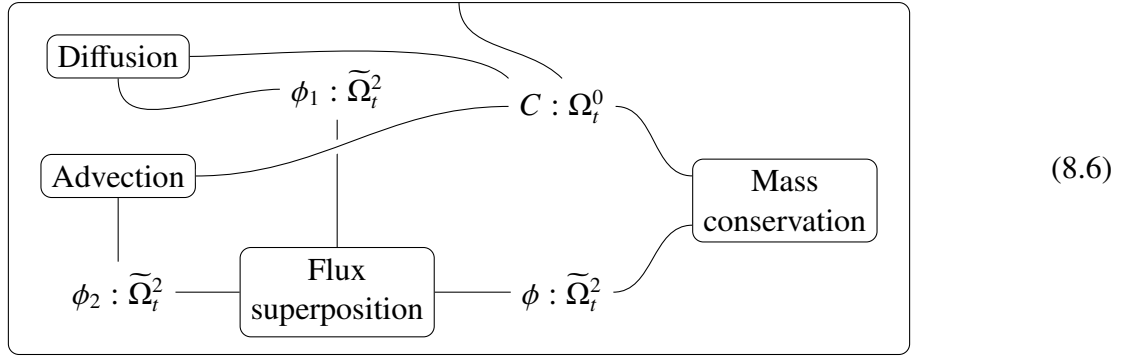
\begin{equation} \label{eq:advection-diffusion-uwd}
    \begin{tikzpicture}[baseline=(current bounding box.center),tips=proper]
        \node (C) {$C: \Omega_t^0$};
        \node (F) [below=4em of C] {$\phi: \widetilde\Omega_t^2$};
        \node (CF) [fit=(C) (F), inner sep=0] {};
        \node (superposition) [boxUWD, left=of F, align=center] {Flux \\ superposition};
        \node (F1) [above=4em of superposition] {$\phi_1: \widetilde\Omega_t^2$};
        \node (F2) [left=2em of superposition] {$\phi_2: \widetilde\Omega_t^2$};
        \node (advection) [boxUWD, above=2em of F2] {Advection};
        \node (diffusion) [boxUWD, above=of advection] {Diffusion};
        \node (conservation) [boxUWD, right=of CF, align=center] {Mass \\ conservation};
        \node (outer) [boxUWD, inner sep=1em, fit=(C) (F) (diffusion) (advection) (conservation)] {};
        \path
            (diffusion.east) edge[out=0, looseness=0.5] (C.north west)
            (diffusion.south) edge[out=270, in=180] (F1.west)
            (advection.south) edge (F2)
            (conservation.165) edge[out=180, in=0] (C)
            (conservation.195) edge[out=180, in=0] (F)
            (superposition) edge (F)
            (superposition) edge (F1)
            (superposition) edge (F2)
            (outer.north) edge[out=270] (C.north);
        \path[draw,line width=5pt,white]
            (advection.east) edge[out=0, in=180] (C.west);
        \path
            (advection.east) edge[out=0, in=180] (C.west);
    \end{tikzpicture}
\end{equation}
and then apply it to the same set of open cartesian diagrams. The result, which is the same either way due to the associativity law of an operad algebra, is an open cartesian diagram presenting the advection-diffusion equation
\begin{equation*}
    \partial_t \widetilde C = d (k\star) dC - \mathcal{L}_{\bm v} \widetilde C
    \qquad\text{where}\qquad
    \widetilde C := \star C.
\end{equation*}
This equation, compact though it is, already contains substantial physical content, which is formally represented by its hierarchical construction.
\end{example}

Nesting of undirected wiring diagrams may be repeated indefinitely, as hinted at in this section's final example, which involves a three-level hierarchy.

\begin{example}[Advection-diffusion-reaction]
An advection-reaction-diffusion system combines the simultaneous advection and diffusion of one or more substances in a moving fluid with chemical reactions or other transformations between the substances \cite[Chapter 4]{socolofsky2005}. For simplicity, we consider only a single substance undergoing a single reaction, declining to model the other substances involved in the reaction.

Since a reacting system has inflows and/or outflows of substance due to reaction, we replace the open diagram \eqref{eq:transport-conservation-diagram} for the conservation of mass in a \emph{closed} system with one for conservation in an \emph{open} system, namely the open cartesian diagram
\begin{equation*}
    \begin{tikzcd}
    	& {(\dt C_\mathrm{flux}, S): (\Omega_t^0)^2} & {S: \Omega_t^0} \\
    	{C: \Omega_t^0} & {\dt C: \Omega_t^0} & {\dt C_{\mathrm{flux}}: \Omega_t^0} & {d\phi: \widetilde\Omega_t^3} & {\phi: \widetilde\Omega_t^2}
    	\arrow["{\pi_2}", from=1-2, to=1-3]
    	\arrow["{\pi_1}"', from=1-2, to=2-3]
    	\arrow["{+}"', from=1-2, to=2-2]
    	\arrow["{\partial_t}", from=2-1, to=2-2]
    	\arrow["{\star^{-1}}"', from=2-4, to=2-3]
    	\arrow["d"', from=2-5, to=2-4]
    \end{tikzcd}
\end{equation*}
exposing the concentration $C$, negative flux $\phi$, and source or sink of substance $S$. Likewise, we replace the UWD \eqref{eq:transport-uwd} with the following:
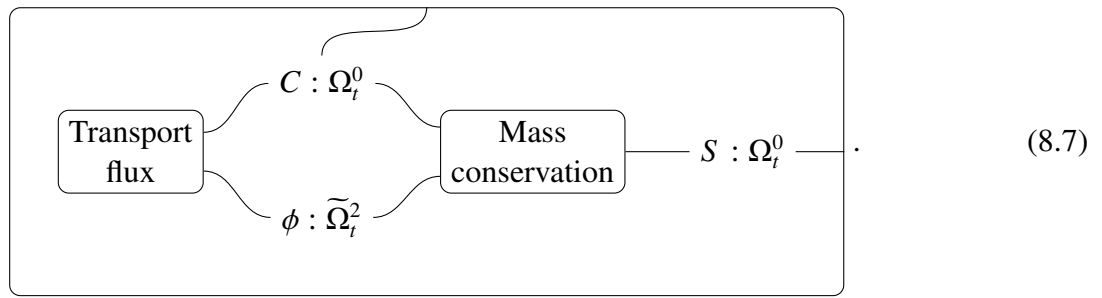
\begin{equation} \label{eq:open-transport-uwd}
    \begin{tikzpicture}[baseline=(current bounding box.center),tips=proper]
        \node (C) {$C: \Omega_t^0$};
        \node (F) [below=of C] {$\phi: \widetilde\Omega_t^2$};
        \node (CF) [fit=(C) (F), inner sep=0] {};
        \node (transport) [boxUWD, left=2em of CF, align=center] {Transport \\ flux};
        \node (conservation) [boxUWD, right=2em of CF, align=center] {Mass \\ conservation};
        \node (S) [right=2em of conservation] {$S: \Omega_t^0$};
        \node (outer) [boxUWD, inner sep=1.5em, fit=(C) (F) (S) (transport) (conservation)] {};
        \path
            (transport.15) edge[out=0, in=180] (C)
            (transport.345) edge[out=0, in=180] (F)
            (conservation.165) edge[out=180, in=0] (C)
            (conservation.195) edge[out=180, in=0] (F)
            (conservation.east) edge (S)
            (outer.east) edge (S)
            (outer.north) edge[out=270, in=90] (C.north);
    \end{tikzpicture}\ .
\end{equation}
We can then proceed as in the previous examples. For instance, applying the new UWD to the diffusion flux diagram \eqref{eq:diffusion-flux-diagram} yields the nonhomogeneous diffusion equation, which was defined directly in \cref{ex:diffusion-with-source}. We can also compose with the UWD \eqref{eq:advection-diffusion-uwd} to obtain a more refined UWD for advection-diffusion with a source or sink.

To incorporate reaction, we introduce a yet higher level of hierarchy via the UWD
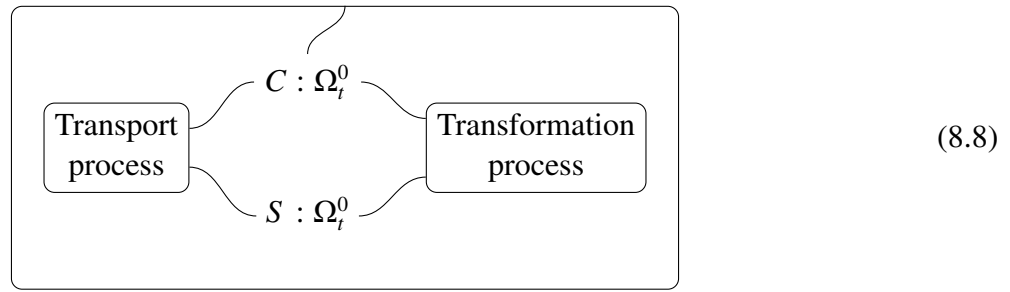
\begin{equation} \label{eq:transport-transformation-uwd}
    \begin{tikzpicture}[baseline=(current bounding box.center),tips=proper]
        \node (C) {$C: \Omega_t^0$};
        \node (S) [below=of C] {$S: \Omega_t^0$};
        \node (CS) [fit=(C) (S), inner sep=0] {};
        \node (transport) [boxUWD, left=2em of CS, align=center] {Transport \\ process};
        \node (transformation) [boxUWD, right=2em of CS, align=center] {Transformation \\ process};
        \node (outer) [boxUWD, inner xsep=1em, inner ysep=1.5em, fit=(C) (S) (transport) (transformation)] {};
        \path
            (transport.15) edge[out=0, in=180] (C)
            (transport.345) edge[out=0, in=180] (F)
            (transformation.165) edge[out=180, in=0] (C)
            (transformation.195) edge[out=180, in=0] (F)
            (outer.north) edge[out=270, in=90] (C.north);
    \end{tikzpicture}
\end{equation}
coupling mass transport with reaction or transformation processes. Into the box labeled ``transport process'' we expect to substitute the UWD \eqref{eq:open-transport-uwd} or any of its concrete applications sketched above. As for the ``transformation process'' box, a particularly simple choice is the open diagram
\begin{equation*}
    \begin{tikzcd}
        C: \Omega_t^0 \arrow[r, "\beta"] & S: \Omega_t^0
    \end{tikzcd}
\end{equation*}
for a first-order growth or decay, according to the sign of the rate constant $\beta \in \R$. Any additional substances which may be driving this process are left unmodeled. Applying the UWD \eqref{eq:transport-transformation-uwd} to the appropriate open cartesian diagrams yields the basic example
\begin{equation*}
    \partial_t \widetilde C = d (k\star) dC - \mathcal{L}_{\bm v} \widetilde C + \beta \widetilde C,
    \qquad\text{where}\qquad
    \widetilde C := \star C,
\end{equation*}
of an advection-diffusion-reaction equation on a Riemannian manifold.
\end{example}

\section{Weak equivalence of diagrams} \label{sec:weak-equivalences}

It is possible for nonisomorphic diagrams to present physical theories that are equivalent in the sense of having interchangeable solutions, as shown by the two different presentations of the heat equation in \cref{ex:heat-equation-morphisms}. In this section, we make a preliminary study of such ``weak equivalences'' of diagrams: morphisms of diagrams that establish one-to-one correspondences between lifts of their domain and codomain diagrams.

The foundational \cref{thm:diagram-opfibration} established that lifts of a diagram $D$ in a category $\cat{C}$ through a discrete opfibration can be pushed forward along any morphism out of $D$ in the diagram category $\DiagOp(\cat{C})$. Understanding when morphisms in the dual category $\Diag(\cat{C})$ have the same property is essential to characterizing weak equivalence. To state this property precisely, it is helpful to weaken the defining condition of a discrete opfibration (\cref{def:discrete-opfibration}) so that it does not apply everywhere. Let us say, using a nonstandard but self-explanatory phrase, that a functor $\pi: \cat{E} \to \cat{C}$ is a \emph{discrete opfibration at a morphism $f: x \to y$} in $\cat{C}$ if, for every object $\overline x \in \cat{E}$ with $\pi(\overline x) = x$, there exists a unique morphism $\overline f: \overline x \to \overline y$ in $\cat{E}$ such that $\pi(\overline f) = f$. Thus, a functor into $\cat{C}$ is a discrete opfibration in the usual sense exactly when it is discrete opfibration at every morphism in $\cat{C}$.

\begin{definition}[Pushforward property] \label{def:pushforward-property}
Given a category $\cat{C}$, a morphism $(R,\rho)$ in the diagram category $\Diag(\cat{C})$ has the \emph{pushforward property} if, for any discrete opfibration $\pi: \cat{E} \to \cat{C}$, the functor $\Diag(\pi): \Diag(\cat{E}) \to \Diag(\cat{C})$ is a discrete opfibration at $(R,\rho)$.
\end{definition}

\begin{lemma}[Closure of pushforward property]
Let $\cat{C}$ be a category.
\begin{enumerate}[(i),nosep]
\item The class of morphisms in $\Diag(\cat{C})$ having the pushforward property forms a wide subcategory of $\Diag(\cat{C})$.
\item For any discrete opfibration $\pi: \cat{E} \to \cat{C}$, the pushforward property is preserved by lifts through $\Diag(\pi)$.
\end{enumerate}
\end{lemma}
\begin{proof}[Sketch of proof]
To prove part (i), suppose $D \xrightarrow{(R,\rho)} E \xrightarrow{(S,\sigma)} F$ are composable morphisms in $\Diag(\cat{C})$ having the pushforward property. Let $\pi: \cat{E} \to \cat{C}$ be a discrete opfibration and $\overline D$ be a lift of $D$ through $\pi$. A lift of the composite $(T,\tau) := (R,\rho) \cdot (S,\sigma)$ through $\Diag(\pi)$ with domain $\overline D$ is obtained by composing the lift $(R,\overline\rho): \overline D \to \overline E$ of $(R,\rho)$ at $\overline D$ with the subsequent lift $(S,\overline\sigma): \overline E \to \overline F$ of $(S,\sigma)$ at $\overline E$. As for uniqueness, any lift $(T,\overline{\tau})$ of the composite $(T,\tau)$ with domain $\overline D$ must factorize as $(R,\overline\rho) \cdot (S,\overline\sigma)$, as seen by applying the unique lifting property of $\pi$ to the components of $\overline\tau$. That identities in $\Diag(\cat{C})$ have the pushforward property is easy to see, and also follows from stronger results below.

For part (ii), suppose $(R,\overline\rho)$ is a lift through $\Diag(\pi)$ of a morphism $(R,\rho)$ having the pushforward property. Since the composite of another discrete opfibration $\mu: \cat{F} \to \cat{E}$ with $\pi: \cat{E} \to \cat{C}$ is again a discrete opfibration, it straightforward to show that $(R,\overline\rho)$ has the pushforward property using the fact that $(R,\rho)$ does.
\end{proof}

A weak equivalence of diagrams is a special kind of diagram morphism having the pushforward property, namely one with an invertible 2-cell. Recall that diagram morphism $(R,\rho)$ is said to be \emph{strong} if $\rho$ is invertible, and \emph{strict} if $\rho$ is an identity.

\begin{definition}[Weak equivalence] \label{def:weak-equivalence}
A \emph{weak equivalence} between diagrams $D$ and $D'$ in a category $\cat{C}$ is a strong morphism $(R,\rho): D \to D'$ in $\DiagOp(\cat{C})$ such that the morphism $(R,\rho^{-1}): D' \to D$ in $\Diag(\cat{C})$ has the pushforward property.
\end{definition}

It follows from the lemma above that, for any category $\cat{C}$, weak equivalences of diagrams in $\cat{C}$ form a wide subcategory of $\DiagOp(\cat{C})$, and for any discrete opfibration $\pi: \cat{E} \to \cat{C}$, weak equivalences are preserved by lifts through $\DiagOp(\pi)$.

The property of defining a bijection between lifts, which weak equivalences are expected to exhibit, is captured by the concept of a discrete bifibration. A functor $\pi: \cat{E} \to \cat{C}$ is a \emph{discrete fibration} if its opposite $\pi^\op: \cat{E}^\op \to \cat{C}^\op$ is a discrete opfibration; explicitly, $\pi: \cat{E} \to \cat{C}$ is a discrete fibration if for every morphism $f: x \to y$ in $\cat{C}$ and every object $\overline y \in \pi^{-1}(y)$, there exists a unique morphism $\overline f: \overline x \to \overline y$ in $\cat{E}$ such that $\pi(\overline f) = f$. A functor is a \emph{discrete bifibration} if it is both a discrete fibration and a discrete opfibration. One can similarly define the notion of a functor into $\cat{C}$ that is a \emph{discrete (bi)fibration at some morphism $f:x\to y$} in $\cat{C}$.

With this terminology, the definition of a weak equivalence is justified by:

\begin{proposition}[Weak equivalences biject lifts] \label{prop:weak-equivalence-bifibration}
If $\pi: \cat{E} \to \cat{C}$ is a discrete obfibration, then the functor $\DiagOp(\pi): \DiagOp(\cat{E}) \to \DiagOp(\cat{C})$ is a discrete bifibration at any weak equivalence of diagrams in $\cat{C}$.
\end{proposition}
\begin{proof}
First, observe that the functor $\DiagOp^{\text{2-iso}}(\pi): \DiagOp^{\text{2-iso}}(\cat{E}) \to \DiagOp^{\text{2-iso}}(\cat{C})$ restricting $\DiagOp(\pi)$ to strong morphisms will be a discrete (op or bi)fibration at such a morphism if and only if the unrestricted functor $\DiagOp(\pi)$ is, because a discrete opfibration $\pi: \cat{E} \to \cat{C}$ reflects isomorphisms. A similar statement applies to the functor $\Diag^{\text{2-iso}}$, and the map $(R,\rho) \mapsto (R,\rho^{-1})$ gives an isomorphism $\DiagOp^{\text{2-iso}}(\cat{C})^\op \cong \Diag^{\text{2-iso}}(\cat{C})$.

Let $(R,\rho): D \to D'$ be a weak equivalence of diagrams in $\cat{C}$. On one hand, the functor $\DiagOp(\pi)$ is a discrete opfibration at $(R,\rho)$ by \cref{thm:diagram-opfibration}. On the other hand, $\Diag(\pi)$ is a discrete opfibration at $(R,\rho^{-1}): D' \to D$ by the definition of a weak equivalence; equivalently, using the observation above, $\DiagOp(\pi)^\op$ is a discrete opfibration at $(R,\rho)$. We conclude that $\DiagOp(\pi)$ is a discrete bifibration at $(R,\rho)$.
\end{proof}

Essentially by definition, diagram morphisms satisfying the pushforward property give a statement dual to \cref{thm:diagram-opfibration}, and weak equivalences of diagrams establish a bijection between lifts through any discrete opfibration. However, these definitions are of little utility unless they can be verified in concrete situations. To that end, the rest of this section is dedicated to finding sufficient conditions for the pushforward property, and hence for weak equivalences, that can be checked mechanically in particular examples.

{\bf Initial functors.} A simple criterion, not the most general possible, for a diagram morphism to have the pushforward property is that the functor between indexing categories be \emph{initial} \cite[\S IX.3]{maclane1978}, \cite[\S 8.3]{riehl2014}. This important fact, which is a special case of \cref{thm:diagram-initial-opfibration} below, will not be immediately apparent from the definition:

\begin{definition}[Initial functor] \label{def:initial-functor}
A functor $R: \cat{J} \to \cat{J}'$ is \emph{initial} if, for every object $j' \in \cat{J}'$, the comma category $R/j'$ is nonempty and connected.
\end{definition}

\noindent Recall that the \emph{comma category} $R/j'$ has as objects, the pairs $(j,f')$, where $j \in \cat{J}$ and $f': Rj \to j'$ in $\cat{J}'$, and as morphisms $(j,f') \to (k,g')$, the morphisms $h: j \to k$ in $\cat{J}$ making the triangle commute:
\begin{equation*}
  \begin{tikzcd}
    Rj && Rk \\
    & {j'}
    \arrow["Rh", from=1-1, to=1-3]
    \arrow["{g'}", from=1-3, to=2-2]
    \arrow["{f'}"', from=1-1, to=2-2]
  \end{tikzcd}.
\end{equation*}
Also, recall that a category is \emph{connected} whenever any pair of objects can be joined by a finite zig-zag of morphisms.

\begin{example}[Weak equivalence from initial functor]
Recall from \cref{ex:heat-equation-morphisms} the morphism in $\DiagOp(\Sh_\R(M \times [0,\infty)))$
\begin{equation*}
    \left(
    \begin{tikzcd}
    	{C: \Omega_t^0} & {\dt C: \Omega_t^0} & {d\phi: \widetilde \Omega_t^3} \\
    	{dC: \Omega_t^1} & {} & {\phi: \widetilde \Omega_t^2}
    	\arrow["d"', from=1-1, to=2-1]
    	\arrow["{\partial_t}", from=1-1, to=1-2]
    	\arrow["k\star"', from=2-1, to=2-3]
    	\arrow["d"', from=2-3, to=1-3]
    	\arrow["{\star^{-1}}"', from=1-3, to=1-2]
    \end{tikzcd}
    \right)
    \quad\longrightarrow\quad
    \left(
    \begin{tikzcd}
    	{C: \Omega_t^0} & {\dt C: \Omega_t^0}
    	\arrow["{k \Delta}"', shift right=1, from=1-1, to=1-2]
    	\arrow["{\partial_t}", shift left=1, from=1-1, to=1-2]
    \end{tikzcd}
    \right)
\end{equation*}
between presentations of the heat equation, which is well-defined when $k$ is constant over the manifold $M$. The diagram morphism is strict and its map between indexing categories is initial, as can be checked directly from the definition. Therefore, by the result to be shown below, the morphism of diagrams is a weak equivalence of diagrams, as anticipated in \cref{ex:heat-equation-morphisms}.
\end{example}

The definition of an initial functor in terms of comma categories is concrete and often easy to check in particular situations, but provides little conceptual insight. Initial functors are motivated by their classical characterization as the functors that preserve limits under change of indexing category. Specifically, for any functor $R: \cat{J}' \to \cat{J}$ and diagram $D: \cat{J} \to \cat{C}$, there is a canonical morphism
\begin{equation} \label{eq:initial-functor-limits}
    \lim_{\cat{J}} D \to \lim_{\cat{J}'} D \circ R
\end{equation}
in $\cat{C}$, which exists whenever the limits involved do and results from applying \cref{prop:functorality-limits} to the morphism $(R, \id_{D \circ R}): (\cat{J}, D) \to (\cat{J}', D \circ R)$ in $\DiagOp(\cat{C})$.

\begin{proposition}[Initial functors preserve limits] \label{prop:initial-functor-limits}
A functor $R: \cat{J}' \to \cat{J}$ is initial if and only if, for any category $\cat{C}$ and diagram $D: \cat{J} \to \cat{C}$, the canonical morphism \eqref{eq:initial-functor-limits} is an isomorphism whenever the limits involved exist.

The same statement holds when the category $\cat{C}$ is specialized to $\Set$.
\end{proposition}

\noindent The proof of this proposition, which is not essential for our purposes, can be found in \cite[Theorem IX.3.1]{maclane1978} or \cite[Lemma 8.3.4]{riehl2014}. However, given the connection between lifts and cones (\cref{prop:lifts-and-cones}), the result is suggestive of why initial functors are relevant to the lifting properties of diagram morphisms.

A useful sufficient condition for a functor to be initial is:

\begin{lemma}[Full initial functors] \label{lem:full-initial-functor}
Let $R: \cat{J} \to \cat{J}'$ be a full functor. Then, for $R$ to be initial, it is enough that the comma category $R/j'$ be nonempty and connected for every object $j \in \cat{J}'$ not in the essential image of $R$.

In particular, any full and essentially surjective functor is initial.
\end{lemma}
\begin{proof}
Suppose $j' \in \cat{J}'$ is in the essential image of $R$, so that we can find an object $k \in \cat{J}$ and isomorphism $m: Rk \xrightarrow{\cong} j'$. The category $R/j'$ is nonempty because it contains the object $(k, m)$. Suppose $(j, f')$ is another object of $R/j'$. Since $R$ is full, the composite
\begin{equation*}
    Rj \xrightarrow{f'} j' \xrightarrow{m^{-1}} Rk
\end{equation*}
is the image under $R$ of some morphism $f: j \to k$ in $\cat{J}$ and hence $f: (j,f') \to (k,m)$ is a morphism in $R/j'$. This shows that $(k,m)$ is weakly terminal in $R/j'$. In particular, the category $R/j'$ is connected.
\end{proof}

\begin{lemma}[Closure of initial functors] \label{lem:closure-initial-functor}
  \leavevmode
  \begin{enumerate}[(i),nosep]
    \item Composites of initial functors are initial.
    \item Equivalences of categories are initial.
  \end{enumerate}
  In particular, the initial functors comprise a wide subcategory of $\Cat$.
\end{lemma}
\begin{proof}
  Part (i) is \cite[Exercise IX.3.2]{maclane1978}. It is also a special case of \cref{lem:closure-relatively-initial-map} proved below. Part (ii) follows from from \cref{lem:full-initial-functor}, since equivalences of categories are full and essentially surjective.
\end{proof}

{\bf Relatively initial morphisms of diagrams.} The condition that a morphism of diagrams have an initial functor between the indexing categories is stronger than it needs to be, as it does not even make reference to the diagrams themselves or the transformation between them. This defect is remedied by the following definition, which seems to be original.

\begin{definition}[Relatively initial map]
Let $\cat{C}$ be a category and $(R,\rho): (\cat{J}, D) \to (\cat{J}',D')$ be a morphism in $\Diag(\cat{C})$.

For any $j' \in \cat{J}'$, define a category $(R,\rho)/j'$ that has as objects, the pairs $(j,f')$, where $j \in \cat{J}$ and $f': Rj \to j'$ in $\cat{J}'$, and as morphisms $(j,f') \to (k,g')$, the morphisms $h:j \to k$ in $\cat{J}$ making the pentagon commute:
\begin{equation*}
    \begin{tikzcd}
    	Dj && Dk \\
    	{D'Rj} && {D'Rk} \\
    	& {D'j'}
    	\arrow["{\rho_j}"', from=1-1, to=2-1]
    	\arrow["Dh", from=1-1, to=1-3]
    	\arrow["{\rho_k}", from=1-3, to=2-3]
    	\arrow["{D'f'}"', from=2-1, to=3-2]
    	\arrow["{D'g'}", from=2-3, to=3-2]
    \end{tikzcd}
\end{equation*}
Composition and identities in $(R,\rho)/j'$ are defined as in $\cat{J}$, and so there is a projection functor $(R,\rho)/j' \to \cat{J}$.

We say that the functor $R: \cat{J} \to \cat{J}'$ is \emph{initial relative to $\rho$} or that the diagram morphism $(R,\rho)$ is \emph{relatively initial} if, for every object $j' \in \cat{J}'$, the category $(R,\rho)/j'$ is nonempty and connected.
\end{definition}

For any $j' \in \cat{J}'$, the category $(R,\rho)/j'$ has the same objects as the comma category $R/j'$ but generally has more morphisms, since if $h: j \to k$ is a morphism from $(j,f')$ to $(k,g')$ in $R/j'$, then the naturality of $\rho$ gives a commuting diagram
\begin{equation*}
    \begin{tikzcd}
    	Dj && Dk \\
    	{D'Rj} && {D'Rk} \\
    	& {D'j'}
    	\arrow["{\rho_j}"', from=1-1, to=2-1]
    	\arrow["Dh", from=1-1, to=1-3]
    	\arrow["{\rho_k}", from=1-3, to=2-3]
    	\arrow["{D'f'}"', from=2-1, to=3-2]
    	\arrow["{D'g'}", from=2-3, to=3-2]
    	\arrow["{D'Rh}", from=2-1, to=2-3]
    \end{tikzcd}
\end{equation*}
and hence a morphism from $(j,f')$ to $(k,g')$ in $(R,\rho)/j'$. Thus, being relatively initial is weaker than being initial, and any extension of an initial functor $R: \cat{J} \to \cat{J}'$ to a morphism $(R,\rho): (\cat{J},D) \to (\cat{J}',D')$ in $\Diag(\cat{C})$ is relatively initial. On the other hand, initiality is formally a special case of relative initiality, since a functor $R: \cat{J} \to \cat{J}'$ is initial if and only $R$ is initial relative to the identity $\id_R$, i.e., if and only if the morphism $(R, \id_R): (\cat{J}, R) \to (\cat{J}', \id_{\cat{J}'})$ in $\Diag(\cat{J}')$ is relatively initial.

\begin{example}[Relatively initial but not initial]
Inspired by the Lie advection equation of \cref{ex:advection}, let $\bm v$ be a vector field on a $n$-dimensional smooth manifold $M$. The Lie derivative with respect to $\bm v$ satisfies $\mathcal{L}_{\bm v} = d \circ \iota_{\bm v}$ when restricted to densities on $M$ (twisted $n$-dimensional forms on $M$). Thus, there is a strict morphism in $\Diag(\Sh_\R(M))$
\begin{equation*}
    \left(
    \begin{tikzcd}
    	{\widetilde\Omega^n} \\
    	{\widetilde\Omega^n}
    	\arrow["{\mathcal{L}_{\mathbf{v}}}"', from=1-1, to=2-1]
    \end{tikzcd}
    \right)
    \qquad\to\qquad
    \left(
    \begin{tikzcd}
    	{\widetilde\Omega^n} & {\widetilde\Omega^{n-1}} \\
    	{\widetilde\Omega^n}
    	\arrow["{\iota_{\mathbf{v}}}", from=1-1, to=1-2]
    	\arrow["d", from=1-2, to=2-1]
    	\arrow["{\mathcal{L}_{\mathbf{v}}}"', from=1-1, to=2-1]
    \end{tikzcd}
    \right)
\end{equation*}
between free diagrams, where the map between indexing categories is the obvious inclusion. Because the indexing categories are free, the map between them is not initial, but because the diagrams commute, the diagram morphism is relatively initial. The next theorem, the main result of this section, shows that the two diagrams are weakly equivalent, as one would expect.
\end{example}

\begin{theorem}[Lifting relatively initial morphisms of diagrams] \label{thm:diagram-initial-opfibration}
Any relatively initial morphism of diagrams has the pushforward property.

Consequently, if $\cat{C}$ is a category and $(R,\rho): D \to D'$ is a strong morphism in $\DiagOp(\cat{C})$ such that the morphism $(R,\rho^{-1}): D' \to D$ in $\Diag(\cat{C})$ is relatively initial, then $(R,\rho)$ is a weak equivalence of diagrams.
\end{theorem}
\begin{proof}
Let $(R, \rho): (\cat{J}, D) \to (\cat{J}', D')$ in $\Diag(\cat{C})$ be a relatively initial morphism of diagrams in a category $\cat{C}$. We must show that, for any discrete opfibration $\pi: \cat{E} \to \cat{C}$ and any lift $\overline D$ of $D$ through $\pi$, there exists a unique morphism $(R, \overline\rho): (\cat{J}, \overline D) \to (\cat{J}', \overline D')$ in $\Diag(\cat{E})$ such that $\Diag(\pi)(R, \overline\rho) = (R, \rho)$.
\begin{equation*}
    \begin{tikzcd}[row sep=large, column sep=large]
    	{\mathsf{J}} & {\mathsf{E}} \\
    	{\mathsf{J}'} & {\mathsf{C}}
    	\arrow["R"', from=1-1, to=2-1]
    	\arrow[""{name=0, anchor=center, inner sep=0}, "{\overline D}", from=1-1, to=1-2]
    	\arrow[""{name=1, anchor=center, inner sep=0}, "{\overline D'}"', dashed, from=2-1, to=1-2]
    	\arrow["\pi", from=1-2, to=2-2]
    	\arrow["{D'}"', from=2-1, to=2-2]
    	\arrow["\overline\rho"', shorten <=2pt, shorten >=2pt, Rightarrow, dashed, from=0, to=1]
    \end{tikzcd}
    \qquad=\qquad
    \begin{tikzcd}[row sep=large, column sep=large]
    	{\mathsf{J}} \\
    	{\mathsf{J}'} & {\mathsf{C}}
    	\arrow["R"', from=1-1, to=2-1]
    	\arrow[""{name=0, anchor=center, inner sep=0}, "{D'}"', from=2-1, to=2-2]
    	\arrow[""{name=1, anchor=center, inner sep=0}, "D", from=1-1, to=2-2]
    	\arrow["\rho"', shorten <=2pt, shorten >=2pt, Rightarrow, from=1, to=0]
    \end{tikzcd}
\end{equation*}

First, we prove that the lifts $\overline{D}'$ and $\overline\rho: \overline D \To \overline{D}' \circ R$ are unique whenever they exist. Since $\pi: \cat{E} \to \cat{C}$ is a discrete opfibration, for each $j \in \cat{J}$, the object $\overline{D} j \in \cat{E}$ along with the morphism $\pi(\overline{D} j) = Dj \xrightarrow{\rho_j} D'Rj$ in $\cat{C}$ has a unique lift through $\pi$ to a morphism $\overline{D} j \xrightarrow{\overline\rho_j} e_j$ in $\cat{E}$, which must be the $j$ component of $\overline\rho$. This proves the uniqueness of $\overline\rho$. The diagram $\overline D'$ is determined on objects in the image of $R$ by the equations $\overline{D}' Rj = e_j$, $j \in \cat{J}$. Moreover, once $\overline{D}'$ is defined at an object $Rj$, it is determined on all morphisms in $\cat{J}'$ outgoing from $Rj$ (including their codomains), since $\overline{D}'$ must project to $D'$ and lifts through $\pi$ are unique. In particular, since $(R,\rho)/j'$ is nonempty for every $j' \in \cat{J}'$, the diagram $\overline D'$ is determined everywhere and so is unique if it exists.

It remains to show that both $\overline{D}'$ and and $\overline{\rho}$ actually exist, the proof of which fully uses the assumption that $(R,\rho)$ is relatively initial.

For any objects $j' \in \cat{J}'$ and $(j, f') \in (R,\rho)/j'$, the object $e_j \in \cat{E}$ along with the morphism $D'Rj \xrightarrow{D' f'} D'j'$ in $\cat{C}$ has a unique lift through $\pi$ to a morphism $e_j \xrightarrow{\overline f'} \bullet$ in $\cat{E}$. Note that, despite the notation, the lift $\overline f'$ depends not only on $f': Rj \to j'$ but also on $j$. However, if $f'$ is in the image of $R$, say $f' = Rf$ for $f: j \to k$ in $\cat{J}$, then the codomain of $\overline f'$ is $e_k$. To see this, note that the (possibly ill-formed) square
\begin{equation*}
    \begin{tikzcd}
    	{\overline{D} j} & {\overline{D} k} \\
    	{e_j} & {\bullet \overset{?}{=} e_k}
    	\arrow["{\overline\rho_j}"', from=1-1, to=2-1]
    	\arrow["{\overline{D} f}", from=1-1, to=1-2]
    	\arrow["{\overline\rho_k}", shift left=3, from=1-2, to=2-2]
    	\arrow["{\overline f'}", from=2-1, to=2-2]
    \end{tikzcd}
\end{equation*}
projects under $\pi: \cat{E} \to \cat{C}$ to the square
\begin{equation*}
    \begin{tikzcd}
    	Dj & Dk \\
    	{D'Rj} & {D'Rk}
    	\arrow["Df", from=1-1, to=1-2]
    	\arrow["{\rho_j}"', from=1-1, to=2-1]
    	\arrow["{\rho_k}", from=1-2, to=2-2]
    	\arrow["DRf", from=2-1, to=2-2]
    \end{tikzcd}.
\end{equation*}
But the latter square commutes, by the naturality of $\rho$, so the former square must also commute, by the uniqueness of lifts through $\pi$. In particular, the codomain of $\overline f'$ is $e_k$.

With these preliminaries, we can define $\overline D'$ on objects. For each $j' \in \cat{J}'$, the category $(R,\rho)/j'$ is nonempty, so choose an object $(j,f')$. We wish to define $\overline{D}' j'$ to be the codomain of $\overline f'$, but first we must check that the codomain depends only on $j'$, not on the choice of $(j, f') \in (R,\rho)/j'$. If $h: j \to k$ is a morphism in $(R,\rho)/j'$ from $(j,f')$ to $(k,g')$, then letting $h' = Rh$, the diagram
\begin{equation*}
    \begin{tikzcd}[row sep=scriptsize, column sep=small]
    	{\overline{D} j} && {e_j} \\
    	{e_j} && {e_k} \\
    	& {\bullet \overset{?}{=} \bullet}
    	\arrow["{\overline\rho_j}", from=1-1, to=1-3]
    	\arrow["{\overline\rho_j}"', from=1-1, to=2-1]
    	\arrow["{\overline h'}", from=1-3, to=2-3]
    	\arrow["{\overline f'}"', from=2-1, to=3-2]
    	\arrow["{\overline g'}", from=2-3, to=3-2]
    \end{tikzcd}
\end{equation*}
projects under $\pi: \cat{E} \to \cat{C}$ to
\begin{equation*}
    \begin{tikzcd}[row sep=scriptsize, column sep=small]
    	Dj && {D'Rj} \\
    	{D'Rj} && {D'Rk} \\
    	& {D'j'}
    	\arrow["{\rho_j}", from=1-1, to=1-3]
    	\arrow["{\rho_j}"', from=1-1, to=2-1]
    	\arrow["{D'Rh}", from=1-3, to=2-3]
    	\arrow["{D'f'}"', from=2-1, to=3-2]
    	\arrow["{D'g'}", from=2-3, to=3-2]
    \end{tikzcd}
    \qquad\leftrightsquigarrow\qquad
    \begin{tikzcd}[row sep=scriptsize, column sep=small]
    	Dj && Dk \\
    	{D'Rj} && {D'Rk} \\
    	& {D'j'}
    	\arrow["Dh", from=1-1, to=1-3]
    	\arrow["{\rho_j}"', from=1-1, to=2-1]
    	\arrow["{\rho_k}", from=1-3, to=2-3]
    	\arrow["{D'f'}"', from=2-1, to=3-2]
    	\arrow["{D'g'}", from=2-3, to=3-2]
    \end{tikzcd},
\end{equation*}
where the squiggly arrow indicates that the latter two diagrams are equivalent. Now the last diagram commutes, by definition of morphisms in $(R,\rho)/j'$, so the first diagram also commutes, by uniqueness of lifts. This implies that $\cod \overline f'$ is constant on connected components of $(R,\rho)/j'$. But $(R,\rho)/j'$ is connected by assumption, so $\cod \overline f'$ is constant on all of $(R,\rho)/j'$. Thus, we may define $\overline{D}' j' := \cod \overline f'$.

Having defined it on objects, it remains to define the diagram $\overline{D}'$ on morphisms. Given a morphism $h': j' \to k'$ in $\cat{J}'$, let $\overline h'$ be the unique lift of $D'h'$ through $\pi$ with domain $\overline{D}' j'$ and then define $\overline{D}' h' := \overline h'$. We must check that the codomain of $\overline h'$ is $\overline{D}' k'$. Taking any object $(j,f')$ in $(R,\rho)/j'$ and letting $g' = f' \cdot h'$, the pair $(j,g')$ is an object of $(R,\rho)/k'$, so comparing the lifts $\overline{D}' Rj \xrightarrow{\overline f'} \overline{D}' j' \xrightarrow{\overline h'} \bullet$ and $\overline{D}' Rj \xrightarrow{\overline g'} \overline{D}' k'$ proves that $\cod \overline h' = \overline{D}' k'$, as required.

We have shown that the data of the functor $\overline{D}': \cat{J'} \to \cat{E}$ and the transformation $\overline\rho: \overline{D} \To \overline{D}' \circ R$ exist and are unique. The naturality of $\overline\rho$ follows from an earlier argument and the functorality of $\overline{D}'$ follows, as usual, from uniqueness of lifts.
\end{proof}

Street and Walters established a similar result many years ago \cite{street1973}. As part of proving that initial functors and discrete opfibrations form the left and right parts of an orthogonal factorization system~\cite{riehl2008}, called the \emph{comprehensive factorization}, they show that an extension-lifting problem
\begin{equation*}
    \begin{tikzcd}
    	{\sf{J}} & {\sf{E}} \\
    	{\sf{J}'} & {\sf{C}}
    	\arrow["R"', from=1-1, to=2-1]
    	\arrow["{\overline D}", from=1-1, to=1-2]
    	\arrow["\pi", from=1-2, to=2-2]
    	\arrow["{D'}"', from=2-1, to=2-2]
    	\arrow["{\overline D'}", dashed, from=2-1, to=1-2]
    \end{tikzcd}
\end{equation*}
in $\Cat$, where $R$ is an initial functor and $\pi$ is a discrete opfibration, has a unique solution, indicated by the dashed arrow \cite[Theorem 4]{street1973}. This is the special case of \cref{thm:diagram-initial-opfibration} where the transformation $\rho$ is an identity and the functor $R$ is initial.

To close the section, we show that relative initiality has the properties expected of a generalization of initiality. The analogue of \cref{prop:initial-functor-limits}, that initial functors preserve limits, is:

\begin{proposition}[Relatively initial maps preserve limits]
Suppose $\cat{C}$ is a category and $(R,\rho): D \to D'$ is a strong morphism in $\DiagOp(\cat{C})$ such that $(R,\rho^{-1}): D' \to D$ in $\Diag(\cat{C})$ is relatively initial. Then the pushforward by $(R,\rho)$ of cones over $D$ to cones of $D'$, defined by \cref{eq:diagram-hom-cones}, sends limits of $D$ to limits of $D'$.
\end{proposition}
\begin{proof}
First, we show that a relatively initial morphism $(R,\rho): (\cat{J}, D) \to (\cat{J}',D')$ in $\Diag(\cat{C})$ pushes forward cones over $D$ to cones over $D'$. Since cones over a diagram are the same as lifts of the diagram through a generalized element projection, this is a special case of \cref{thm:diagram-initial-opfibration}, but we give a self-contained proof to make the construction explicit.

Suppose $\lambda$ is a cone over $D$ with apex $x \in \cat{C}$. Define a cone $\lambda'$ over $D'$ with apex $x$ by choosing, for each $j' \in \cat{J}'$, an object $(j,f') \in (R,\rho)/j'$ and setting
\begin{equation*}
    \lambda_{j'}' := (x \xrightarrow{\lambda_j} Dj \xrightarrow{\rho_j} D'Rj \xrightarrow{D'f'} D'j').
\end{equation*}
To prove that $\lambda'$ is a cone, consider a morphism $h': j' \to k'$ in $\cat{J}'$ and let $(j,f')$ and $(k,g')$ be the objects of $(R,\rho)/j'$ and $(R,\rho)/k'$ defining the legs $\lambda_{j'}'$ and $\lambda_{k'}'$ of $\lambda'$. By assumption, the category $(R,\rho)/k'$ is connected, so we can find a zig-zag of morphisms in $(R,\rho)/k'$ connecting $(j, f' \cdot h')$ and $(k,g')$, which yields a commutative diagram:
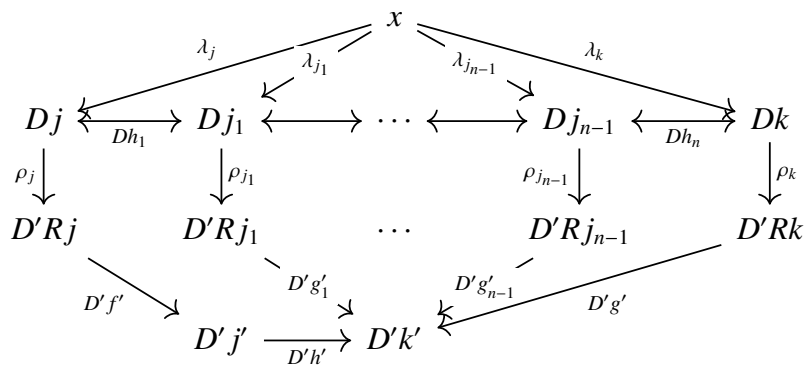
\begin{equation*}
    \begin{tikzcd}
    	&& x \\
    	Dj & {D j_1} & \cdots & {D j_{n-1}} & Dk \\
    	{D'Rj} & {D'R j_1} & \cdots & {D'R j_{n-1}} & {D'R k} \\
    	& {D'j'} & {D'k'}
    	\arrow["{\lambda_j}"', from=1-3, to=2-1]
    	\arrow["{\lambda_{j_1}}"{description}, from=1-3, to=2-2]
    	\arrow["{D h_1}"', tail reversed, from=2-1, to=2-2]
    	\arrow["{\rho_j}"', from=2-1, to=3-1]
    	\arrow["{D'f'}"', from=3-1, to=4-2]
    	\arrow["{D'h'}"', from=4-2, to=4-3]
    	\arrow["{\rho_{j_1}}", from=2-2, to=3-2]
    	\arrow["{D' g_1'}"{description}, from=3-2, to=4-3]
    	\arrow["{\lambda_{j_{n-1}}}"{description}, from=1-3, to=2-4]
    	\arrow["{\rho_{j_{n-1}}}"', from=2-4, to=3-4]
    	\arrow["{\lambda_k}", from=1-3, to=2-5]
    	\arrow["{\rho_k}", from=2-5, to=3-5]
    	\arrow["{D' g_{n-1}'}"{description}, from=3-4, to=4-3]
    	\arrow["{D'g'}", from=3-5, to=4-3]
    	\arrow["{D h_n}"', tail reversed, from=2-4, to=2-5]
    	\arrow[tail reversed, from=2-2, to=2-3]
    	\arrow[tail reversed, from=2-3, to=2-4]
    \end{tikzcd}
\end{equation*}
Thus, using the definition of $\lambda'$, we obtain the required equation $\lambda_{j'}' \cdot D'h' = \lambda_{k'}'$.

As for the proposition itself, now suppose $(R,\rho): (\cat{J},D) \to (\cat{J}',D')$ is a strong morphism in $\DiagOp(\cat{C})$ such that $(R,\rho^{-1})$ in $\Diag(\cat{C})$ is relatively initial. In one direction, given a cone $\lambda$ over $D$ with apex $x$, a cone $\lambda'$ over $D'$ with the same apex $x$ is defined by
\[\lambda'_{j'} := (x \xrightarrow{\lambda_{Rj'}} DRj' \xrightarrow{\rho_{j'}} D'j').\]
In the other direction, given a cone $\lambda'$ over $D'$ with apex $x$, the above construction yields a cone $\lambda$ over $D$ with legs
\[\lambda_j := (x \xrightarrow{\lambda'_{j'}} D'j' \xrightarrow{\rho_{j'}^{-1}} DRj' \xrightarrow{Df} Dj)\]
for any choice of $(j',f) \in (R,\rho^{-1})/j$. Moreover, when $j$ is in the image of $R$, we can take $j'$ with $Rj' = j$ and $f = \id_j$, so that
\[\lambda_j = (x \xrightarrow{\lambda'_{j'}} D'j' \xrightarrow{\rho_{j'}^{-1}} DRj' = Dj).\]
It is then straightforward to show that these mappings define a bijection between cones over $D$ and $D'$ and in fact an isomorphism between the categories of cones over $D$ and $D'$. Since limits are terminal objects in the category of cones, they are preserved by the mappings.
\end{proof}

Finally, for completeness, we record the following closure properties.

\begin{lemma}[Closure of relatively initial maps] \label{lem:closure-relatively-initial-map}
\leavevmode
\begin{enumerate}[(i),nosep]
\item Composites of relatively initial morphisms of diagrams are relatively initial.
\item Isomorphisms of diagrams are relatively initial.
\end{enumerate}
In particular, for any category $\cat{C}$, the relatively initial morphisms of diagrams in $\cat{C}$ form a wide subcategory of $\Diag(\cat{C})$.
\end{lemma}
\begin{proof}
Part (ii) follows from \cref{lem:closure-initial-functor} (ii), so we need only prove part (i).

Suppose $(\cat{I}, D) \xrightarrow{(R,\rho)} (\cat{J}, E) \xrightarrow{(S,\sigma)} (\cat{K}, F)$ is a composable sequence of relatively initial morphisms in $\Diag(\cat{C})$. We must show that for every $k \in \cat{K}$, the category $(R,\rho) \cdot (S,\sigma) / k$ is nonempty and connected. It is clearly nonempty, since we can choose an object $(j, Sj \xrightarrow{h} k)$ in $(S,\sigma)/k$ and then choose an object $(i, Ri \xrightarrow{g} j)$ in $(R,\rho)/j$ and so obtain an object
\begin{equation*}
    (i,\ SRi \xrightarrow{Sg} Sj \xrightarrow{h} k)
    \quad\text{in}\quad
    (R,\rho) \cdot (S,\sigma) / k.
\end{equation*}
As for connectnedness, consider two objects $(i, SRi \xrightarrow{h} k)$ and $(i', SRi' \xrightarrow{h'} k)$ in the category $(R,\rho) \cdot (S,\sigma) / k$. Since $(S,\sigma)/k$ is connected, we can find a zig-zag of morphisms connecting $(Ri, h)$ and $(Ri', h')$ in $(S,\sigma)/k$. Without loss of generality, assume that this zig-zag consists of a single morphism $g: Ri \to Ri'$ in $\cat{J}$. Now, since $(R,\rho) / Ri'$ is connected, we can find a zig-zag of morphisms connecting $(i, g)$ and $(i', \id_{Ri'})$ in $(R,\rho)/Ri'$, yielding a commutative diagram:
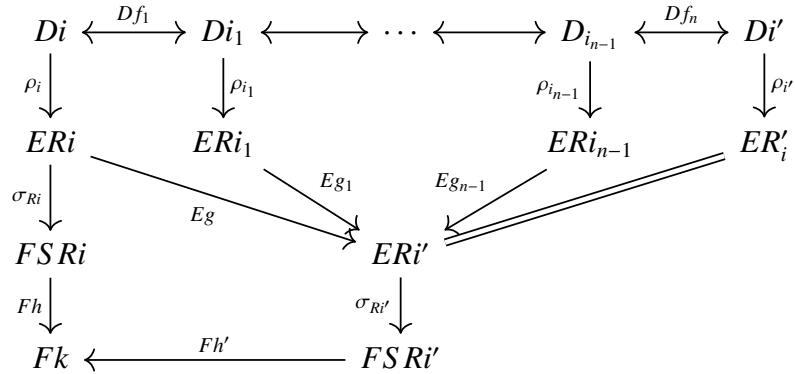
\begin{equation*}
    \begin{tikzcd}
    	Di & {D i_1} & \cdots & {D_{i_{n-1}}} & {Di'} \\
    	ERi & {ER i_1} && {ER i_{n-1}} & {ER_i'} \\
    	FSRi && {ERi'} \\
    	Fk && {FSRi'}
    	\arrow["Eg"', from=2-1, to=3-3]
    	\arrow["{Df_1}", tail reversed, from=1-1, to=1-2]
    	\arrow["{Df_n}", tail reversed, from=1-4, to=1-5]
    	\arrow["{\rho_{i_1}}", from=1-2, to=2-2]
    	\arrow["{Eg_1}", from=2-2, to=3-3]
    	\arrow["{E g_{n-1}}"', from=2-4, to=3-3]
    	\arrow[equal, from=2-5, to=3-3]
    	\arrow["{\rho_{i'}}", from=1-5, to=2-5]
    	\arrow["{\sigma_{Ri}}"', from=2-1, to=3-1]
    	\arrow["{\rho_{i_{n-1}}}"', from=1-4, to=2-4]
    	\arrow["{\rho_i}"', from=1-1, to=2-1]
    	\arrow["{\sigma_{Ri'}}"', from=3-3, to=4-3]
    	\arrow["Fh"', from=3-1, to=4-1]
    	\arrow["{Fh'}"', from=4-3, to=4-1]
    	\arrow[tail reversed, from=1-2, to=1-3]
    	\arrow[tail reversed, from=1-4, to=1-3]
    \end{tikzcd}.
\end{equation*}
Applying the naturality of the transformation $\sigma$ to the morphisms $g =: g_0, g_1, \dots, g_{n-1}$, an equivalent commutative diagram is:
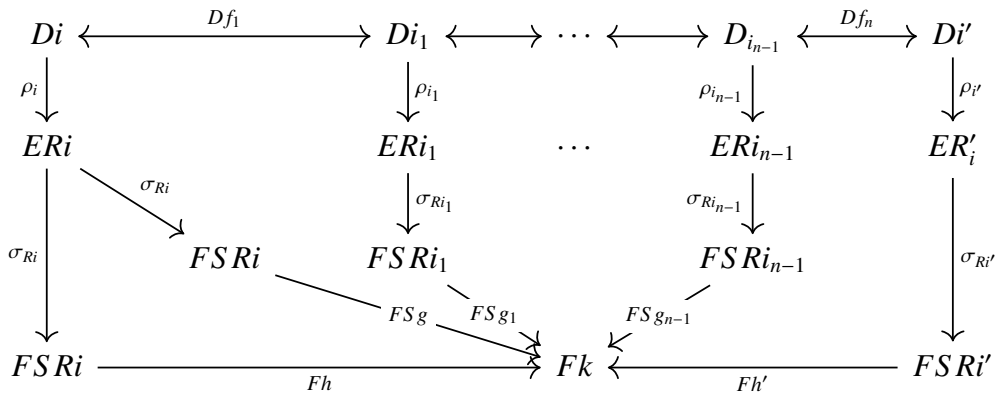
\begin{equation*}
    \begin{tikzcd}
    	Di && {D i_1} & \cdots & {D_{i_{n-1}}} & {Di'} \\
    	ERi && {ER i_1} & \cdots & {ER i_{n-1}} & {ER_i'} \\
    	& FSRi & {FSR i_1} && {FSR i_{n-1}} \\
    	FSRi &&& Fk && {FSRi'}
    	\arrow["{D f_1}", tail reversed, from=1-1, to=1-3]
    	\arrow["{D f_n}", tail reversed, from=1-5, to=1-6]
    	\arrow["{\rho_{i_1}}", from=1-3, to=2-3]
    	\arrow["{\rho_{i'}}", from=1-6, to=2-6]
    	\arrow["{\rho_{i_{n-1}}}"', from=1-5, to=2-5]
    	\arrow["{\rho_i}"', from=1-1, to=2-1]
    	\arrow["Fh"', from=4-1, to=4-4]
    	\arrow["{Fh'}", from=4-6, to=4-4]
    	\arrow[tail reversed, from=1-3, to=1-4]
    	\arrow[tail reversed, from=1-5, to=1-4]
    	\arrow["{\sigma_{R i_1}}", from=2-3, to=3-3]
    	\arrow["{FSg_1}"{description}, from=3-3, to=4-4]
    	\arrow["{\sigma_{Ri'}}", from=2-6, to=4-6]
    	\arrow["{\sigma_{R i_{n-1}}}"', from=2-5, to=3-5]
    	\arrow["{FS g_{n-1}}"{description}, from=3-5, to=4-4]
    	\arrow["{\sigma_{Ri}}", from=2-1, to=3-2]
    	\arrow["FSg"{description}, from=3-2, to=4-4]
    	\arrow["{\sigma_{Ri}}"', from=2-1, to=4-1]
    \end{tikzcd}.
\end{equation*}
Recognizing the components of the transformation $\rho \circ (\sigma * R)$, we obtain a zig-zag connecting $(i,h)$ and $(i',h')$ in $(R,\rho) \cdot (S,\sigma) / k$.
\end{proof}

\section{Conclusions and future work} \label{sec:conclusion}

In this paper, we have undertaken to systematize, formalize, and generalize the diagrams employed by Tonti and others to present differential equations in physics. The mathematical foundation is provided by category-theoretic diagrams and morphisms between them, as well as extensions to diagrams involving cartesian products or tensor products. A wide variety of physical systems, from electromagnetism to transport phenomena to fluid mechanics, have been shown to fit into this framework. Morphisms of diagrams, not considered in the physics literature, have been used to formalize relationships between different systems, as well as to modularly construct complex, possibly multiphysical, systems from basic physical principles and experimental laws. On the theoretical side, we have explained how solving systems of equations or boundary value problems amounts to solving lifting or extension-lifting problems of diagrams, and we have studied in some detail how morphisms from both categories of diagrams interact with lifts of diagrams (\cref{thm:diagram-opfibration,thm:diagram-initial-opfibration}).

It should be noted that Tonti's attempted classification of physical theories, besides its distinctive use of diagrams, involves a rich ontology of concepts in physics. For example, Tonti classifies physical variables into source, configuration, and energy variables \cite[Chapter 5]{tonti2013} and physical equations into defining, topological, behavior, and phenomenological equations, each with further subclassifications~\cite[Chapter 6]{tonti2013}. No attempt has been made here to formalize these distinctions, and it is an interesting question whether or not this is possible.

Our proposed notion of a weak equivalence of two diagrams, which establishes a one-to-one correspondence between lifts of the diagrams, merits further study. We introduced a generalization of an initial functor, called a relatively initial morphism of diagrams, and showed that it leads to a sufficient condition for weak equivalence, but we have left open whether this is also a necessary condition, or whether our notion of weak equivalence even admits any characterization in terms of concrete, easily verified conditions. Moreover, the term ``weak equivalence'' is suggestive of methods from homotopy theory such as localization and calculus of fractions \cite{gabriel1967}. Preliminary investigations suggest that homotopical methods may be applicable to initial or relatively initial morphisms of diagrams. This would aid in understanding quotient categories of diagrams where weak equivalences are formally inverted, a construction not pursued here.

To sidestep the complexities of weak equivalence, one might seek ways of formalizing physical theories or systems of equations that are less sensitive than diagrams to differences of presentation. The category-theoretic approach to logic, pioneered by Lawvere \cite{lawvere1963}, reconstructs logical theories as algebraic objects, making them presentation invariant. It would be worthwhile to understand how physical theories formalized using either category-theoretic diagrams or categorical logic are related.

Finally, although this work has its origins in the partial differential equations of continuum physics, with most examples drawn from that subject, we emphasize that its foundation in category-theoretic diagrams is far more general. In \cref{ex:discrete-heat-equation,ex:discrete-dirichlet-problem}, difference equations were presented using diagrams in $\Vect_\R$, requiring only linear algebra, rather than differential geometry. In future work, we intend to explore how probabilistic and statistical models, such as structural equation models, can be presented by diagrams in suitable categories. We hope and expect that many further applications will be found.

\section*{Acknowledgments}

The authors thank David Jaz Myers for first noticing the important role of initial functors in this project, David Spivak for further discussion on that point, John Baez for helpful conversations about the differential-geometric setting and the connection between lifts and limits, Owen Lynch for helpful discussions about category actions defining lifting problems, and Jesus Arias, Maia Gatlin, and Clayton Kerce for working with us to apply these ideas to computational physics and multiphysics simulation. Authors Baas, Fairbanks, and Patterson acknowledge support from the DARPA Computable Models (CompMods) program through Award HR00112090067; Fairbanks also acknowledges support from the DARPA Young Faculty Award Program through Award W911NF2110323.

\section*{Conflict of interest}

The authors declare no conflict of interest.

\newpage
\bibliographystyle{aims}
\bibliography{refs}

\end{document}